\documentclass[11pt]{article}

\usepackage{amsmath,amsthm,verbatim,amssymb,amsfonts,amscd, graphicx}
\usepackage{hyperref}
\usepackage{graphics}
\usepackage{bm}
\usepackage{algorithm}
\usepackage{algorithmicx}
\usepackage{algpseudocode}
\usepackage{color}
\usepackage{multirow}
\usepackage{relsize}
\usepackage{enumerate}
\usepackage{enumitem}
\usepackage{thm-restate}

\theoremstyle{plain}
\newtheorem{theorem}{Theorem}
\newtheorem{corollary}{Corollary}
\newtheorem{lemma}{Lemma}
\newtheorem{assumption}{Assumption}

\theoremstyle{plain}
\newtheorem{definition}{Definition}
\newtheorem{oracle}{Oracle}

\newcommand\CONDITION[2]%
{\begin{tabular}[t]{@{}l@{}l@{}}
		#1&#2
	\end{tabular}%
}
\algdef{SE}[WHILE]{While}{EndWhile}[1]%
{\algorithmicwhile\ \CONDITION{#1}{\ \algorithmicdo}}%
{\algorithmicend\ \algorithmicwhile}
\algdef{SE}[FOR]{For}{EndFor}[1]%
{\algorithmicfor\ \CONDITION{#1}{\ \algorithmicdo}}%
{\algorithmicend\ \algorithmicfor}
\algdef{S}[FOR]{ForAll}[1]%
{\algorithmicforall\ \CONDITION{#1}{\ \algorithmicdo}}
\algdef{SE}[REPEAT]{Repeat}{Until}{\algorithmicrepeat}[1]%
{\algorithmicuntil\ \CONDITION{#1}{}}
\algdef{SE}[IF]{If}{EndIf}[1]%
{\algorithmicif\ \CONDITION{#1}{\ \algorithmicthen}}%
{\algorithmicend\ \algorithmicif}%
\algdef{C}[IF]{IF}{ElsIf}[1]%
{\algorithmicelse\ \algorithmicif\ \CONDITION{#1}{\ \algorithmicthen}}

\newcommand{\R}{\mathbb{R}}

\DeclareMathOperator*{\argmax}{arg\,max}

\newcommand{\diag}{\operatorname{\mathbf{diag}}}
\newcommand{\mdiag}{\diag}

\newcommand{\bone}{\mathbf{1}}
\newcommand{\eye}{\mI}

\newcommand{\nnz}{\mathrm{nnz}}

\newcommand{\eq}{{\ \leftarrow\ }}

\newcommand{\mdecide}{$\code{MMatrix-Decide}$}
\newcommand{\findperronvalue}{$\code{Find-Perron-Value}$}

\newcommand{\boldVar}[1]{\mathbf{#1}}
\newcommand{\mvar}[1]{\boldVar{#1}}
\newcommand{\vvar}[1]{\vec{#1}}

\newcommand{\defeq}{\stackrel{\mathrm{{\scriptscriptstyle def}}}{=}}
\newcommand{\eqdef}{\defeq}

\newcommand{\otilde}{\tilde{O}}

\newcommand{\norm}[1]{\|#1\|}

\newcommand{\va}{\vvar a}
 \newcommand{\vb}{\vvar b}
 \newcommand{\vd}{\vvar d}
 \newcommand{\ve}{\vvar e}

 \newcommand{\vl}{\vvar l}
 \newcommand{\vell}{\vvar \ell}

 \newcommand{\vp}{\vvar p}
 
 \newcommand{\vr}{\vvar r}

 \newcommand{\vv}{\vvar v}
 \newcommand{\vw}{\vvar w}
 \newcommand{\vx}{\vvar x}
 \newcommand{\vy}{\vvar y}
 \newcommand{\vz}{\vvar z}

 \newcommand{\code}{\mathtt }

 \newcommand{\vzero}{\vvar 0}
 \newcommand{\vones}{\vvar 1}

\newcommand{\ma}{\mvar A}
 \newcommand{\mb}{\mvar B}
 \newcommand{\mc}{\mvar C}
 \newcommand{\md}{\mvar D}

\newcommand{\mI}{\mvar I}

 \newcommand{\ml}{\mvar L}
\newcommand{\mm}{\mvar M}
 \newcommand{\mn}{\mvar N}
 \newcommand{\mP}{\mvar P}
 
 \newcommand{\mr}{\mvar R}
 \newcommand{\ms}{\mvar S}

 \newcommand{\mv}{\mvar V}
 \newcommand{\mw}{\mvar W}
 \newcommand{\mx}{\mvar X}
 \newcommand{\my}{\mvar Y}
\newcommand{\mz}{\mvar Z}
\newcommand{\mzero}{\mvar 0}

\newcommand{\mlap}{\mvar{\mathcal{L}}}

\newcommand{\Pright}{\mP^{(r)}}
\newcommand{\Pleft}{\mP^{(l)}}
\newcommand{\Prightk}{\mP^{(r, k)}}
\newcommand{\Pleftk}{\mP^{(l, k)}}

\newcommand{\nn}{{n \times n}}

\newcommand{\lmax}{\lambda_{\mathrm{max}}}
\newcommand{\lmin}{\lambda_{\mathrm{min}}}

\newcommand{\poly}{\mathrm{poly}}

\newcommand{\opZ}{\mz}

\newcommand{\Exp}{\mathlarger{\mathbb{E}} \!}
\newcommand{\eigl}{\vv_l}
\newcommand{\eigr}{\vv_r}

\newcommand{\trans}{\top}
\newcommand{\invtrans}{{- \top}}
\newcommand{\pseudo}{+}

\renewcommand{\Pright}{\mP^{(r)}}
\renewcommand{\Pleft}{\mP^{(l)}}
\renewcommand{\Prightk}{\mP^{(l, k)}}
\renewcommand{\Pleftk}{\mP^{(r, k)}}
\renewcommand{\nn}{n \times n}
\renewcommand{\opZ}{\mz}

\newcommand{\fullsymmscale}{$\code{SymMMatrix-Scale}$}
\newcommand{\SDDsolve}{$\code{SDDSolve}$}

\begin{document}

\title{Perron-Frobenius Theory in Nearly Linear Time: \\
\normalsize Positive Eigenvectors, M-matrices, Graph Kernels, and Other Applications
}
\author{
AmirMahdi Ahmadinejad\\
Stanford University\\
ahmadi@stanford.edu\and 
Arun Jambulapati\\
Stanford University\\
jmblpati@stanford.edu\and 
Amin Saberi\\
Stanford University\\
saberi@stanford.edu\and
Aaron Sidford\\
Stanford University\\
sidford@stanford.edu
}
\date{}
\maketitle
\thispagestyle{empty}

\begin{abstract}
In this paper we provide nearly linear time algorithms for several problems closely associated with the classic Perron-Frobenius theorem, including computing Perron vectors, i.e. entrywise non-negative eigenvectors of non-negative matrices, and solving linear systems in asymmetric M-matrices, a generalization of Laplacian systems. The running times of our algorithms  depend nearly linearly on the input size and   polylogarithmically on the desired accuracy and problem condition number. 

Leveraging these results we also provide improved running times for a broader range of problems including computing random walk-based graph kernels, computing Katz centrality, and more. The running times of our algorithms improve upon previously known results which either depended polynomially on the condition number of the problem, required quadratic time, or only applied to special cases.

We obtain these results by providing new iterative methods for reducing these problems to solving linear systems in Row-Column Diagonally Dominant (RCDD) matrices. Our methods are related to the classic shift-and-invert preconditioning technique for eigenvector computation and constitute the first alternative to the result in Cohen et al. (2016) for reducing stationary distribution computation and solving directed Laplacian systems to solving RCDD systems.

\end{abstract}
\newpage
\setcounter{page}{1}

\section{Introduction}


The \emph{Perron-Frobenius theorem} of O. Perron (1907) and G. Frobenius (1912) is a fundamental result in linear algebra that has had far reaching implications over the past century. Notable applications include the study of Markov chains in probability theory, the theory of dynamical systems, economic analysis (such as Leontief's input-output model), and modeling and analysis of social networks.

In its simplest form, the Perron-Frobenius theorem states that every  positive square matrix $\ma \in \R_{>0}^{n\times n}$ has a unique largest real eigenvalue $\lambda > 0$ such that $\ma \vv = \lambda \vv$ for some positive vector $\vv \in \R_{>0}^{n}$. This vector, known as the \emph{Perron-Frobenius vector} or \emph{Perron vector}, plays a central role in understanding positive matrices and a broad classes of linear systems known as M-matrices. Through the natural representation of directed graphs with non-negative edge weights via their adjacency matrices, the Perron vector provides a critical tool to understand numerous graph related properties.

Given its far-reaching applications and its role as a cornerstone to numerous disciplines, providing faster algorithms for computing Perron vectors and positive eigenvectors of non-negative matrices has been a central problem studied extensively for decades \cite{mises1929praktische,wielandt1944beitrage,parlett1974rayleigh,ipsen1997computing, francis1961qr,kublanovskaya1962some,cuppen1980divide,householder2013theory,parlett1998symmetric}. However, despite extensive research on this problem and numerous proposed algorithms (see  \cite{golub2001eigenvalue} for a survey), all previously known algorithms either run in quadratic time, are applicable only in special cases \cite{DaitchS09,coh16,coh17}, or depend more than polylogarithmically on the desired accuracy and  the relevant condition numbers (see Section~\ref{sub:previous_work} for further discussion).

In this paper, we provide the first nearly linear time algorithm\footnote{We use the term ``nearly linear time" to refer to any running time that is $O(z^{1 + o(1)})$ where $z$ is a measure of the size of the input. This is in contrast to certain works \cite{coh17,spielman2014} which use the phrase ``almost linear" to refer to this running time and reserve the phrase ``nearly linear" for running times which hide only polylogarithmic factors. We adopt this notation as our running time depend on the best running time for solving row-column diagonally dominant systems \cite{coh16}. While the state of the art for those is currently almost linear, their (expected) improvement to nearly linear time will result in the same running time for us too.} for computing Perron vectors and prove the following theorem.

\begin{restatable}[Perron Vector Nearly Linear Time]{theorem}{perron}
\label{thm:perron}
Given non-negative irreducible matrix $\ma\in \R_{\geq 0}^{n\times n}$ and $\delta>0$, there is an Algorithm (See Algorithm~\ref{alg:perron}) which in $\otilde(m)$ time with high probability in $n$ computes real number $s>0$, and positive vectors $\vl,\vr,\in \R_{>0}^{n}$ such that $(1-\delta)\rho(\ma) < s \leq \rho(\ma)$,
\[
\|(s \eye - \mm) \vr\|_\infty \leq \frac{\delta}{K^2}\norm{\vr}_\infty ,
\text{ and }
\|\vl^\top (s \eye - \mm)\|_\infty \leq \frac{\delta}{K^2}\norm{\vl}_\infty
\]
where $K = \Theta(\kappa(\eigl)+\kappa(\eigr))$ and $\eigl$, $\eigr$ are the left and right Perron vectors of $\ma$.
\end{restatable}

To achieve this result we consider the more general problem of solving linear systems for a broad class of matrices, known as M-matrices. A matrix $\mm$ is an M-matrix if it can be written of the form $\mm = s \mI - \mb$ where $s > 0$ and $\mb$ is a non-negative matrix with spectral radius (i.e. largest eigenvalue) of magnitude at most $s$. These systems generalize symmetric diagonally dominant systems \cite{spielman04}, row-column diagonally dominant systems \cite{coh16}, and factor width 2 systems \cite{DaitchS09} and have numerous applications.

The previously known algorithms for solving linear systems in M-matrices suffer from the same issues as we mentioned for computing the Perron vector. In this paper, we present the first nearly linear time algorithm for solving this problem and prove the following theorem:

\begin{restatable}[M-Matrices in Nearly Linear Time]{theorem}{msolve}
    \label{thm:msolve}
    Let $\ma \in \R^{\nn}_{\geq 0}$ with $m$ nonzero entries, $s > 0$, $\rho(\ma) < s$, and $\mm = s \eye - \ma$. For all $\epsilon > 0$ and $K \geq \max \{s \|\mm^{-1}\|_\infty, s \|\mm^{-1}\|_1\}$ there is an algorithm (See Algorithm~\ref{alg:SolveM}) which runs in $\otilde(m)$ time and with high probability computes an operator $P$ where for any vector $\vb$ it is the case $\Exp \|\vb - \mm P(\vb) \|_2 \leq \epsilon \| \vb \|_2$.
\end{restatable}

As an immediate corollary of this theorem, we can show that if a matrix $\mm \in \R^{n \times n}$ has the property that after negating its off-diagonal entries it becomes an M-matrix, then we can also solve systems in $\mm$ in nearly linear time (See Section \ref{sec:factor2}). We provide a specialization of this result for symmetric matrices that shows we can solve factor-width 2 matrices, that is matrices of the form $\mm = \mc^\top \mc$ where each row of $\mc$ has at most $2$ non-zero entries (See Section~\ref{sec:factor2}). Further in Section~\ref{sec:symmscaling} we prove a specialization of Theorem~\ref{thm:msolve} for the case where $\mm$ is symmetric. In this case we are able to get a shorter analysis with a tighter running time bound.

To prove Theorem~\ref{thm:msolve} we build upon and enhance an exciting  line of work on developing nearly linear time algorithms for computing the stationary distribution of Markov chains, a special case of the Perron vector problem, and solving linear systems in directed Laplacians, a special cases of the problem of solving M-matrices \cite{coh16, coh17}.  This line of work achieved these results in two steps. First \cite{coh16} provided a reduction from these problems to solving linear systems in Eulerian Laplacians or more broadly, row-column diagonally dominant (RCDD) systems. Second, \cite{coh17} provided a nearly linear time algorithm for solving these systems.

These results suggest a natural line of attack to achieve our results, extend the reduction from \cite{coh16} to reduce solving M-matrices to solving RCDD matrices. While there are reasons to hope this might be possible (See Section~\ref{sub:intro_tech_motiv}), there are two significant barriers to this approach (See Section~\ref{sub:previous_work}). 
First, the analysis in \cite{coh16} leveraged the fact that directed Laplacians are always M-matrices, and there is a simple linear time procedure to check that a matrix is a directed Laplacian (and therefore an M-matrix). However, in general the mere problem of computing whether or not a matrix is an M-matrix involves computing the spectral radius or top eigenvector of a non-negative matrix and prior to our work there was no nearly linear time algorithm for it. Consequently, extending \cite{coh16} would require a different type of analysis that leverages properties of M-matrices that cannot be easily tested. Second, while the algorithm and analysis of the reduction in \cite{coh16} is fairly simple and short, it does not follow any previously known analysis of an iterative method and thus it does not seem clear how to modify the algorithm and analysis in a principled manner. 

To circumvent these issues, we provide a new algorithm for reducing solving M-matrices to solving RCDD systems that provides the first known alternative to the reduction in \cite{coh16} for solving directed Laplacians. Interestingly, our algorithm is a close relative of a classic and natural method for computing eigenvectors of a matrix known as \emph{shift-and-invert preconditiong} \cite{saad1992numerical, GarberHJKMNS16} and can be easily analyzed by a simple bound on how well different regularized M-matrices approximate or \emph{precondition} each other. While, the careful error analysis of our method is detailed (see Section~\ref{sec:scaling}) its core fits on just two pages and can be found in Section~\ref{sec:overview}. 
We believe this is one of the primary technical contributions of our paper. (See Section~\ref{sub:approach} and Section~\ref{sec:overview} for further discussion.)

Finally, we show how to leverage these results to achieve faster running times for a wide range of problems involving Perron vectors, M-matrices, and random walks.  In particular, we show our method can be applied in the analysis of Leontief input-output models, 
computing Katz Centrality~\cite{katz1953new}, computing the top singular value and its corresponding vectors of a non-negative matrix, and computing random walk based graph kernels. 
While some of these are direct applications, others require opening up our framework and leveraging its flexibility, e.g computing the top singular value and its corresponding vectors of a non-negative matrix.  (See Section~\ref{sec:applications})

Just as Laplacian systems have been crucial for solving a broad range of combinatorial problems \cite{kelner2014almost,madry2010fast,orecchia2012approximating,spielman2013local,spielman2014}, we hope that the work in this paper will ultimately open the door to the development of even faster and simpler linear algebraic primitives for processing graphs and probing their random walk structure and find even broader applications than those provided in this paper.

\subsection{M-Matrices, Positive Matrix Sums, and Perron Vectors}
\label{sub:intro_tech_motiv}

Our results leverage a tight connection between invertible M-matrices, matrix series, random walks, Perron vectors, and RCDD matrices. Here we elaborate on this connection to motivate previous work on these problems and our approach to solving them. 

To demonstrate this connection, first observe that a matrix $\mm \in \R^{n \times n}$ is an invertible M-matrix if and only if it is of the form $\mm = s \mI - \mb$ where $\rho(\mb) < s$. In this case it is easy to see that 
\[
(s \mI - \mb)^{-1} = \frac{1}{s} \left(\mI - \frac{1}{s} \mb\right)^{-1}
= \frac{1}{s} \sum_{i = 0}^{\infty} \left(\frac{1}{s} \mb\right)^{i} ~.
\]
Consequently, the question of testing whether a matrix is an M-matrix and solving linear systems in it is closely related to the following problems:

\begin{definition}[M-Matrix Problem]
Given matrix $\mm$ and a vector $\vb$ either solve the linear system $\mm \vx = \vb$ or prove that $\mm$ is not an M-matrix.
\end{definition}

\begin{definition}[Geometric Non-negative Matrix Series Problem] 
Given a non-negative matrix $\mb \in \R^{n \times n}_{\geq 0}$ and a vector $b$ compute $\sum_{i = 0}^{\infty} \mb^i b$ or prove it diverges.
\end{definition}

\begin{definition}[Perron Vector Problem] \label{def:perron}
Given a non-negative  matrix $\mb \in \R^{n \times n}_{\geq 0}$ computes $\rho(\mb)$, its largest eigenvector, and a Perron vector, $\vv_r \in \R^{n}_{\geq 0}$ such that $\mb \vv_r = \rho(\mb) \vv_r$. 
\end{definition}

We provide the first nearly linear time algorithms for the above three problems. This is done by leveraging a somewhat classic fact that Perron vector problem is also intimately connected to the M-matrix problem. If $\mb \in \R^{n \times n}$ is a non-negative matrix then,
\[
\lim_{s \rightarrow \rho(\mb)^+} (s \mI- \mb)^{-1} \vones 
= 
\lim_{s \rightarrow \rho(\mb)^+} \frac{1}{s} \sum_{i = 0}^{\infty} \left(\frac{1}{s} \mb \right)^{i} \vones 
= \vv_r
\]
where $\lim_{s \rightarrow \rho(\mb)^+}$ is the limit of $s$ approach $\rho(\mb)$ from above and $\vv_r$ is a non-negative Perron vector of $\mb$. Consequently, given a solution of the M-matrix problem we can obtain arbitrary good approximations to the Perron vector problem. 

Second, (and perhaps more surprising), recent work on solving RCDD systems \cite{coh16, coh17} implies that solutions to the Perron vector problem lead to nearly linear time algorithms for the M-Matrix Problem. Suppose we have an invertible M-matrix with $\mm = \mI - \mb$ and we have computed positive left and right Perron vectors $\vl \in \R^{n}_{> 0}$ and $\vr \in \R^{n}_{> 0}$ respectively, then it is easy to see (See Section~\ref{sec:overview}) that $\mv_l \mm \mv_r$ is RCDD, where $\mv_l$ and $\mv_r$ are the diagonal matrices associated with $\vv_l$ and $\vv_r$ respectively. Consequently, given Perron vectors and RCDD solvers we can solve M-matrices.

These two results together provide chicken and egg problem for solving M-matrices. We can solve M-matrices with Perron vectors and compute Perron vectors with M-matrix solvers, however it is unclear how to achieve either individually. While there have been previous results for escaping this conundrum to provide solvers for special cases of M-matrices, they are only applicable in settings where we were able to certify that the input matrix was an M-matrix in nearly linear time. A key result of this paper is that a variant of shift-and-invert preconditioning carefully analyzed leveraging properties of M-matrices lets us efficiently escape this problem and obtain nearly linear time algorithms.

\subsection{Previous Work}
\label{sub:previous_work}

Here we briefly review previous work on solving M-matrices and computing Perron vectors. These problems have been studied extensively and there are too many results to list. Here we briefly survey various areas of research relevant to our theoretical improvements. For a more comprehensive treatment of these problems see  \cite{golub2001eigenvalue}.  Previous work on the additional applications of these results are deferred to Section~\ref{sec:applications}.

Each of the problems described in Section~\ref{sub:intro_tech_motiv} can easily be solved using an algorithm for solving square linear systems. Moreover, given a matrix $\ma \in \R^{n \times n}$ and a vector $\vb \in \R^n$, it is well-known that a linear system $\ma \vx = \vb$ can be solved in $O(n^\omega)$, where $\omega < 2.373$ \cite{Williams12}, or can be solved in $O(m n)$ where $m = \nnz(\ma)$ is the number of nonzero entries of $\ma$.\footnote{Note that conjugate gradient achieves $O(m n)$ time only using exact arithmetic. It is open how to achieve a comparable running time where arithmetic operations can only be performed on numbers with a polylogarithmic number of bits. See \cite{muscoms17} for further discussion.} However, in both theory and practice $\ma$ is often sparse, e.g. $\nnz(\ma) = O(n)$, in which case these algorithms run in quadratic, i.e. $\Omega(n^2)$ time, which may be prohibitively expensive for many purposes.

To improve upon the performance of these generic linear system solvers a wide array of iterative methods and practical heuristics have been proposed. However, even in the special case of computing the stationary distribution of Markov chains, until the work of \cite{coh16, coh17}, the previously best known running times either only applied in special cases or depended polynomially on the desired accuracy or conditioning of the problem.  

Recently, in an important special case of the problems considered in Section~\ref{sub:intro_tech_motiv}, namely computing the \emph{stationary distribution}, and solving \emph{directed Laplacians} (i.e. $\mm = \mI - \mb$ for these $\mb$) it was shown how to solve these problems in nearly linear time \cite{coh17}. This result was achieved by a result in \cite{coh17} showing that RCDD matrices can be solved in nearly linear time and a result in \cite{coh16}, reducing each problem to solving RCDD systems with only polylogarithmic overhead. 

These works created hope for achieving the results of this paper, e.g. by extending the reduction in \cite{coh16}, however there are significant issues to this approach. First, the analysis in \cite{coh16} leveraged that it is easy to certify that a matrix is a directed Laplacian. The method incrementally computed the associated stationary distribution for this matrix using that a directed Laplacian has a kernel. However, as discussed even certifying that a matrix is an M-matrix involves computing the spectral radius of a non-negative matrix, for which there were no previously known nearly linear time algorithm. Second, the reduction in \cite{coh16} was not known to be an instance of any previously studied iterative framework and therefore it is difficult to adapt it to new cases. 

The only other work directly relevant to the problems we consider was the exciting work of \cite{DaitchS09} showing that symmetric M-matrices and more broadly factor-width 2 matrices could be solved in nearly linear time, provided an explicit factorization is given, i.e. for the input matrix $\mm$ there was a known $\mc$ with at most $2$ non-zero entries per row such that $\mc^\top \mc = \mm$. Note that all symmetric M-matrices are factor-width $2$, and though \cite{DaitchS09} provided a nearly linear time algorithm which given $\mc$ could reduce solving M to solving symmetric diagonally dominant (SDD) matrices, i.e. symmetric RCDD matrices, our result is the first to actual compute this factorization in nearly linear time when it is not given.

\subsection{Overview of Approach}
\label{sub:approach}

We achieve the results of our paper by leveraging the RCDD solvers of \cite{coh17}. As discussed in Section~\ref{sub:intro_tech_motiv} this would suffice to solve M-matrices in nearly linear time provided we could computed Perron vectors in nearly linear time. However, as discussed in Section~\ref{sub:previous_work} there are barriers to applying the previous methods of \cite{coh16, DaitchS09} for computing Perron vectors in our more general setting.

To overcome this difficulty, we provide a new reduction from solving M-matrices to solving RCDD systems, that can serve as an alternative for the reductions in both \cite{coh16} and \cite{DaitchS09}.
Our method is loosely related to a method known as \emph{shift-and-invert preconditioning} for computing eigenvectors of matrices \cite{saad1992numerical, GarberHJKMNS16}. The crux of this method is that if a matrix $\ma$ has top eigenvalue $\lambda > 0$ with top eigenvector $\vv$, then the matrix $(c \mI - \ma)^{-1}$ for $c > \lambda$ also has top eigenvector $\vv$ with eigenvalue $1/(c - \lambda)$. Consequently, by performing power method on $(c \mI - \ma)^{-1}$, i.e. solving linear systems in $c \mI- \ma$ we can compute a top eigenvector of $\ma$. Shift-and-invert precondition does precisely this, leveraging that power method may converge faster for this matrix as bringing $c$ towards $\lambda$ may accentuate the gap between $\lambda$ and the next smallest eigenvector.

Now, as we have discussed computing the left and right Perron vectors for $\mb$ suffices to solve linear systems in the M-matrix $s \mI - \mb$. Furthermore, similar to shift-and-invert we know that if $s \rightarrow \rho(\mb)^+$ solving linear systems in this matrix would suffice to compute the left and right Perron vectors. 
Our method works by considering the matrices $\mm_\alpha \defeq \mm + \alpha \mI = (s + \alpha) \mI - \mb$ for $\alpha > 0$. By an very minor strengthening of a not as well known fact about M-matrices,  it can be shown that actually just solving two linear system in $\mm_\alpha$ suffice to get left and right scalings to make it RCDD. Consequently, if we have a solver for any $\mm_\alpha$ we can get the scalings for $\mm_\alpha$ and have a new solver for that matrix.

Now, we can always pick $\alpha$ so large that $\mm_{\alpha}$ is RCDD and therefore we can compute its scalings. However, this argument does not yet suffice to get a scaling for any other $\mm_{\alpha}$. To circumvent this we show that in fact for any $\alpha$ and $\alpha'$ that are multiplicatively close, the matrix $\mm_{\alpha'}^{-1} \mm_{\alpha}$ is close to the identity in an appropriate norm and therefore a solver for one can be used through preconditioning to yield a solver for the other. 

Combining these insights yields a natural algorithm: compute a scaling for $\mm_{ \alpha}$ for large $\alpha$, use this to solve linear systems in $\mm_{ \alpha}$ and thereby $\mm_{ \alpha / 2}$. Use this solver to get scaling for $\mm_{ \alpha / 2}$ and repeat. We provide the basic mathematical analysis for this assuming exact solvers in Section~\ref{sec:overview} and then show how to build upon this to solve our problem in the inexact case in Section~\ref{sec:scaling_proofs}. In Section~\ref{sec:perron-alg} we then show how to carefully binary search to solve the Perron problem and certify if $\mm$ matrices are M-matrices.

\subsection{Overview of Results}

In this paper we give the first nearly linear time algorithms for Perron Vector computation and solving linear systems in M-matrices. Our main results are the following:


\perron*

\msolve*

As discussed, we achieve these results by providing reductions from solving these problems to solving linear systems in RCDD matrices. Our reduction is the first alternative to the reduction in \cite{coh16} for reducing stationary distribution computation and solving directed Laplacians to solving RCDD matrices and the first alternative to the reduction in \cite{DaitchS09} from solving symmetric M-matrices given a factorization to solving symmetric Laplacians. Our main theorem corresponding to this reduction is the following:

\begin{restatable}[M-Matrix Scaling in Nearly Linear Time]{theorem}{mscaler}
    \label{thm:mscale}
    Let $\ma \in \R^{\nn}_{\geq 0}$ with $m$ nonzero entries. Let $s > 0$ and $\rho(\ma) < s$, and let $\mm = s\eye - \ma$. For all $\epsilon > 0$ and $K \geq \max (s \|\mm^{-1}\|_\infty, s \|\mm^{-1}\|_1)$ there is an algorithm (See Algorithm~\ref{alg:MMatrix}) which runs in $\otilde(m)$ time and with high probability computes a pair of diagonal matrices $(\ml, \mr)$ where $\ml((1+\epsilon)s\eye - \ma)\mr$ is RCDD with high probability.
    \end{restatable}

Leveraging these results we obtain faster algorithms for a host of problems (See Section~\ref{sec:applications} for a more detailed discussion). Key results from this section include

\begin{itemize} 
\item {\bf Leontief economies:} these are models capturing interdependencies between different sectors within a national economy and how the output of one sector can be used as input to another sector changes in the production of one will affect the others. We give a near linear time for checking the Hawkin-Simons condition~\cite{hawkins1949note} that guarantees the existence of a non-negative output vector that solves the equilibrium relation in which demand equals supply.
\item {\bf Katz centrality:} we give a nearly linear time algorithm for computing the Katz centrality which defines the relative influence of a node within a network as a linear function of the influence of its neighbors.  
\item {\bf Left and right top singular vectors:} we can also compute the top left-right singular values and associated top lef-right singular vectors of a non-negative matrix in nearly linear time. For this application we need additional ingredients solve linear systems in $\mm$ when $\mm$ is RCDD (see Oracle~\ref{oracle:RCDDSolve}). 
\item{\bf Graph kernels:} graph kernels capture the similarity between two graphs with applications in social networks, studying chemical compounds, comparison and function prediction of protein structures, and analysis of semantic structures in natural language processing \cite{kash2003, kash2004, vish2010}. 
One key obstacle in applying the existing algorithms for computing the kernel function between two graphs is their large running time~\cite{vish2010}. Using our methodology, we obtain improved running times for computing canonical kernel functions known as random walk kernels.
\end{itemize} 


\subsection{Paper Organization}

The rest of the paper is organized as follows. In Section~\ref{sec:prelim} we cover preliminaries including notation and facts about M-matrices and Perron vectors that we use throughout the paper. Then, in Section~\ref{sec:overview} we provide a semi-rigorous technical overview of our approach ignoring errors from approximate numerical computations to make the insights of our results clear. 

In Section~\ref{sec:scaling} we provide the algorithm for computing RCDD scalings of M-matrices and in Section~\ref{sec:perron-alg} we show how to use this method to compute Perron vectors and achieve efficient running times for solving and certifying M-matrices with unknown condition number. In Section~\ref{sec:applications} we provide our applications. 

Finally, in Appendix~\ref{sec:mmatrix_facts} we provide proofs of facts about M-matrices and Perron vectors that we use throughout the paper and in Section~\ref{sec:scaling_proofs} we provide proofs missing from Section~\ref{sec:scaling_proofs}.

\section{Preliminaries}
\label{sec:prelim}

In this section we provide our notation (Section~\ref{sub:prelim_notation}) as well as facts about M-matrices (Section~\ref{sub:prelim_mmatrices}), the Perron-Frobenius theorem (Section~\ref{sub:perron-frobenius}), and RCDD matrices (Section~\ref{sub:prelim_rcdd}) that we use throughout the paper.


\subsection{Notation}
\label{sub:prelim_notation}

\textbf{Variables}: We use bold to denote matrices and arrows to denote vectors. We use $\mI, \mzero \in \R^{n \times n}$ to denote the identity matrix and the all-zero matrix respectively. We use $\vzero, \vones \in \R^n$ to denote the all-zeros and all-ones vectors respectively. We use $\ve_i \in \R^n$ to denote the $i^{th}$ standard basis vector, i.e. the vector where $\ve_i(j) = 0$ for $j \neq i$ and $\ve_i(i) = 1$.\\
\\
\textbf{Matrix Operations}: For square matrix $\ma$ we use $\ma^\trans$ to denote its transpose, $\ma^{- 1}$ to denote its inverse, and $\ma^\pseudo$ to denote its Moore-Penrose pseuduinverse. For notational convenience we use  $\ma^\invtrans$ to denote the inverse of the transpose of a matrix, i.e. $\ma^{\invtrans} \defeq (\ma^{-1})^\trans = (\ma^\trans)^{-1}$. Given a vector $\vx \in \R^n$, we use
 $\mdiag(\vx) \in \R^{n \times n}$ to denote the diagonal matrix where $\mdiag(\vx)_{ii} = \vx(i)$ and whenever it is unambiguous, we will refer to $\mdiag(\vx)$ as $\mx$.\\
\\
\textbf{Matrix Ordering}: For symmetric matrices $\ma, \mb \in \R^{n \times n}$ we use $A \preceq B$ to denote the condition that $\vx^\top \ma \vx \leq \vx^\top \mb \vx$ for all $\vx \in \R^n$ and define $\succeq, \succ,$ and $\prec$ analogously. We say symmetric matrix $\ma \in \R^{n \times n}$ is positive semidefinite (PSD) if $\ma \succeq \textbf{0}$ and positive definite (PD) if $\ma \succ \textbf{0}$. For any PSD $\ma$, we define $\ma^{1/2}$ as the unique symmetric PSD matrix where $\ma^{1/2} \ma^{1/2} = \ma$ and $\ma^{1/2}$ as the unique PSD matrix where $\ma^{-1/2} \ma^{-1/2} = \ma^\pseudo$. \\
\\
\textbf{Matrix Norms}: Given a positive semidefinite matrix (PSD) $\ma$, we define the \emph{$\ma$-norm} (or seminorm) $\|\vx\|_\ma \defeq \sqrt{\vx^\top \ma \vx}$. We define the $\ell_2$ norm $\|\vx\|_2 =  \|\vx\|_\mI = \sqrt{\vx^\top \vx}$. We define the $\ell_2$ norm of a matrix $\ma$ to be the matrix norm \textit{induced} by the $\ell_2$ norm on vectors: for a matrix $\ma$, $\|\ma\|_2 = \max_{\|\vx\|_2 = 1} \|\ma \vx\|_2$. Similarly, for any PSD $\mm$ we define $\|\ma\|_\mm \defeq \max_{\|\vx\|_\mm = 1} \|\ma \vx\|_\mm = \|\mm^{1/2} A \mm^{-1/2}\|_2$. We observe that $\|\ma\|_\mm = \|\ma^\top\|_{\mm^{-1}}$. 
We further define the induced $\ell_1$ and $\ell_\infty$ matrix norms as $\|\ma\|_1 = \max_{\|\vx\|_1 = 1} \|\ma \vx\|_1$ and $\|\ma\|_\infty = \max_{\|\vx\|_\infty = 1} \|\ma \vx\|_\infty$. We observe that $\|\ma\|_1 = \max_j \sum_i |A_{i,j}|$ and $\|\ma\|_\infty = \max_j \sum_i |A_{j,i}| = \|\ma^\top\|_1$.\\
\\
\textbf{Spectral Quantities}: For a square matrix $\ma$, we let $\lmax(\ma)$ and $\lmin(\ma)$ denote its largest and smallest eigenvalues of $\ma$ respectively, we let $\kappa(\ma) \defeq \|\ma\|_2 \|\ma^{-1}\|_2$ denote its condition number, 
and we let $\rho(\ma) \defeq \max\{|\lmax(\ma)|, |\lmin(\ma)| \} = \lim_{k \rightarrow \infty} \norm{\ma^{k}}_2^{1/k}$ denote its spectral radius.

\subsection{M-Matrices}
\label{sub:prelim_mmatrices}

This paper deals extensively with a broad prevalent class of matrices known as M-matrices. We will present some basic definitions and prove some basic facts about M-matrices here. We start by giving the formal definition of an M-matrix:

\begin{definition}[M-Matrix]
A matrix $\mm \in \R^{n \times n}$ is an M-matrix if it can be expressed as $\mm = s \mI - \mb$ where $s > 0$ and $\mb$ is an entrywise nonnegative matrix with $\rho(\mb) \leq s$. 
\end{definition}

M-matrices have many useful properties, some of which can be found in \cite{berman1994nonnegative}. We state and prove the properties that are needed in our analysis, in  Appendix (see Section~\ref{sec:mmatrix_facts}).

\subsection{Perron-Frobenius Theorem}
\label{sub:perron-frobenius}

Here we give several definitions and lemmas related to Perron-Frobenius theorem. 
\begin{definition} A matrix $\ma\in \R^{n\times n}$ is called an irreducible matrix, if for every pair of row/column indices $i$ and $j$ there exists a natural number $m$ such that $\ma^m_{i,j}$ is non-zero.
\end{definition}

\begin{theorem}[Perron-Frobenius Theorem]
	Let $\ma \in \R_{\geq 0}^{n \times n}$ be an irreducible matrix with spectral radius $s \eqdef \rho(\ma)$. Then the following facts hold:
	\begin{itemize}
		\item  $s$ is a positive real number and it is an eigenvalue of $\ma$.
		\item The eigenvalue corresponding to $s$ is simple, which means the left and right eigenspaces associated with $s$ are one-dimensional.
		\item $\ma$ has a right eigenvector $\vv_R$ and a left eigenvector $\vv_L$ with eigenvalue $s$, which are component-wise positive. i.e. $\vv_L, \vv_R >0$.
		\item The only eigenvectors whose components are all positive, are the ones associated with $s$.
	\end{itemize}
\end{theorem}

\begin{lemma}[Collatz - Wielandt formula]\label{lem:collatz}
	Given a non-negative irreducible matrix $\ma \in \R_{\geq 0}^{n\times n}$, for all non-negative non-zero vectors $x \in \R_{\geq 0}^n$, let $f(x)$ be the minimum value of $[\ma x]_i / x_i$ taken over all those $i$'s such that $x_i \neq 0$. Then $f$ is a real valued function whose maximum over all non-negative non-zero vectors $x$ is the Perron-Frobenius eigenvalue.
\end{lemma}

\subsection{Row Column Diagonally Dominant Matrices}
\label{sub:prelim_rcdd}

A key result in our paper is to reduce the problem of solving M-matrices to solving a sequence of row column diagonally dominant (RCDD) matrices and then leverage previous work on solving such systems. Here we provide the results and notation that we use in the remainder of the paper. We begin by defining RCDD matrices.

\begin{definition}[Row Column Diagonally Dominant (RCDD) Matrices] A matrix $\ma \in \R^{n \times n}$ is \emph{RCDD} if $\ma_{ii} \geq \sum_{j \neq i} |\ma_{ij}|$ and $\ma_{ii} \geq \sum_{j \neq i} |\ma_{ji}|$ for all $i \in [n]$. It is \emph{strictly RCDD} if each of these inequalities holds strictly.
\end{definition}

Our algorithms make extensive use of recent exciting results that show linear system in RCDD systems can be solved efficiently. Since our algorithms only use RCDD solvers as black box, through this paper we simply assume that we have oracle access to a RCDD solver defined as follows.

\begin{oracle}[RCDD Solver] \label{oracle:RCDDSolve}
 Let $S \in \R^{n \times n}$ be an RCDD matrix, and let $x \in \R^n$ be a vector. Then there is a procedure, called $\code{RCDDSolve}(\ms, \epsilon)$ which can generate an operator $\opZ: \R^n \to \R^n$ where
 \[
 \Exp \left[ \|\vx - \ms \opZ (\vx)\|_2 \right] \leq \epsilon \| \vx \|_2     
 \]
 Further, we can compute and apply $\opZ$ to vectors in $\mathcal{T}_{solve}(m, n,\kappa(\ms), \epsilon)$ time. 
\end{oracle}

Designing faster RCDD solvers is an active area of research and therefore we parameterize our running times in terms of $\mathcal{T}_{solve}$ so that should faster RCDD solvers be developed our running times immediately improve. We achieve our claims about nearly linear time algorithms by leveraging the following results of \cite{coh16, coh17} that RCDD systems can be solved in nearly linear time.

\begin{theorem}[Nearly Linear Time RCDD Solvers \cite{coh16, coh17}]\label{thm:rcdd_solver}
It is possible to implement Oracle~\ref{oracle:RCDDSolve} with 
$$
\mathcal{T}_{solve}(m, n,\kappa(\ms), \epsilon) = O((m + n)^{1 +o(1)} \log^{O(1)}(\kappa(\ms) / \epsilon))\,.
$$
\end{theorem}

\section{A Technical Overview of Approach}
\label{sec:overview}

In this section we give a technical overview of the key results in our paper. We provide several of the core proofs and algorithms we use in the rest of the paper. Our focus in this section is to provide a semi-rigorous demonstration of our results. While the complete analysis of our results is rather technical due to the careful account for the numerical precision of various algorithmic subroutines that can only be carried out approximately, e.g. solving various linear systems, here we provide a vastly simplified analysis that ignores these issues, but conveys the core of our arguments. 

Formally, we make the following assumption through this section (and only this section) to simplify our analysis. In the remainder of the paper we carry on our analysis without it.

\begin{assumption} [Exact Numerical Operations Assumption Made Only in Section~\ref{sec:overview}] \label{asmp:exact-solver}
For any routine that can compute a vector $v$ with $\epsilon$-approximate error in some norm $\|\cdot\|$, i.e. compute $\tilde{v}$ with $\|v - \tilde{v}\| < \epsilon $, in time $O(\mathcal{T} \log^{O(1)}(1/\epsilon))$ we say there is a routine that can compute $v$ exactly in $\otilde(\mathcal{T})$ time, where here $\otilde$ hides factors polylogarithmic in $T$ as well as natural parameters of the problem, e.g. $n$.
\end{assumption} 

Note that this assumption is obviously not met in various settings and can easily be abused by performing arithmetic operations on overly large numbers for free. However, it certainly simplifies analysis greatly, as for instance, under this assumption, Theorem~\ref{thm:rcdd_solver}. implies that a $n \times n$ RCDD matrix with $m$ non-zero entries can be solved exactly in $\otilde(m + n)$ time. However, in many cases (and is formally the case in this paper), analysis under the assumption loses only polylogarithmic factors in the overall running time. 

Again, we make this assumption purely to simplify exposition in this section. In the remainder of the paper we carry on our analysis without it.

\subsection{From M-Matrices to RCDD Matrices}\label{sec:mmatrix-rcdd}

We begin our technical overview by formally stating the connections between M-matrices and diagonal scaling that make such matrices RCDD.

Given a non-negative irreducible matrix $\ma \in \R^{n \times n}$, suppose $\mm  = s\mI - \ma$ is an invertible M-matrix for some $s > 0$. As we discussed in Section~\ref{sub:intro_tech_motiv} we know that if $\vv_l \in \R^n_{> 0}$ and $\vv_r \in \R^{n}_{> 0}$ are the left and right Perron vectors of $\ma$ respectively then 
\[
\vv_l^\top \mm = (s - \rho(\ma)) \vv_l^\top > 0
\text{ and }
\mm \vv_r = (s - \rho(\ma)) \vv_r > 0
\]
Consequently, for $\mv_l = \mdiag(\vv_l)$ and $\mv_r = \mdiag(\vv_r)$ this property implies that $\mv_l \mm \mv_r$ is RCDD since its off-diagonal entries are non-positive and 
$$
\mv_l \mm \mv_r \vones = \mv_l \mm \vv_r > 0 \text{ and }  
\vones^\top \mv_l \mm \mv_r = 
\vv_l^\top \mm \mv_r > 0.
$$
Thus any invertible M-matrix can be scaled with positive diagonal matrices to become RCDD. Furthermore, since systems of linear equations with RCDD coefficient matrices can be solved in nearly linear time this implies that given these scaling M-matrices can be solved in nearly linear time, as to solve  $\mm z = b$, one can use an RCDD solver to find $z = \mv_r^{-1} (\mv_l \mm \mv_r)^{-1} \mv_l^{-1} b$ in nearly linear time. 

More interestingly, the reduction also works in the other direction. This means if one can solve a linear system described with an M-matrix efficiently, then it is possible to find $\vell, \vr \in \R^{n}_{> 0}$ such that $\ml \mm \mr$  is RCDD where $\ml = \mdiag(\vell)$ and $\mr = \mdiag(\vr)$. Again from M-matrix theory we know that if $\mm$ is an invertible M-matrix then $\mm^{-1}$ is entrywise non-negative. Consequently, if we apply $\mm^{-1}$ to any positive vector, the resulting vector remains positive. Given this observation, let $\vl = \mm^{\invtrans} \vones$, and $\vr = \mm^{-1} \vones$. Now $\mm \vr = \mm \mm^{-1} \vones > \vzero$ and similarly $\vl^\top \mm > \vzero$. Thus, $\ml \mm \mr$ is RCDD and we see that one can look at either of $\vl$ or $\vr$ as a certificate for the matrix $\mm$ to be an M-matrix.

We use both of these directions in our algorithm to either find the desired scalings or solve linear systems fast. We further use the following lemma which is a formal statement of what we discussed above. The proof is deferred to Appendix (Section~\ref{sec:mmatrix_facts}).

\begin{restatable}{lemma}{MMatrixRCDD}
	\label{lemma:MMatrix_RCDD}
	Let $\ma$ be a nonnegative matrix, and consider $\mm = s \mI - \ma$ for some $s > 0$. Then if $\mm$ is an M-matrix then for every pair of vectors $\vx$ and $\vy$ where $\mm \vy > 0$ and $\mm^\top \vx > 0$ we have $\ms = \mx \mm \my$ is RCDD. Further, if there exist positive vectors $\vx$, $\vy$ where $\ms = \mx \mm \my$ is RCDD, then $\mm$ is an M-matrix.
\end{restatable}

Also, the following observation about M-matrices and these scaling is important in our analysis (see Section~\ref{sec:mmatrix_facts} for the proof).
\begin{restatable}{lemma}{pdscaling}
	\label{lem:pdscaling}
	Given a non-negative matrix $\ma \in \R_{\geq 0}^{n\times n}$, where $\mm= s \mI-\ma$ is an invertible M-matrix, l et $\ml, \mr \in \R_{>0}^{n\times n}$ be the left and right scalings that make $\ml \mm \mr$ RCDD and let $\md = \ml \mr^{-1}$. Then $\md^{1/2} \mm \md^{-1/2} + \md^{-1/2} \mm \md^{1/2}$ is PD. 
\end{restatable}


\subsection{Algorithm for Solving Linear Systems in M-matrices}
Relying on the connection between M-matrices and RCDD matrices discussed in Section~\ref{sec:mmatrix-rcdd} and the fact that linear systems described by RCDD matrices can be solved in nearly linear time Section~\ref{sub:prelim_rcdd}, in this section we present one of our central results which is a simple iterative algorithm that can be used to solve systems of linear equations described by M-matrices. 
Particularly, we give a nearly linear time algorithm for solving the following problem.

\begin{definition}[M-matrix Scaling Problem]
Given non-negative $\ma \in \R_{\geq 0}^{n\times n}$ and value $s$ where $\mm = s \mI - \ma$ is an invertible M-matrix, i.e. $\rho(\ma)< s $, compute scaling vectors $\vell, \vr \in \R^{n}$, such that $\ml \mm \mr$ is strictly RCDD, where $\ml = \diag(l)$ and $\mr = \diag(r)$.
\end{definition}

The main statement we will prove in this section is Theorem~\ref{thm:mmatrix-scale-overview} given below. The fully rigorous version of this Theorem~\ref{thm:mscale} is one of the main results of the paper and provides a tool for solving linear systems in M-matrices in nearly linear time. We also use this result later to give our nearly linear time algorithm for finding the top eigenvalue and eigenvectors of non-negative matrices as described by Perron-Frobenius theorem. We remark that the following statement is written assuming Assumption~\ref{asmp:exact-solver} and the full version of this theorem is stated in Theorem~\ref{thm:mscale}.

\begin{theorem}[M-matrix Scaling, Informal]\label{thm:mmatrix-scale-overview}
Given a value $s > 0$, a non-negative matrix $\ma \in \R_{\geq 0}^{n\times n}$ with $m$-nonzero entries where $\rho(\ma) < s$, and $\epsilon>0$ let $\mm_\alpha \eqdef (1+\alpha) s \mI - \ma$. Algorithm~\ref{alg:MMatrix-overview}, finds positive diagonal matrices $\ml, \mr \in \R_{>0}^{n\times n}$ such that $\ml \mm_\epsilon \mr$ is RCDD in time $\tilde{O}(m \log(\frac{\|\ma\|_1+ \|\ma\|_\infty}{s \epsilon}))$.
\end{theorem}

Note that in solving the above problem we can assume without loss of generality that $s = 1$: observe that $\mm' = \frac{1}{s} \mm = \eye - \frac{1}{s} \ma$ is an M-matrix and additionally note that $\mm_\alpha' = \frac{1}{s} \mm_\alpha$. Thus any scaling of $\mm_\epsilon'$ is a scaling of $\mm_\epsilon$. Note that the running time is unaffected by such scalings since $\| \frac{1}{s}\ma \| = \frac{1}{s} \| \ma \|$. Thus from now on whenever we discuss Theorem \ref{thm:mmatrix-scale-overview} we assume $s=1$ for simplicity. To solve the M-matrix scaling problem with $s=1$ we take a simple iterative approach of solving and computing scalings for $\mm_\alpha$ defined for all $\alpha \geq 0$ by 
$
\mm_{\alpha} \eqdef \mm + \alpha \mI
$.
In particular, we consider the scaling vectors
$
\vell_\alpha \eqdef \mm_{\alpha}^{-\top} \bone \text{ and } \vr_\alpha \eqdef \mm_{\alpha}^{-1} \bone
$
and let $\ml_{\alpha}\eqdef \diag(\vell_\alpha)$ and $\mr_{\alpha}\eqdef \diag(\vr_\alpha)$ (see Algorithm~\ref{alg:MMatrix-overview}). As discussed in Section~\ref{sec:mmatrix-rcdd}, $\ml_{\alpha} \mm_{\alpha} \mr_{\alpha}$ is RCDD. Note that to solve $\mm_{\alpha} \vx = \vb$, we can instead solve $\ml_{\alpha}^{-1}(\ml_{\alpha} \mm_{\alpha} \mr_{\alpha})\mr_{\alpha}^{-1} \vx = \vb$, which is equivalent to computing $\vx = \mr_{\alpha} (\ml_{\alpha} \mm_{\alpha} \mr_{\alpha})^{-1} \ml_{\alpha} \vb$. This means solving an equation in $\mm_{\alpha}$ reduces to inverting a RCDD matrix. 

The above discussion shows computing $\ml_{\alpha}$ and $\mr_{\alpha}$  gives us access to $\mm_{\alpha}^{-1}$ in nearly linear time. We show computing $\ml_{\alpha}$ and $\mr_{\alpha}$ for one value of $\alpha$, suffices to compute it for another nearby. To do this, we show ${\mm_{\alpha'}}^{-1}$ for $\alpha'$ constantly close to $\alpha$ is a good preconditioner of $\mm_{\alpha}$, i.e. $\mm_{\alpha'}^{-1} \mm_{\alpha} \approx \mI$ in a specific norm. It is known that when this happens a simple procedure known as preconditioned Richardson (Algorithm~\ref{alg:prichardson}) can solve a linear system in $\mm_{\alpha}$ using just $\otilde(1)$ applications of $\mm_{\alpha}$ and $\mm_{ \alpha'}^{-1}$, i.e. nearly linear time.

\newcommand{\prichardson}{$\code{Prec-Richardson}$}

\begin{algorithm} [h!]\caption{Exact Preconditioned Richardson}\label{alg:prichardson}
	\begin{algorithmic}[1]
		\Function{\prichardson}{$\mm$, $\mP$, $\vb$, $\epsilon$, $\vx_0$,  $n_{iter}$}
		\State \textbf{Input:} $\mm \in \R^{n\times n}$, $\mP \in \R^{n\times n}$ is a preconditioner for $\mm$, $\vb \in \R^n$, $\vx_0\in \R^n$
		\State \textbf{Input:} $n_{iter}$ is an upper-bound on the number of iterations to run the algorithm.
		\State \textbf{Output:} Either $\vx \in \R^n$ with $\|\mm \vx - \vec{b}\|_2 < \epsilon \|\vx_0\|_2$ or failure after $n_{iter}$ iterations
		\State $\vx \eq \vx_0$
		\Repeat
		\State $\vx \eq \vx - {\mP} [\mm \vx - \vb]$
		\Until{$\|\mm \vx - \vec{b}\|_2 <  \epsilon \|\mm \vx_0-\vb\|_2$ or number of iterations exceeds $n_{iter}$} 
		\State \Return $\vx$
		\EndFunction
	\end{algorithmic}	
\end{algorithm}

The following lemmas provides the formal bound on $\mm_{\alpha'}^{-1}\mm_{\alpha}$ that allows this procedure to work.

\begin{lemma}\label{lem:preconditioner}
Let $\md$ be a symmetric PSD matrix such that $\md^{1/2} \mm \md^{-1/2} + \md^{-1/2} \mm^\top \md^{1/2}$ is PSD. Then for all $\alpha, \alpha' > 0$ the following holds
$$
\|\mm_{\alpha'}^{-1} \mm_\alpha - \mI \|_{\md \rightarrow \md} 
= \left|\alpha-\alpha'\right|\cdot\norm{\md^{1/2}\mm_{\alpha}^{-1}\md^{-1/2}}_{2}
\leq \frac{|\alpha'-\alpha|}{\alpha'} ~.
$$
\end{lemma}
\begin{proof}
Note that $\mm_{\alpha}=\mm_{\alpha'}+(\alpha-\alpha')\mI$. Consequently,
$
\mm_{\alpha'}^{-1}\mm_{\alpha}-\mI=(\alpha-\alpha')\mm_{\alpha'}
$
and therefore
\[
\norm{\mm_{\alpha'}^{-1}\mm_{\alpha}-\mI}_{\md\rightarrow\md}=\left|\alpha-\alpha'\right|\cdot\norm{\md^{1/2}\mm_{\alpha}^{-1}\md^{-1/2}}_{2}
\]
yielding the first equality of the claim. For the inequality, note that
\[
\norm{\md^{1/2}\mm_{\alpha}^{-1}\md^{-1/2}}_{2}
= \sqrt{\lambda_{\max}\left(\md^{-1/2} \mm_{\alpha}^{\invtrans}
\md\mm_{\alpha}^{-1}\md^{-1/2}\right)}=\frac{1}{\sqrt{\lambda_{\min}\left(\md^{1/2}\mm_{\alpha}\md^{-1}\mm_{\alpha}^{\top}\md^{1/2}\right)}}\,.
\]
Furthermore, we see that
\begin{align*}
\md^{1/2}\mm_{\alpha}\md^{-1}\mm_{\alpha}^{\top}\md^{1/2} & =\md^{1/2}\mm\md^{-1}\mm\md^{1/2}+\alpha^{2}\mI+\alpha\left[\md^{-1/2}\mm^{\top}\md^{1/2}+\md^{1/2}\mm\md^{-1/2}\right]\,.
\end{align*}
Since $\md^{1/2}\mm\md^{-1}\mm\md^{1/2}\succeq\mzero$ and $\md^{-1/2}\mm^{\top}\md^{1/2}+\md^{1/2}\mm\md^{-1/2}\succeq\mzero$
by assumption, we have 
\[
\md^{1/2}\mm_{\alpha}\md^{-1}\mm_{\alpha}^{\top}\md^{1/2}\succeq\alpha^{2}\mI\text{ and }\lambda_{\min}(\md^{1/2}\mm_{\alpha}\md^{-1}\mm_{\alpha}^{\top}\md^{1/2})\geq\alpha^{2}
\]
yielding the result. 
\end{proof}

By Lemma~\ref{lem:preconditioner}, any $\md$  which makes $\md^{1/2} \mm \md^{-1/2}$ PSD lets us bound $\|\mm_{\alpha'}^{-1} \mm_{\alpha} - \mI \|_{\md}$. Furthermore, in Lemma~\ref{lem:pdscaling} we showed that for $\md = \ml_{0} \mr^{-1}_{0}$ $\md^{1/2} \mm \md^{-1/2}$ is PSD. Therefore, we already know an appropriate scaling $\md$.

Having all these observations we are ready to give an efficient algorithm for solving M-matrix scaling problem. 
Our algorithm (Algorithm~\ref{alg:MMatrix-overview}) starts with a large $\alpha$ (e.g. $\alpha =\max\{\|\ma\|_\infty, \|\ma\|_1\}$) such that $\mm_{\alpha}$ trivially becomes RCDD, then halves $\alpha$ in each iteration and find $\ml_{\alpha}, \mr_{\alpha}$ for the new $\alpha$. 
Using Lemmas~\ref{lem:pdscaling} and \ref{lem:preconditioner}, we just showed there exists a norm where $\|\mm_{\alpha}^{-1} \mm_{\alpha/2} - \mI \|_{\md} \leq 1/2$.  
Provided we can solve linear systems in $\mm_{\alpha}$, we compute $\ml_{\alpha/2}$ and $\mr_{\alpha/2}$ using preconditioned Richardson (see Algorithm~\ref{alg:prichardson}). By Assumption~\ref{asmp:exact-solver}, this requires solving $\tilde{O}(1)$ linear systems in $\mm_{\alpha}$. 
Since we already have access to $\ml_{\alpha}$ and $\mr_{\alpha}$ we can solve systems in $\mm_{\alpha}$ exactly in $\otilde(m)$ (again by Assumption~\ref{asmp:exact-solver}). Therefore $\ml_{\alpha/2}$ and $\mr_{\alpha/2}$ can be computed in $\otilde(m)$. Now to find scalings for $\mm_\epsilon$, we iterate $\otilde(\log(\frac{\|\ma\|_1+\|\ma\|_\infty}{\epsilon}))$ times, giving the overall running time of $\otilde(m \log(\frac{\|\ma\|_1+\|\ma\|_\infty}{\epsilon})$. This proves the statement in Theorem~\ref{thm:mmatrix-scale-overview}.

\newcommand{\mscale}{$\code{MMatrix-Scale-Overview}$}
\begin{algorithm} [h]
	\caption {Compute left and right scalings of an M-matrix to make it RCDD.} \label{alg:MMatrix-overview}
	\begin{algorithmic}[1]
		\Function{ \mscale}{$A$, $\epsilon$}
		\State \textbf{Input:} $\ma \in \R_{\geq 0}^{n\times n}$ is an irreducible matrix with $\rho(\ma)<1$, and $\epsilon >0$ is a real number.
		\State \textbf{output:} Positive diagonal scalings $\ml$ and $\mr$ such that $\ml [(1+\epsilon) \mI - \ma] \mr$ is RCDD
		\State $\alpha \eq \max\{\|\ma\|_1, \|\ma\|_\infty\}$
		\State $\ml_{\alpha}, \mr_{\alpha} \eq I$
		\While{$\alpha > \epsilon$}
		
		\State  \textit{{\footnotesize //$\mm_\alpha = (1+\alpha) I - A$.}}
		\State  \textit{\footnotesize //In the following lines, we assume access to an efficient solver for $\mm_{\alpha}^{-1}$ as $L_{\alpha}$, $R_{\alpha}$ are already computed.}
		\State $\vec{\ell}_{\alpha/2} \eq$ {\sc \prichardson}($M^\top_{\alpha/2}$, ${\mm_{\alpha}^{-1}}^\top$, $\vec{\bone}$, $\vec{0}$, 0, $\infty$)
		\State $\vec{r}_{\alpha/2} \eq$ {\sc \prichardson}($M_{\alpha/2}$, $\mm_{\alpha}^{-1}$, $\vec{\bone}$, $\vec{0}$, 0, $\infty$)
		\State $\alpha \eq \alpha/2$
		\EndWhile
		\State \Return $\ml_{\alpha},\mr_{\alpha}$ with $\ml_{\alpha} M_{\epsilon} \mr_{\alpha}$ being RCDD 
		\EndFunction
	\end{algorithmic}	
\end{algorithm}

We used Assumption~\ref{asmp:exact-solver} twice in the above argument. In Section~\ref{sec:scaling} through a careful analysis, we show how to adjust each of these assumptions to the approximate world while avoiding to increase the running time by more than a polylogarithmic factor in $n$, $\frac{1}{\epsilon}$, and $\kappa(\ml_0)+\kappa(\mr_0)$.

\subsection{Finding the top eigenvalue and eigenvectors of positive matrices}\label{sec:perron-alg-overview}
In this section, we present an overview of our result on computing the largest eigenvalue and corresponding eigenvectors of a non-negative matrix as characterized by Perron-Frobenius theorem. In the last section we showed how to find the left and right scalings for an M-matrix to make it RCDD. Here we discuss how to use this idea to come up with a nearly linear time algorithm for the Perron vector problem (see Definition~\ref{def:perron}) or equivalently M-matrix decision problem, defined below.

\begin{definition}[M-matrix decision problem]
	Given a non-negative irreducible matrix $\ma \in \R_{\geq0}^{n\times n}$, and real number $s > 0$, decide whether $s\mI- \ma$ is an M-matrix, or equivalently $\rho(\ma) < s$. 
\end{definition}
If we know the top eigenvalue of $\ma$, then to solve the M-matrix decision problem we only need to check $\rho(A) < s$. Thus a solver for the Perron problem, gives an immediate solution for the M-matrix decision problem. On the other hand, if we have a solver for the M-matrix decision problem, to get a solver for the Perron problem, we need to use a search strategy to find the smallest $s$ such that $\rho(A) < s$. We use the latter approach combined with the ideas from Algorithm~\ref{alg:MMatrix-overview} to come up with a nearly linear time algorithm for the Perron problem. Theorem~\ref{thm:perron} demonstrates this result. A full detailed analysis of this theorem is given in Section~\ref{sec:perron-alg}. Below we discuss the main ideas and steps that are taken to prove Theorem~\ref{thm:perron}, and point to the corresponding subsections in Section~\ref{sec:perron-alg}.

\subsubsection{The M-matrix Decision Problem}\label{sec:mmatrix-decision-overview}

Here we show how a slightly modified version of Algorithm~\ref{alg:MMatrix-overview} can be used to solve the M-matrix decision problem. Note that if Algorithm~\ref{alg:MMatrix-overview} succeeds to find a scaling for $\mm_\epsilon$, it is also a proof that $\rho(\ma) < 1+\epsilon$. 
However, Algorithm~\ref{alg:MMatrix-overview} only works on the premise that $\rho(\ma) < 1$. 
If we remove this assumption, we are no longer able to apply Lemmas~\ref{lem:pdscaling} and~\ref{lem:preconditioner}. In other words, it might no longer be the case that $\mm_{\alpha}^{-1}$ is a good preconditioner for $\mm_{\alpha/2}$ (i.e. satisfies $\mm_{\alpha}^{-1} \mm_{\alpha/2} \approx \mI$ in some norm) or that $\mm_{\alpha/2}$ is an M-matrix. Therefore, preconditioned Richardson might fail to terminate within a reasonable time, or find an appropriate scaling. To give an example on what might go wrong, note that in our analysis we rely on the fact $\vr_{\alpha/2} = \mm_{\alpha/2}^{-1} \vones  > \vzero$, but this is not necessarily true when $\mm_{\alpha/2}$ is not an M-matrix.

 Remember we used Assumption~\ref{asmp:exact-solver}, and Lemmas~\ref{lem:pdscaling},~\ref{lem:preconditioner}, to show each invocation of preconditioned Richardson in Algorithm~\ref{alg:MMatrix-overview} generates an exact solution after $\otilde(1)$ iterations. Interestingly, we can use this to test whether $\mm$ is an M-matrix or not. Given that we have access to an exact solver for $\mm_\alpha$, if preconditioned Richardson fails to find an appropriate scaling after $\otilde(1)$ iterations, it is a proof that $\mm$ is not an M-matrix. On the other hand, if we limit the number of iterations for preconditioned Richardson to $\otilde(1)$ and our algorithm still succeeds to find appropriate scalings for $\mm_\epsilon$, this is a proof that $\mm_\epsilon$ is actually an M-matrix. This way we can come up with a nearly linear time algorithm which either reports that $\mm$ is not an M-matrix, or provides a proof that $\mm_\epsilon$ is an M-matrix. 

For a full discussion on the M-matrix decision problem and the bounds we get for solving this problem refer to Section~\ref{sec:mmatrix-decision}. In this section there are further complications that we address to deal with the fact that the condition number of the top eigenvectors of $\ma$ is not known. The main issue comes up with choosing the number of iterations to run the preconditioned richardson. We discuss this issue briefly here, since it impacts the running time of the algorithm in a way other than tunning the internal solvers precisions. For a detailed discussion we refer the reader to Section~\ref{sec:kfix}. 

In the aforementioned scheme, given that $\mm_\alpha^{-1}$ is a preconditioner for $\mm_{\alpha/2}$, we assumed (by Assumption~\ref{asmp:exact-solver}) there exists a function $f = \otilde(1)$ such that the preconditioned Richardson (see Algorithm~\ref{alg:MMatrix-overview}) should terminate with an exact solution after $f(\cdot)$ iterations. To get the above strategy actually work, we need to know an upper-bound on $f$, which turns out to have a polylogarithmic dependence on $\kappa(\vv_l)$ and $\kappa(\vv_r)$, where $\vv_l$ and $\vv_r$ are the left and right eigenvectors of $\ma$. Unfortunately, we do not know $\kappa(\vv_l)$ and $\kappa(\vv_r)$ in advance. 

One potential fix to this problem is to guess an upper-bound $K$ on the condition numbers of left and right eigenvectors and double it until $K \geq \kappa(\vv_l) + \kappa(\vv_r)$. This simple doubling scheme does not quite work as it is not clear when we should stop doubling $K$. Instead we provide a more complex doubling scheme which can solve the problem. For this we refer reader to Section~\ref{sec:kfix}.


%
%
\subsubsection{Perron problem}
We use our algorithm for the M-matrix decision problem to provide an algorithm for the Perron problem. Our approach for solving the Perron problem consists of two stages. First, we compute the largest eigenvalue, namely $\rho(\ma)$, within an $\epsilon$ multiplicative error. Second, we use this computed value to find approximate eigenvectors corresponding to $\rho(\ma)$. Our method for finding the largest eigenvalue of $\ma$ is depicted in Algorithm~\ref{alg:perronvalue}, function {\sc \findperronvalue} (Section~\ref{sec:perron-sub}).  By using the algorithm for the M-matrix decision problem, for a given $s \in \R_{>0}$, we can either determine $\rho(\ma) > s$ (i.e. $\mI - \ma/s$ is not an M-matrix), or $\rho(\ma) \leq (1+\epsilon) s$ (i.e. $(1+\epsilon)\mI - \ma/s$ is an M-matrix). We use this fact in {\sc \findperronvalue}, to apply a binary search which finds $\rho(A)$ within an $\epsilon$ multiplicative error. For a detailed discussion and analysis of this algorithm refer to Section~\ref{sec:perron-sub}.

Next, in Section~\ref{sec:approx-pvector} we show how to find an approximate eigenvector of $A$ corresponding to its largest eigenvalue. We define an approximate eigenvector as follows.

\begin{definition}[Approximate Eigenvector] \label{def:approx-ev}
	Given a matrix $\ma \in \R^{n \times n}$ and a vector norm $\|.\|$, let $\lambda$ be a non-zero eigenvalue of $\ma$, then a non-zero vector $\vr \in \R^n$ is an $\epsilon$-approximate eigenvector of $\ma$ corresponding to $\lambda$ if $\|(\mI-\frac{\ma}{\lambda})\vr\| < \epsilon \|\vr\|$.
\end{definition}

In Lemma~\ref{lem:perron-approx}, we show that $[(1+\epsilon)\mI - \ma / \rho(\ma)]^{-1} \vones$ is a $2\epsilon$-approximate eigenvector corresponding to $\rho(\ma)$ (in infinity norm). Note that we are able to compute $[(1+\epsilon)\mI - \ma / \rho(\ma)]^{-1} \vones$ in nearly linear time with our machinery for the matrix scaling problem (refer to $\mm_\epsilon$ and $\vr_\alpha$ definitions in Section~\ref{sec:mmatrix-decision-overview}). Therefore if we have $\rho(\ma)$ we can compute its corresponding approximate eigenvectors efficiently. We then show similar bounds hold if instead of $\rho(\ma)$ we use an $1+\epsilon$ approximation of it. This completes our two step approach for finding both the top eigenvalue and eigenvectors of a non-negative irreducible matrix. This simple approach could work under Assumption~\ref{asmp:exact-solver}, if the terms we are hiding in $\otilde(1)$ notation in {\sc \mscale} algorithm did not depend on $\kappa(\vv_l)$ and $\kappa(\vv_r)$, or either we knew a tight upper-bound for them.  However, as discussed in the previous subsection our algorithm for the M-matrix decision problem depends on knowing this bound. We show a fix for this in Section~\ref{sec:kfix}.

\section{Numerically Stable M-Matrix Scaling}
\label{sec:scaling}

\newcommand{\fullmscale}{$\code{MMatrix-Scale}$}

In this section,  we analyze algorithm \fullmscale$(\ma, \epsilon, K)$ in full detail. The algorithm can be viewed as a cousin to shift-and-invert preconditioning methods for  the recovery of the top eigenvector of a matrix. Given an input M-matrix $\mm = s \eye - \ma$, our algorithm maintains  a left-right scaling\footnote{As a reminder, a pair of diagonal matrices $(\ml,\mr)$ is a left-right scaling of an matrix $\mm$ if $\ml\mm \mr$ is RCDD.} $(\ml_\alpha , \mr_\alpha)$ of the matrix $\mm_\alpha \eqdef \mm + \alpha s \eye$ for progressively smaller values of $\alpha$. These scalings will be obtained by approximately solving the linear systems
\[
\mm_\alpha \vr_\alpha \approx \vones \hspace{3mm} \text{and} \hspace{3mm} \mm_\alpha^\top \vl_\alpha \approx \vones.
\]
As discussed in the previous section, we observe that we can assume $s=1$, as for any $\alpha \geq 0$ the matrix $\mm_\alpha' = \frac{1}{s}\mm_\alpha = (1+\alpha)\eye - \frac{1}{s}\ma$: solving the scaling problem for $\mm' = \eye - \frac{1}{s}\ma$ is identical to solving the scaling problem for $\mm$. We asume this throughout the rest of the section. We will choose our initial choice of $\alpha$ to be such that $\mm_\alpha$ is itself RCDD: thus our initial scalings can simply be $\vl = \vr = \frac{1}{\alpha} \vones$. In every step of our procedure, we will use our computed left-right scaling $(\ml_{2\alpha}, \mr_{2\alpha})$ of $\mm_{2 \alpha}$ to obtain a scaling of $\mm_{\alpha}$. We do this in three parts. We first use the recent result on solving RCDD linear systems to compute a good preconditioner $\mP_{2 \alpha}$ for the matrix $\mm_{2 \alpha}$ given the left-right scaling $(\ml_{2\alpha}, \mr_{2\alpha})$. We then show that under a certain measure of error this preconditioner for $\mm_{2 \alpha}$ also acts as a preconditioner for $\mm_\alpha$. Finally, we argue that by applying a standard iterative procedure preconditioned with $\mP_{2 \alpha}$ to solve $\mm_\alpha \vr_\alpha = \vones$ and $\mm_\alpha^\top \vl_\alpha = \vones$ we can solve these linear systems in a small number of iterations precisely enough to recover a scaling of $\mm_{\alpha}$. Once we have obtained our scaling of $\mm_\alpha$ we repeatedly halve $\alpha$ until our additive identity factor becomes small enough.

As discussed we begin by showing that solving M-matrices reduces to solving RCDD linear systems if we are given the left-right scaling of our desired M-matrix. 

\begin{algorithm} [t]
    \caption {\sc $\code{SolvefromScale}(\mm, \ml, \mr, \delta)$} \label{alg:SolveMFromScale}    
\begin{algorithmic}[1]
    \State \textbf{Input:} $\mm\in \R^{n\times n}$ is an M-matrix, $\ml , \mr \in \R^{\nn}$ are diagonal matrices where $\ml \mm \mr$ is RCDD, $\delta >0$ is a real number.
    \State \textbf{output:} Operators $\Pleft, \Pright$ satisfying Theorem~\ref{lemma:MSolve}.
    \State $\ms \gets \ml \mm \mr$;
    \State $\opZ^{(r)} \gets \code{RCDDSolve}(\ms, \frac{\delta}{\kappa(\ml)})$ generated by Oracle~\ref{oracle:RCDDSolve};
    \State $\opZ^{(l)} \gets \code{RCDDSolve}(\ms^\top, \frac{\delta}{\kappa(\mr)})$ generated by Oracle~\ref{oracle:RCDDSolve};
    \State \Return operators $\Pright(\vx) = \mr \opZ^{(r)}(\ml \vx)$ and $\Pleft(\vx) = \ml \opZ^{(l)}(\mr \vx)$; 
\end{algorithmic}	
\end{algorithm}
\begin{lemma}\label{lemma:MSolve}
    Let $\mm$ be an M-matrix. Let $\ml$ and $\mr$ be diagonal matrices where $\ms = \ml \mm \mr$ is RCDD. Let $\md$ be a (potentially unknown) positive definite matrix with condition number at most $\gamma$. Assume we have access to Oracle~\ref{oracle:RCDDSolve}. Then Algorithm~\ref{alg:SolveMFromScale} with $\delta = \frac{\epsilon}{\gamma}$ generates operators $\Pright$  and $\Pleft$ 
    \[
    \Exp \left[ \|\vb - \mm \Pright (\vb) \|_\md \right] \leq \epsilon \|\vb\|_\md
    ~ \text{ and } ~
    \Exp \left[ \|\vb - \mm^\top \Pleft (\vb) \|_{\md^{-1}} \right] \leq  \epsilon \|\vb\|_{\md^{-1}}  
    ~.
    \]
    for any given $\vb$. We can compute and apply both $\Pright$ and $\Pleft$ to $\vb$ in time 
    \[O(\mathcal{T}_{solve}(m, n, \kappa(\ml \mm \mr), \frac{\epsilon}{\gamma \kappa(\mr)}) + \mathcal{T}_{solve}(m, n, \kappa(\ml \mm \mr), \frac{\epsilon}{\gamma \kappa(\ml)})) ~.
    \]
\end{lemma}
\begin{proof}
    Consider the output of Algorithm~\ref{alg:SolveMFromScale}. We observe that 
    \begin{align*}
    \|\vb - \mm \Pright (\vb) \|_2 &=\|\vb - \mm \mr \opZ^{(r)}(\ml \vb) \|_2 \\
    &= \|\vb - \ml^{-1} \ms \mr^{-1} \mr \opZ^{(r)}(\ml \vb) \|_2  \\
    &= \|\ml^{-1} \ml \vb - \ml^{-1} \ms \opZ^{(r)}(\ml \vb) \|_2  \\
    &\leq   \|\ml^{-1}\|_2 \,  \|\ml \vb - \ms \opZ^{(r)}(\ml \vb) \|_2 .
    \end{align*}
    However by the guarantee of Oracle~\ref{oracle:RCDDSolve} and the construction of $\opZ^{(r)}$ we have that 
    \[
    \Exp \left[ \|\ml \vb - \ms \opZ^{(r)}(\ml \vb) \|_2 \right] \leq \frac{\epsilon}{\gamma \kappa(\ml)} \|\ml \vb\|_2 \leq \frac{\epsilon}{\gamma \kappa(\ml)} \, \|\ml\|_2 \, \|\vb\|_2.
    \]
    Taking expectation over $\opZ^{(r)}$ and combining these two facts gives
    \[
    \Exp \left[ \|\vb - \mm \Pright (\vb) \|_2 \right] \leq \frac{\epsilon}{\gamma \kappa(\ml)} \, \|\ml^{-1}\|_2 \, \|\ml\|_2 \, \|\vb\|_2 = \frac{\epsilon}{\gamma} \|\vb\|_2.
    \]
    Now, we note that $\|\vy\|_2 \leq \|\vz\|_2$ implies $\|\vy\|_\md \leq \kappa(\md) \|\vz\|_\md$ for any $\vy, \vz \in \R^n$, and so we have
    \[
    \Exp \left[ \|\vb - \mm \Pright (\vb) \|_\md \right] \leq  \frac{\kappa(\md) \epsilon}{\gamma} \|\vb\|_\md \leq \epsilon \|\vb\|_\md.
    \]
    Repeating this analysis with $\Pleft$ in the $\md^{-1}$ norm gives the analogous
    \[
    \Exp \left[ \|\vb - \mm^\top \Pleft (\vb) \|_{\md^{-1}} \right] \leq  \epsilon \|\vb\|_{\md^{-1}}
    \]
    as desired. To bound the running time, we observe that computing and applying $\Pright$ and $\Pleft$ to vectors requires only computing and applying $\opZ^{(r)}$ and $\opZ^{(l)}$ and $O(n)$ additional work to apply $\mr$ and $\ml$. By the running time assumption of Oracle~\ref{oracle:RCDDSolve} and the error parameters chosen, we obtain that $\code{SolvefromScale(\mm,\ml,\mr,\frac{\epsilon}{\gamma})}$ runs in time
    \[
    O\left(\mathcal{T}_{solve}\left(m, n, \kappa(\ml \mm \mr), \frac{\epsilon}{\gamma \kappa(\mr)}\right) + \mathcal{T}_{solve}\left(m, n,\kappa(\ml \mm \mr), \frac{\epsilon}{\gamma \kappa(\ml)}\right)\right).    
    \]
\end{proof}

\begin{algorithm} [t]
    \caption {\sc \fullmscale$(A, \epsilon, K)$} \label{alg:MMatrix}    
\begin{algorithmic}[1]        
            \State \textbf{Input:} $\ma \in \R_{\geq 0}^{n\times n}$ is an irreducible matrix with $\rho(\ma)<1$, and $\epsilon, K >0$ are real numbers.
            \State \textbf{Output:} Positive diagonal scalings $\ml$ and $\mr$ such that $\ml [(1+\epsilon) \eye - A] \mr$ is RCDD
        \State $\alpha \eq 2 \max\{\|A\|_1, \|A\|_\infty\}$
        \State $\vl_\alpha \gets \frac{1}{\alpha} \vones, \quad \vr_\alpha \gets \frac{1}{\alpha} \vones$
        \While{$\alpha > \epsilon$}
            \State $\alpha \eq \alpha/2, \quad  k \gets 0$
            \State $\vl_{\alpha}^{\, \, (0)} \gets \vzero, \quad \vr_{\alpha}^{\, \, (0)} \gets \vzero$
            \State $\mm_\alpha = (1+\alpha) \eye - A$.
            \While{$\|\mm_{\alpha} \vr_{\alpha}^{\, \, (k)} - \vones\|_\infty > \frac{1}{2}$ or  $\|\mm_{\alpha}^\top \vl_{\alpha}^{\, \, (k)} - \vones\|_\infty > \frac{1}{2}$}
                \State $\Prightk_{2 \alpha}, \Pleftk_{2 \alpha}  \gets  \code{SolvefromScale}(\mm_{2\alpha}, \ml_{2 \alpha}, \mr_{2\alpha}, \frac{1}{8 K})$
                \State $\vl_{\alpha}^{\, \,(k+1)} \eq \vl_{\alpha}^{\, \, (k)} - \Pleftk_{2 \alpha} (\mm^{\top}_{\alpha} \vl_{\alpha}^{\, \, (k)} - \vones)$ 
                \State $\vr_{\alpha}^{\, \, (k+1)} \eq \vr_{\alpha}^{\, \, (k)} - \Prightk_{2 \alpha} (\mm_\alpha \vr_{\alpha}^{\, \, (k)} - \vones)$
                \State $k \gets k + 1$
            \EndWhile
            \State $\vl_\alpha, \vr_\alpha \eq \vl_\alpha^{\, \, (k)}, \vr_\alpha^{\, \, (k)}$
        \EndWhile
        
    \State \Return $\ml_\alpha,\mr_\alpha$ with $\ml_\alpha \mm_{\epsilon} \mr_\alpha$ being RCDD 
\end{algorithmic}	
\end{algorithm}
With this fact, we move to the second part of our proof. Lemma~\ref{lemma:MSolve} tells us that we can compute operators $\Pright_{2 \alpha}$ and $\Pleft_{2 \alpha}$ given $\mm_{2 \alpha}$ and $(\ml_{2 \alpha}, \mr_{2 \alpha})$ which approximately act as $\mm_{2\alpha}$ and $\mm_{2\alpha}^\top$ respectively. Further, we can ensure the error of our operators is small in any positive definite $\md$-norm. We will use this result to show that $\Pright_{2 \alpha}$ and $\Pleft_{2 \alpha}$ function are adequate preconditioners for solving $\mm_\alpha \vr_\alpha = \vones$ and $\mm_\alpha^\top \vl_\alpha = \vones$ respectively.

\begin{lemma} \label{lemma:MProg}
    Let $\mm$ be an M-matrix. Let $\md$ be any positive definite matrix where $\md^{-1/2} \mm^\top \md^{1/2} + \md^{1/2} \mm \md^{-1/2}$ is positive semidefinite. Let $\gamma \geq \kappa(\md)$. Let $\alpha > 0$ be some parameter. Then if $(\Pright_{2 \alpha}, \Pleft_{2\alpha}) = \code{SolvefromScale}(\mm_{2\alpha}, \ml_{2 \alpha}, \mr_{2\alpha},  \frac{1}{8 \gamma})$, for any vector $\vx$ we have
    \[
    \Exp \left[ \| \vx - \mm_\alpha \Pright_{2 \alpha} (\vx) \|_\md \right] \leq \frac{3}{4} \|\vx\|_\md \quad \text{and} \quad \Exp \left[ \|\vx - \mm_\alpha^\top \Pleft_{2 \alpha} (\vx) \|_{\md^{-1}} \right] \leq \frac{3}{4} \|\vx\|_{\md^{-1}}.
    \]
\end{lemma}
\begin{proof} 
    Observe that by the triangle inequality
    \begin{align*}
        \left\lVert\vx - \mm_\alpha \Pright_{2 \alpha} (\vx)\right\rVert_\md &= \left\lVert \vx  - \mm_{2\alpha} \Pright_{2 \alpha} (\vx) + \mm_{2\alpha} \Pright_{2 \alpha} (\vx) - \alpha \mm_{2\alpha}^{-1} \vx + \alpha \mm_{2\alpha}^{-1} \vx - \mm_\alpha \Pright_{2 \alpha} (\vx) \right\rVert_\md  \\
        &\leq \left\lVert\vx  - \mm_{2\alpha} \Pright_{2 \alpha}(\vx) \right\rVert_\md + \alpha \left\lVert\mm_{2\alpha}^{-1} \vx \right\rVert_\md\\ &\phantom{{}={}} +\left\lVert(\mm_{2\alpha} \Pright_{2 \alpha}(\vx) - \alpha \mm_{2\alpha}^{-1}\vx - \mm_\alpha \Pright_{2 \alpha}(\vx)\right\rVert_\md ~.
    \end{align*}
    We bound the expected value over $\Pright_{2\alpha}$ of each of these three terms in order. We see first that
    \[
    \Exp \left[ \left\lVert\vx - \mm_{2\alpha} \Pright_{2 \alpha}(\vx) \right\rVert_\md \right] \leq \frac{1}{8} \left\lVert\vx\right\rVert_\md.
    \]
    by the definition of $\Pright_{2\alpha}$ and Lemma~\ref{lemma:MSolve}. Further,
    \[
    \Exp \left[ \left\lVert\mm_{2\alpha}^{-1} x \right\rVert_\md \right] =  \left\lVert\mm_{2\alpha}^{-1} x \right\rVert_\md \leq \left\lVert\mm_{2\alpha}^{-1}\right\rVert_\md \left\lVert \vx \right\rVert_\md
    \]
    by the definition of the induced $\md$-norm. Finally, we have
    \begin{align*}
    \Exp \left[ \left\lVert \mm_{2\alpha} \Pright_{2 \alpha}(\vx) - \alpha \mm_{2\alpha}^{-1} \vx - \mm_\alpha \Pright_{2 \alpha}(\vx)\right\rVert_\md \right] &=  \Exp \left[ \left\lVert\alpha \Pright_{2 \alpha}(\vx) - \alpha \mm_{2\alpha}^{-1} \vx \right\rVert_\md \right] \\
    &= \alpha \Exp \left[ \left\lVert\mm_{2 \alpha}^{-1} (\vx - \mm_{2 \alpha} \Pright_{2 \alpha}(\vx))\right\rVert_\md \right] \\
    &\leq \alpha \left\lVert\mm_{2 \alpha}^{-1}\right\rVert_\md   \Exp \left[ \left\lVert\vx - \mm_{2 \alpha} \Pright_{2 \alpha}(\vx)\right\rVert_\md \right] \\
    &\leq \frac{\alpha}{8} \left\lVert\mm_{2 \alpha}^{-1}\right\rVert_\md  \left\lVert\vx\right\rVert_\md,
    \end{align*}
    where we again use Lemma~\ref{lemma:MSolve} in the last step. Taking expectation over $\Pright_{2\alpha}$ in our original fact and applying these three inequalities yields
    \[
    \Exp \left[ \left\lVert \vx -  \mm_\alpha \Pright_{2 \alpha}(\vx)\right\rVert_\md \right] \leq \Big( \frac{1}{8}  +  \frac{9 \alpha}{8} \left\lVert\mm_{2\alpha}^{-1}\right\rVert_\md \Big) \left\lVert \vx\right\rVert_\md.
    \]
    By repeating this analysis for the left scaling, we can further obtain 
    \[
    \Exp \left[ \left\lVert \vx-  \mm_\alpha^\top \Pleft_{2 \alpha}(\vx)\right\rVert_{\md^{-1}} \right] \leq \Big( \frac{1}{8}  +  \frac{9 \alpha}{8} \left\lVert\mm_{2\alpha}^{-\top}\right\rVert_{\md^{-1}} \Big) \left\lVert \vx \right\rVert_{\md^{-1}} = \Big( \frac{1}{8}  +  \frac{9 \alpha}{8} \left\lVert \mm_{2\alpha}^{-1}\right\rVert_\md \Big) \left\lVert \vx \right\rVert_{\md^{-1}}
    \]
    as $\left\lVert\mm_{2\alpha}^{-1}\right\rVert_\md  = \left\lVert\mm_{2\alpha}^{-\top}\right\rVert_{\md^{-1}}$.  Now, note that by Lemma \ref{lem:preconditioner} we have $ \left\lVert\mm_{2\alpha}^{-1}\right\rVert_{\md} \leq \frac{1}{2\alpha}$. Substituting this into the above expressions gives us
    \[
    \Exp \left[ \left\lVert \vx -  \mm_\alpha \Pright_{2 \alpha}(\vx)\right\rVert_\md \right] \leq \Big( \frac{1}{8}  +  \frac{9 \alpha}{8} \left\lVert \mm_{2\alpha}^{-1}\right\rVert_\md \Big) \left\lVert x\right\rVert_\md \leq \Big( \frac{1}{8}  +  \frac{9 \alpha}{16 \alpha} \Big) \left\lVert x\right\rVert_\md  < \frac{3}{4} \left\lVert x\right\rVert_\md.
    \]
    and
    \[
    \Exp \left[ \left\lVert \vx-  \mm_\alpha^\top \Pleft_{2 \alpha}(\vx)\right\rVert_{\md^{-1}} \right] \leq \Big( \frac{1}{8}  +  \frac{9 \alpha}{8} \left\lVert\mm_{2\alpha}^{-1}\right\rVert_\md \Big) \left\lVert x\right\rVert_{\md^{-1}} \leq \Big( \frac{1}{8}  +  \frac{9 \alpha}{16 \alpha} \Big) \left\lVert x\right\rVert_{\md^{-1}} < \frac{3}{4} \left\lVert x\right\rVert_{\md^{-1}}
    \] 
    as claimed.
\end{proof}
This lemma gives us a useful handle on how much we can expect our error to decay in each iteration of the inner loop of our algorithm. To formalize this, we show the following lemma:
\begin{lemma} \label{lemma:scaling_innerloop}
    Let $\md$ be a (potentially unknown) matrix where $\md^{-1/2} \mm^\top \md^{1/2} + \md^{1/2} \mm \md^{-1/2}$ is PD, and let $\gamma \geq \kappa(\md)$. If we choose $K = \frac{1}{8 \gamma}$ in Algorithm~\ref{alg:MMatrix}, we have for some $k = O(\log (n \kappa(\md)))$ that $\|\mm_{\alpha} \vr_{\alpha}^{\, \, (k)} - \vones\|_\infty \leq \frac{1}{2}$ and $\|\mm_{\alpha}^\top \vl_{\alpha}^{\, \, (k)} - \vones\|_\infty \leq \frac{1}{2}$ with high probability with respect to $n$ for any $\alpha > 0$. Each iteration of the while loop to compute each $\vr_{\alpha}^{\, \, (i)}$ and $\vl_{\alpha}^{\, \, (i)}$ can be implemented in time
    \[
    O\Big(\mathcal{T}_{solve}\big(m, n, \kappa(\ml_{2 \alpha} \mm_{2 \alpha} \mr_{2 \alpha}), \frac{1}{8 \gamma \kappa(\mr_{2 \alpha})}\big) \! +\! \mathcal{T}_{solve}\big(m, n, \kappa(\ml_{2 \alpha} \mm_{2 \alpha} \mr_{2 \alpha}), \frac{1}{8 \gamma \kappa(\ml_{2 \alpha})}\big)\Big).
    \]
\end{lemma}
\begin{proof}
    Let us first consider the algorithm's update to $\vr_\alpha$. Observe that  
    \[
    \mm_\alpha \vr_\alpha^{\, \, (k+1)} - \vones = \mm_\alpha \vr_{\alpha}^{\, \, (k)} - \vones - \mm_\alpha \Prightk_{2 \alpha}(\mm_\alpha \vr_{\alpha}^{\, \, (k)} - \vones).
    \]
    With this, let us define $\ve^{\, \,  (k)}_r = \mm_\alpha \vr_{\alpha}^{\, \, (k)} - \vones$ and consider $\mathbb{E} [ \| \ve^{\, \, (k+1)}_r \|_\md ]$ where the expectation is taken over $\Prightk_{2 \alpha}$. By Lemma~\ref{lemma:MProg} it follows that
    \[
    \Exp \left[ \left\lVert \ve^{\, \, (k+1)}_r \right\rVert_\md \right] = \Exp \left[ \left\lVert \ve^{\, \, (k)}_r - \mm_\alpha \Pright_{2 \alpha}(\ve^{\, \, (k)}_r) \right\rVert_\md \right] \leq \frac{3}{4} \left\lVert \ve^{\, \, (k)}_r \right\rVert_\md ~.
    \]
    Thus by inducting on $k$ we obtain 
    \[
    \Exp \left[ \|\ve^{\, \, (k)}_r\|_\md \right] \leq  \Big(\frac{3}{4}\Big)^k \|\ve^{\, \, (0)}_r \|_\md =\Big(\frac{3}{4}\Big)^k \|\vones\|_\md
    \]
    for any $k$, where the expectation is taken over all the $\mP_{2\alpha}^{(r,i)}$ chosen. Now note that if for vectors $\vy, \vz$ we have $\| \vy \|_\md \leq \| \vz \|_\md$ then $\| \vy \|_2 \leq \kappa(\md) \| \vz \|_2$. Applying this to the above yields
    \[
    \Exp \left[ \|\ve^{\, \, (k)}_r\|_\infty \right] \leq \Exp \left[ \|\ve^{\, \, (k)}_r\|_2 \right] \leq \kappa(\md) \Big(\frac{3}{4}\Big)^k \|\vones \|_2 = \sqrt{n} \kappa(\md) \Big(\frac{3}{4}\Big)^k.
    \]
    Applying the same chain of reasoning to $\vl_\alpha$ with $\mP^{(l, i)}_{2 \alpha}$ in the $\md^{-1}$ norm allows us to similarly conclude
    \[
    \Exp \left[ \|\ve^{\, \,(k)}_l\|_\infty \right] \leq \Exp \left[ \|\ve^{\, \,(k)}_l\|_2 \right] \leq \kappa(\md) \Big(\frac{3}{4}\Big)^k \|\vones\|_2 =  \sqrt{n} \kappa(\md) \Big(\frac{3}{4}\Big)^k
    \]
    by Lemma~\ref{lemma:MProg}. With the choice of $k$ given above, both of these quantities are at most $\frac{1}{\poly(n)}$. With this, we have 
    \[
    \Pr( \|\ve^{\, \,(k)}_r \|_\infty > 1/2 ) \leq  2 \Exp \left[ \|\ve^{\, \, (k)}_r \|_\infty \right]\leq \frac{1}{\poly(n)} 
    \]
    \[
    \Pr( \|\ve^{\, \, (k)}_l \|_\infty > 1/2 ) \leq  2 \Exp \left[ \|\ve^{\, \, (k)}_l \|_\infty \right]\leq \frac{1}{\poly(n)} 
    \]
    by Markov's inequality-- by the union bound we thus have both $\|\ve^{\, \, (k)}_r \|_\infty \leq 1/2$ and $\|\ve^{\, \, (k)}_l \|_\infty \leq 1/2$ with probability at least $1 - \frac{1}{\poly(n)}$. The running time per iteration follows trivially from Lemma~\ref{lemma:MSolve}: to compute $\vr_{\alpha}^{\, \, (i+1)}, \vl_{\alpha}^{\, \, (i+1)}$ from  $\vr_{\alpha}^{\, \, (i)}, \vl_{\alpha}^{\, \, (i)}$ we need to compute $\Prightk_{2\alpha}$ and $\Pleftk_{2\alpha}$, apply each of these operators to a single vector, and do $O(m)$ additional work in matrix-vector products. This can clearly be implemented in the desired running time.
\end{proof}

Now, we are ready to prove the main theorem of this section. 

\begin{theorem}
\label{thm:scaling}
    Let $\ma \in \R^{\nn}$ be a nonnegative matrix with $m$ nonzero entries and $\rho(\ma) < 1$, and let $\mm = \eye - \ma$. Pick $\epsilon > 0$, and let $K \geq \max (\|\mm^{-1}\|_\infty, \|\mm^{-1}\|_1)$ be some known parameter. Then, the algorithm \fullmscale$(\ma, \epsilon, K)$ computes a pair of diagonal matrices $(\ml, \mr)$ where $\ml((1+\epsilon)\eye - \ma)\mr$ is RCDD with high probability. This algorithm runs in time
    \[
    O\Big(\mathcal{T}_{solve}\Big(m, n, 18 K^2, \frac{1}{216 K^3}\Big) \log(n K) \log\Big(\frac{\|\ma\|_1 + \|\ma\|_\infty}{\epsilon} \Big) \Big).
    \]
\end{theorem} 
\begin{proof}
    We begin by proving the algorithm's correctness. We observe that at the start of the algorithm, we consider $\mm_\alpha$ with a large value of $\alpha$ to make it RCDD with a trivial scaling by Lemmas \ref{lemma:init_scaling} and \ref{lemma:MMatrix_RCDD}. In addition, at every level of $\alpha$ we hit in the algorithm we use a valid scaling of $\mm_{2\alpha}$ to end up with a pair of vectors $\vl_\alpha, \vr_\alpha$ where $\|\mm_\alpha \vr_\alpha - \vones \|_\infty \leq \frac{1}{2}$ and $\|\mm_\alpha^\top \vl_\alpha - \vones \|_\infty \leq \frac{1}{2}$. Thus, we are able to find $\vl_\alpha, \vr_\alpha$ where for any $i \in [n]$ we have $(\mm_\alpha \vr_\alpha)_i \in [\frac{1}{2}, \frac{3}{2}]$ and $(\mm_\alpha^\top \vl_\alpha)_i \in [\frac{1}{2}, \frac{3}{2}]$. By Lemma~\ref{lemma:MMatrix_RCDD} we therefore see $\ml_\alpha \mm_\alpha \mr_\alpha$ is an RCDD matrix. Therefore when the algorithm terminates, we return a pair of vectors $\vl_{\alpha^*}, \vr_{\alpha^*}$ where $\ml_{\alpha^*} \mm_{\alpha^*} \mr_{\alpha^*}$ is RCDD. But since $\alpha^* \leq \epsilon$ we see $\ml_{\alpha^*} \mm_{\epsilon} \mr_{\alpha^*} = \ml_{\alpha^*} (\mm_{\alpha^*}+ (\epsilon-\alpha^*)\eye) \mr_{\alpha^*}= \ml_{\alpha^*} \mm_{\alpha^*}  \mr_{\alpha^*} + (\epsilon- \alpha^*)\ml_{\alpha^*} \mr_{\alpha^*}$: this is RCDD since $\epsilon - \alpha^* \geq 0$, $\ml_{\alpha^*}$ and $\mr_{\alpha^*}$ are both positive by Lemma~\ref{lemma:MMatrix_monotone}, and the fact that the sum of two RCDD matrices is RCDD.

    We now bound the running time. To do this, we first bound the cost of a single phase of going from scalings of $\mm_{2\alpha}$ to $\mm_{\alpha}$. As alluded to above, by Lemma~\ref{lemma:init_scaling} and the invariant that our algorithm maintains we can assume inductively that $(\mm_{2 \alpha} \vr_{2 \alpha})_i \in [\frac{1}{2}, \frac{3}{2}]$ and $(\mm_{2 \alpha}^\top \vl_{2 \alpha})_i \in [\frac{1}{2}, \frac{3}{2}]$ for any $i$. Define $\vr_0 = \mm_0^{-1} \vones$, $\vl_0 = \mm_0^{-\top} \vones$, and $\md = \ml_0 \mr_0^{-1}$. We observe that $\mm$ is an M-matrix and that $\ml \mm \mr$ is RCDD by Lemma~\ref{lemma:MMatrix_RCDD}: thus by Lemma~\ref{lem:pdscaling} we conclude that $\md^{-1/2} \mm_0^\top \md^{1/2} + \md^{1/2} \mm_0 \md^{-1/2} \succeq \mzero$. By Lemma~\ref{lemma:scaling_condition}, we see that 
    \begin{itemize}
    \item $\kappa(\ml_{2\alpha} \mm_{2\alpha} \mr_{2\alpha}) \leq 18 ||\mm^{-1}||_\infty ||\mm^{-1}||_1 \leq 18 K^2$ 
    \item $\kappa(\ml_{2\alpha}) \leq 3 ||\mm^{-1}||_1 \leq 3 K $
    \item $\kappa(\mr_{2\alpha}) \leq 3 ||\mm^{-1}||_\infty  \leq 3 K$
    \item $\kappa(\md) \leq 9 ||\mm^{-1}||_\infty ||\mm^{-1}||_1 \leq 9 K^2$.
    \end{itemize}
    We note that $9 K^2 \geq \kappa(\md)$, and by Lemma~\ref{lemma:scaling_innerloop} we can compute $\vl_\alpha, \vr_\alpha$ satisfying $(\mm_\alpha \vr_\alpha)_i \in [\frac{1}{2}, \frac{3}{2}]$ and $(\mm_\alpha^\top \vl_\alpha)_i \in [\frac{1}{2}, \frac{3}{2}]$ in $O(\log(n \kappa(\md)))$ iterations each costing 
    \[
    O\Big(\mathcal{T}_{solve}\big(m, n, \kappa(\ml_{2 \alpha} \mm_{2 \alpha} \mr_{2 \alpha}), \frac{1}{72 K^2 \kappa(\mr_{2 \alpha})}\big) \! +\! \mathcal{T}_{solve}\big(m, n, \kappa(\ml_{2 \alpha} \mm_{2 \alpha} \mr_{2 \alpha}), \frac{1}{72 K^2 \kappa(\ml_{2 \alpha})}\big)\Big).
    \]
    time. Plugging in our condition number bounds and simplifying yields that in time
    \[
    O\Big(\mathcal{T}_{solve}\big(m, n, 18 K^2, \frac{1}{216 K^3}\big) \log(n K) \big)
    \]
    we can go from our scalings of $\mm_{2\alpha}$ to scalings of $\mm_\alpha$. Since $\alpha$ is halved only $O(\log(\frac{\|\ma\|_1 + \|\ma\|_{\infty}}{\epsilon}))$ times, the result follows.
\end{proof}

\begin{algorithm} [t]
    \caption {\sc $\code{SolveM}(\ma, \epsilon, K)$} \label{alg:SolveM}    
\begin{algorithmic}[1]
    \State \textbf{Input:} $\ma\in \R^{n\times n}$ is an M-matrix, $K$ some parameter satisfying $K \geq \max \{ \|(\eye - \ma)^{-1}\|_\infty, \|(\eye - \ma)^{-1}\|_1 \}$, $\epsilon >0$ is a real number.
    \State \textbf{output:} Operators P satisfying Theorem~\ref{lemma:MSolve}.
    \State $\ml, mr \gets $ \fullmscale$(\ma, \epsilon, K)$;  
    \State $\ms \gets \ml ((1+\epsilon) \eye -  \ma) \mr$;
    \State $\opZ^{(r)} \gets \code{RCDDSolve}(\ms, \frac{\delta}{\kappa(\ml)})$ generated by Oracle~\ref{oracle:RCDDSolve};
    \State \Return operators $\Pright(\vx) = \mr \opZ^{(r)}(\ml \vx)$
\end{algorithmic}	
\end{algorithm}

Our main result for scaling M-matrices follows as a corollary. We restate it for convenience:
\mscaler*

\begin{proof}
Observe that we can assume $s=1$ without loss of generality. By using Theorem \ref{thm:rcdd_solver} inside of Theorem \ref{thm:scaling}, we can replace the $\mathcal{T}_{solve}$ with $\otilde(m)$: the claim follows immediately.
\end{proof}
We also demonstrate that this scaling procedure also gives an efficient algorithm for solving M-matrices and thereby proves Theorem \ref{thm:msolve}. We restate it as well.
\msolve*
\begin{proof}
Consider Algorithm \ref{alg:SolveM} above. By Theorems \ref{thm:mscale} and \ref{thm:rcdd_solver}, we see that $(\ml, \mr)$ is a valid scaling of $((1+\epsilon) s \eye -  \ma)$ and thus $\code{RCDDSolve}$ generates a solver for $\ms$. The theorem follows by relatively straightforward algebra analogous to Theorem \ref{lemma:MSolve}.
\end{proof}

This concludes the proof of our algorithm for computing RCDD scalings of M-matrices. In many settings the norm of $\mm^{-1}$ is fairly small and can be reasoned about in a relatively straightforward fashion. However, in the next section on computing Perron vectors, we will be computing scalings for M-matrices which are nearly singular. Thus, we will want a bound on our algorithm's running time which is hopefully independent on how ill-conditioned the matrix $\mm$ is. Fortunately, we can obtain such a bound as a straightforward corollary of our main result:
\begin{corollary} \label{corollary:scale_eigenvec}
    Let $\ma \in \R^{\nn}$ be a nonnegative matrix with $m$ nonzero entries and $\rho(\ma) < 1$. Pick $\epsilon \in [0,1]$, and let $\mm_{\epsilon/2} = (1+ \frac{\epsilon}{2})\eye - \ma$. Let $\eigl$ and $\eigr$ denote the left and right top eigenvectors of $\ma$, respectively. Let $\gamma \geq \kappa(\eigl) + \kappa(\eigr)$. Then the algorithm \fullmscale$(\frac{\ma}{1+\epsilon/2}, \epsilon/3, \frac{2 \gamma}{\epsilon})$ generates a pair of diagonal matrices $(\ml, \mr)$ where $\ml((1+\epsilon)\eye - \ma)\mr$ is RCDD with high probability. This procedure runs in time
    \[
    O\left(\mathcal{T}_{solve}\left(m, n, \frac{72 \gamma^2}{\epsilon^2}, \frac{\epsilon^3}{1728 \gamma^3}\right) \log\left(n \frac{\gamma}{\epsilon}\right) \log\left(\frac{\|\ma\|_1 + \|\ma\|_\infty}{\epsilon} \right) \right).
    \]
\end{corollary}
\begin{proof}
    By the previous theorem, \fullmscale$(\frac{\ma}{1+\epsilon/2}, \epsilon/3, K)$ generates a pair of diagonal matrices $(\ml, \mr)$ where $\ml((1+\frac{\epsilon}{3})\eye - \frac{\ma}{1+\epsilon/2})\mr$ is RCDD. Now, multiplying by a constant will not change tha RCDDness of the resulting matrix, and since $(1+\frac{\epsilon}{3})(1+\frac{\epsilon}{2}) \leq  1 + \epsilon$ for $\epsilon \in [0,1]$ we can conclude that $(\ml, \mr)$ is also a valid scaling for  $\ml((1+\epsilon)\eye - \ma)\mr$ with high probability. We must now bound the running time. We observe by Lemma \ref{lemma:eigvec_scaling} that 
    \[
    \|\mm_{\epsilon/2}^{-1}\|_\infty \leq \frac{2\gamma}{\epsilon} \quad \text{and} \quad \|\mm_{\epsilon/2}^{-1}\|_1 \leq \frac{2\gamma}{\epsilon}.
    \]
    Thus, $\frac{2\gamma}{\epsilon}$ functions as a valid choice of $K$ in our previous theorem-- plugging this into our previous running time and simplifying gives the result.
\end{proof}

A simpler algorithm and analysis can be derived in the case of symmetric M-matrices. We assume access to the following oracle for solving SDD linear systems:
\begin{restatable}[SDD Solver]{oracle}{SDDSolve}
    \label{oracle:SDDSolve}
    Let $\ms \in \R^{n \times n}$ be an SDD matrix, and let $x \in \R^n$ be a vector. Then there is a procedure called $\code{SDDSolve}(\ms, \epsilon)$ which can generate a linear operator $\opZ: \R^n \to \R^n$ where with high probability
    \[
    \| \ms^{-1} \vx - \opZ (\vx)\|_\ms \leq \epsilon \| \ms^{-1} \vx \|_\ms     
    \]
    Further, we can compute and apply $\opZ$ to vectors in $\mathcal{T}_{SDDsolve}(m, n, \epsilon)$ time. 
\end{restatable}
With this, we prove the following result.
\begin{restatable}[Symmetric M-matrix Scaling]{theorem}{symmscaling}
\label{thm:symmscaling}
    Let $\ma \in \R^{\nn}$ be a symmetric nonnegative matrix with $m$ nonzero entries and let $s > 0$ be such that $\rho(\ma) < s$. Let $\mm = s \eye - \ma$. Pick $\epsilon > 0$. Then \fullsymmscale$(\ma, \epsilon, K)$ computes a diagonal matrix $\mv$ where $\mv((1+\epsilon)s\eye - \ma)\mv$ is SDD with high probability. This algorithm runs in time
    \[
    O\Big(\mathcal{T}_{SDDsolve}\Big(m, n, \frac{1}{4} \Big) \log(n \kappa(\mm)) \log\Big(\frac{1}{\epsilon} \Big) \Big).
    \]
\end{restatable}

In the interest of brevity, we defer our analysis of this procedure to Section \ref{sec:symmscaling}.
\section{Finding Top Eigenvalue and Eigenvector of Positive Matrices}\label{sec:perron-alg}
In this section we present one of the fundamental results of our paper which is the first nearly linear time algorithm for computing the top eigenvalue and eigenvectors of non-negative matrices characterized by Perron-Frobenius theorem. 
Here we discuss how to use our M-matrix scaling algorithm to come up with nearly linear time algorithms for the Perron problem and M-matrix decision problem.
As we argued previously, these two problems are essentially equivalent. We use the M-matrix scaling algorithm from the previous section to solve the Perron problem, and then we use our reduction combined with some ideas from Algorithm~\ref{alg:MMatrix-overview} to come up with a nearly linear time algorithm for the Perron problem. The following theorem gives the statement that we will prove in this section.

\begin{theorem}[Nearly Linear Time Perron Computation] \label{thm:perron-overview}
	Given a non-negative irreducible matrix $\ma\in \R_{\geq 0}^{n\times n}$, and $\delta>0$, Algorithm~\ref{alg:perron} finds real number $s>0$, and positive vectors $\vl,\vr,\in \R_{>0}^{n}$ such that $(1-\delta)\rho(\ma) < s \leq \rho(\ma)$, $\|(s \eye - \mm) \vr\|_\infty \leq \frac{\delta}{K^2}\norm{\vr}_\infty$, and $\|\vl^\top (s \eye - \mm)\|_\infty \leq \frac{\delta}{K^2}\norm{\vl}_\infty$ in time 
	\[
	O\Big(\mathcal{T}_{solve}\Big(m, n, O(\frac{K^6}{\delta^2}), O(\frac{\delta^3}{K^9}) \Big) \log(n \frac{K}{\delta}) \log\Big(\frac{\|\ma\|_1 + \|\ma\|_\infty}{\epsilon} \Big)  \log(\frac{|s_1-s_2|}{\epsilon\rho(\ma)}) \Big)~,		
	\]	
	where $K = \Theta(\kappa(\eigl)+\kappa(\eigr))$ and $\eigl$, $\eigr$ are the left and right eigenvectors of $\ma$, of eigenvalue  $\rho(\ma)$. 
\end{theorem}

\subsection{The M-matrix Decision Problem}\label{sec:mmatrix-decision}
	Here we show how a slightly modified version of Algorithm~\ref{alg:decision} can be used to solve the M-matrix decision problem. Note that if Algorithm~\ref{alg:decision} succeeds in finding a scaling for $\mm_\epsilon$, we have a proof that $\rho(\ma)< 1+\epsilon$. However, Algorithm~\ref{alg:decision} only works on the premise that $\rho(\ma) < 1$. If this was false, we are no longer able to guarantee that i) our scalings of $\mm_\alpha$ are always positive, ii) the iteration used to compute scalings of $\mm_\alpha$ converges in the claimed number of iterations and iii) we can actually solve the intermediate RCDD systems in the claimed running time. To be more precise, Algorithm \ref{alg:decision} uses some previously obtained RCDD scaling of $\mm_{2 \alpha}$ to find a scaling of $\mm_\alpha$. If $\mm$ was not an M-matrix however (or equivalently if $\rho(A) \geq 1$) we would be unable to use Lemma \ref{lemma:scaling_innerloop} and thus we would have no upper bound on the number of iterations the inner loop of Algorithm \ref{alg:decision} takes to converge. Even worse, we also have no guarantee that $\mm_\alpha$ even admits an RCDD scaling, and thus even if our inner loop did converge we cannot say that it gives a valid scaling. In addition, without the guarantee that $\mm$ is an M-matrix we do not have any bounds on $\|\mm_\alpha^{-1}\|_\infty$ or $\|\mm_\alpha^{-1}\|_1$: thus we are unable to ensure that each RCDD solve runs in the time we need it to to.

	Fortunately, we observe that both of these failure modes for our algorithm can be efficiently checked for. If we want to verify failure mode  i), we can simply check if our scaling vectors $\vr_\alpha$ and $\vl_\alpha$ are not entrywise positive. If we want to check failure mode ii), we can simply count the number of iterations in the inner loop of Algorithm \ref{alg:decision} and ensure that it does not exceed Lemma \ref{lemma:scaling_innerloop}'s guarantee. Finally, to check failure mode iii) we can simply ensure that each RCDD solve does not take too long. This way we can come up with a nearly linear time algorithm which either reports that $\mm$ is not an M-matrix, or provides a proof that $\mm_\epsilon$ is an M-matrix. 
	
	\paragraph{A note on not knowing the condition number of top eigenvectors of $\ma$.} The last paragraph gives the outline of our strategy to solve M-matrix decision problem. However, there is one more slight complication that arises when using the M-matrix scaling algorithm from the previous section. In order to guarantee the correctness of the output of our algorithm, we require knowledge of some parameter $K$ which upper bounds the condition numbers of the left and right eigenvectors of $\ma$. This too can be circumvented, but we will demonstrate our fix for this problem at the end: in the meanwhile, we assume access to some $K \geq \kappa(\eigl) + \kappa(\eigr)$.
	
	Given the above discussion, we are now ready to present our algorithm for the M-matrix decision problem. Algorithm~\ref{alg:decision} gets as input an irreducible non-negative matrix $\ma\in R_{\geq 0}^{n\times n}$, real number $\epsilon > 0$, and $K \in \R_{>0}$ which upper bounds the condition numbers of the left and right eigenvectors of $\ma$. It either reports that $\mm = \eye -\ma$ is not an M-matrix or finds scalings $\vl, \vr \in \R_{\geq 0}^{n}$ such that $\diag(\vl) \mm_\epsilon \diag(\vr)$ is RCDD. Algorithm~\ref{alg:decision} is essentially identical to Algorithm~\ref{alg:MMatrix-overview} with the modifications that we return that $\mm$ is not an M-matrix whenever our inner loop iteration takes too long to converge or the scaling vectors we compute are not positive. The following lemma provides a formal specification of Algorithm~\ref{alg:decision}.

\begin{algorithm} [t]
	\caption {\sc \mdecide$(\ma, \epsilon, \gamma)$} \label{alg:decision}    
	\begin{algorithmic}[1]        
	\State \textbf{Input:} $A\in \R_{\geq 0}^{n\times n}$ is an irreducible matrix with $\rho(A)<1$, and $\epsilon, \gamma >0$ are real numbers.
	\State \textbf{Output:} If $\mm =\eye-\ma$ is an M-matrix then returns $\vl, \vr \in \R_{>0}^n$ such that $\mm_\epsilon \vr >0$ and $\vl^\top \mm_\epsilon > 0$, otherwise returns $\vl,\vr =-\vones, -\vones$.
	\State $\alpha \eq 2 \max\{\|A\|_1, \|A\|_\infty\}$
	\State $\vl_\alpha \gets \frac{1}{\alpha} \vones, \quad \vr_\alpha \gets \frac{1}{\alpha} \vones$
	\While{$\alpha > \epsilon$}
	\State $\alpha \eq \alpha/2, \quad  k \gets 0$
	\State $\vl_{\alpha}^{\, \, (0)} \gets \vzero, \quad \vr_{\alpha}^{\, \, (0)} \gets \vzero$
	\State $\mm_\alpha = (1+\alpha) \eye - A$.
	\State $\mathcal{T}_{RCDD} = O\Big(\mathcal{T}_{solve}\Big(m, n, 18 \gamma^2, \frac{1}{216 \gamma^3}\Big)\Big)$
	\For{$k = 1, 2, ... O(n\gamma) $}
	\State $\Prightk_{2 \alpha}, \Pleftk_{2 \alpha}  \gets  \code{SolvefromScale}(\mm_{2\alpha}, \ml_{2 \alpha}, \mr_{2\alpha}, \delta)$
	\State $\vl_{\alpha}^{\, \,(k+1)} \eq \vl_{\alpha}^{\, \, (k)} - \Pleftk_{2 \alpha} (\mm^{\top}_{\alpha} \vl_{\alpha}^{\, \, (k)} - \vones)$ 
	\State $\vr_{\alpha}^{\, \, (k+1)} \eq \vr_{\alpha}^{\, \, (k)} - \Prightk_{2 \alpha} (\mm_\alpha \vr_{\alpha}^{\, \, (k)} - \vones)$
	\If {applying $\Pleftk_{2 \alpha}$ or $\Prightk_{2 \alpha}$ takes longer than $\mathcal{T}_{RCDD}$ time}
	\State \Return $-\vones,-\vones$ as a sign that $ \mm = \eye - \ma$ is not an M-matrix.
	\EndIf
	\EndFor
	\State $\vl_\alpha, \vr_\alpha \eq \vl_\alpha^{\, \, (k)}, \vr_\alpha^{\, \, (k)}$
	\If{$\frac{3}{2} \vones > \mm_\alpha \vr_\alpha > \frac{1}{2} \vones, \frac{3}{2} \vones > \mm_\alpha^\top \vl_\alpha > \frac{1}{2} \vones, \vr_\alpha > \vzero,$ and $\vl_\alpha > \vzero$}
	\State continue;
	\Else 
	\State \Return $-\vones,-\vones$ as a sign that $\mm = \eye - \ma$ is not an M-matrix.
	\EndIf
	\EndWhile
	\State \Return $\ml_\alpha,\mr_\alpha$ with $\ml_\alpha \mm_{\epsilon} \mr_\alpha$ being RCDD 
	\end{algorithmic}	
\end{algorithm}
 
\begin{lemma}\label{lem:decision}
	Given a non-negative irreducible matrix $\ma\in \R_{\geq 0}^{n\times n}$, and $\epsilon,K \in \R_{>0}$, let $\eigl, \eigr\in \R_{>0}^{n}$ be the left and right Perron vectors of $\ma$, respectively. Given that $K\geq \kappa(\eigl)+\kappa(\eigr)$ the algorithm {\sc \mdecide($\frac{\ma}{1+\epsilon/2}$, $\epsilon/3$, $\frac{2K}{\epsilon}$)} either comes up with a proof that $\mm = \eye-\ma$ is not an M-matrix, or proves $\mm_\epsilon = \mm + \epsilon \eye$ is an M-matrix by finding $\vl, \vr >0$ such that $\ml \mm_\epsilon \mr$ is RCDD. This procedure runs in time
	\[
    O\Big(\mathcal{T}_{solve}\Big(m, n, \frac{72 K^2}{\epsilon^2}, \frac{\epsilon^3}{1728 K^3}\Big) \log(n \frac{K}{\epsilon}) \log\Big(\frac{\|\ma\|_1 + \|\ma\|_\infty}{\epsilon} \Big) \Big).
	\]
\end{lemma}

\begin{proof}
	This algorithm is just Algorithm \ref{alg:MMatrix} with some added tests to ensure we are not in one of the three failure modes of our algorithm. We note that the choice of constants here matches that of Corollary \ref{corollary:scale_eigenvec}. If our algorithm passes all of the checks on lines 14 and 19, we are able to guarantee that we have a valid scaling for every $\mm_\alpha$ seen during the course of our algorithm. Further, we can guarantee that computing a scaling for a given $\alpha$ value takes only $O(\mathcal{T}_{RCDD} \log(n \frac{K}{\epsilon}))$ time since we upper bound the number of RCDD solves and the time that each solve takes. Thus, since there are $\log\Big(\frac{\|\ma\|_1 + \|\ma\|_\infty}{\epsilon} \Big) \Big)$ values of $\alpha$ considered, we obtain a scaling of $(1 + \epsilon/3) \eye - \frac{\ma}{1+\epsilon/2}$ (and consequently a scaling of $(1+\epsilon)\eye - \ma$) in time
	\[
		O\Big(\mathcal{T}_{solve}\Big(m, n, \frac{72 K^2}{\epsilon^2}, \frac{\epsilon^3}{1728 K^3}\Big) \log(n \frac{K}{\epsilon}) \log\Big(\frac{\|\ma\|_1 + \|\ma\|_\infty}{\epsilon} \Big) \Big).
	\] 
	Now to complete the proof we need to argue that if our algorithm fails one of its internal checks then $\eye - \ma$ is not an M-matrix. If we fail the check on line 19, there are two cases. If we fail either $\frac{3}{2} \vones > \mm_\alpha \vr_\alpha > \frac{1}{2} \vones$ or $\frac{3}{2} \vones > \mm_\alpha^\top \vl_\alpha > \frac{1}{2} \vones$, then we contradict Lemma \ref{lemma:scaling_innerloop}'s bound on the number of iterations needed to get this level of precision assuming $\eye - \ma$ was an M-matrix. If instead we fail $\vr_\alpha > \vzero,$ or $\vl_\alpha > \vzero$ then we have a vector which demonstrates that $\mm_\alpha$ is not monotone: thus $\mm_\alpha$ and consequently $\eye - \ma$ is not an M-matrix.  If we fail the check on line 14, then we contradict the running time guarantee of $\code{SolvefromScale}$ with the condition number guarantees provided by Lemmas \ref{lemma:scaling_condition} and \ref{lemma:mmatrix_to_eigenvector}. Thus again we have a disproof that $\eye - \ma$ is an M-matrix, and thus the result is proved.
	\end{proof}

\subsection{Perron problem}\label{sec:perron-sub}
	Given that we have an algorithm for the M-matrix decision problem we are ready to present our algorithm for the Perron problem. We first present an algorithm that computes an $\epsilon$ approximation of the largest eigenvalue of an irreducible matrix $\ma\in \R_{\geq 0}^{\nn}$, assuming that we know some $K \geq \kappa(\eigl) + \kappa(\eigr)$ where $\eigl$ and $\eigr$ are the left and right top eigenvectors of $\ma$. We then use this as a subroutine to give a general algorithm which finds approximate top eigenvalue and eigenvectors simultaneously without using any assumption about $K$.

\begin{algorithm} [t]
	\caption{This algorithm finds the top eigenvalue of a non-negative irreducible matrix.} \label{alg:perronvalue}
	\begin{algorithmic}[1]
		\Function{\findperronvalue}{$\ma$, $s_1$, $s_2$, $\epsilon$, $K$} 
		\State \textbf{Input:} $\ma\in \R_{\geq 0}^{n\times n}$ is an irreducible matrix, $s_1,s_2, K, \epsilon \in \R_{>0}$.
		\State \textbf{Requirements:} $\epsilon < 1$, $s_1 < \rho(\ma) \leq s_2$, $K \geq \kappa(\eigl)+\kappa(\eigr)$
		\If{$(1+\epsilon/2)s_1 \geq s_2$ }
		\State \Return $s_2$
		\EndIf
		\State $s_m \eq \frac{s_1+s_2}{2}$
		\State $\delta = \frac{1}{2} \frac{s_2 - s_1}{s_2 + s_1}$
		\State $\vl, \vr \eq$ {\sc \mdecide}($\frac{\ma}{s_m(1+\delta/2)}$, $\delta/3$, $\frac{2K}{\delta}$)
		\If{$\vl,\vr < \vzero$}
		\State \Return {\sc \findperronvalue}($\ma$, $s_m$, $s_2$, $\epsilon$, $K$)
		\Else
		\State \Return {\sc \findperronvalue}($\ma$, $s_1$, $(1+\epsilon/2)s_m$, $\epsilon$, $K$)
		\EndIf
		\EndFunction
	\end{algorithmic}	
\end{algorithm}

\begin{lemma}\label{lem:perrron-find}
	Given a non-negative irreducible matrix $\ma\in \R_{\geq 0}^{\nn}$, and $s_1,s_2,\epsilon,K \in \R_{\geq 0}$, let $\eigl, \eigr \in \R_{>0}^n$ be the left and right eigenvectors of $\ma$. Given $K \geq \kappa(\eigl)+\kappa(\eigr)$, $\rho(\ma) \in (s_1,s_2]$, and $\epsilon < 1/2$, {\sc \findperronvalue($\ma$, $s_1$, $s_2$, $\epsilon$, $K$)} returns $s\in \R_{>0}$, such that $\rho(\ma) \leq s < (1+\epsilon)\rho(\ma)$ in time 
	\[
	O\Big(\mathcal{T}_{solve}\Big(m, n, \frac{7200 K^2}{\epsilon^2}, \frac{\epsilon^3}{1728000 K^3}\Big) \log(n \frac{K}{\epsilon}) \log\Big(\frac{\|\ma\|_1 + \|\ma\|_\infty}{\epsilon} \Big)  \log(\frac{|s_1-s_2|}{\epsilon\rho(\ma)}) \Big)	
	\]
\end{lemma}

\begin{proof}
	{\sc \findperronvalue} is a recursive function. We show in each recursive step $|s_1-s_2|$ is decreased by a multiplicative factor, while the requirement that  $\rho(\ma) \in (s_1,s_2]$ is preserved. By invoking {\sc \mdecide}($\frac{\ma}{s_m(1+\delta/2)}$, $\delta/3$, $\frac{2K}{\delta}$) we get one of the following outcomes by Lemma \ref{lem:decision}:
	\begin{itemize}
		\item $\eye - \frac{\ma}{s_m}$ is not an M-matrix: This means $\rho(\ma) \in (s_m, s_2]$.
		\item $(1+\delta) \eye - \frac{\ma}{s_m}$ is an M-matrix: This means $\rho(\ma) \in (s_1, (1+\delta)s_m]$.
	\end{itemize}
	Note that 
	$$
	\max\{|s_m-s_2| , |s_1-(1+\delta)s_m|\} = \max\{\frac{|s_2 - s_1|}{2}, \frac{|s_2 - s_1|}{2} + \delta \frac{s_1 + s_2}{2} \} \leq \frac{3}{4} |s_2 - s_1|
	$$
	since $\delta$ was chosen to be $\frac{1}{2} \frac{s_2 - s_1}{s_2 + s_1}$. Thus in every iteration we cut the size of our interval by a constant factor. To finish, we will need an upper bound on the number of iterations needed before terminating as well as a lower bound on $\delta$ since it is the precision used inside of the {\sc \mdecide} calls. To do this, we make two simple observations. First, if we ever make a call where $s_2 \leq (1+\epsilon) s_1$, then we have $\rho(\ma) \leq s_2 \leq (1+\epsilon) s_1 \leq (1+\epsilon) \rho(\ma)$: we can (and do) return $s_2$ as our approximate spectral radius. Second, if we call the algorithm with $s_2 - s_1 \leq \frac{\epsilon}{2} \rho(\ma)$, we observe that $\rho(\ma) \leq s_2 = s_1 + |s2 - s1| \leq s_1 + \frac{\epsilon}{2} \rho(\ma)$ and thus $\frac{s_1}{1-\epsilon/2} \geq \rho(\ma)$. But now 
	\[
	s_2 \leq s_1 + \frac{\epsilon}{2} \rho(\ma) \leq s_1 + \frac{\epsilon s_1}{2 - \epsilon} \leq (1 + \epsilon) s_1
	\]
	and by the previous observation $s_2$ is an $\epsilon$-approximate largest eigenvalue for $\ma$. With this insight as well as the previous guarantee on the interval size we can immediately bound the number of iterations before terminating by $O(\log(\frac{|s_1-s_2|}{\epsilon\rho(\ma)}))$: after this many steps we can guarantee the intervals we are considering have size at most $\frac{\epsilon}{2} \rho(\ma)$ and we consequently return our upper bound as a valid estimator. 
	
	To bound $\delta$, we note that we can freely assume that $s_2 - s_1 > \frac{\epsilon}{2} \rho(\ma)$: if otherwise we wouldn't even attempt to run {\sc \mdecide}. With this observation, we split our analysis into cases. If $s_2 - s_1 > \frac{1}{5} (s_2 + s_1)$, then we immediately see that 
	\[
	\delta = \frac{1}{2} \frac{s_2 - s_1}{s_2 + s_1} > \frac{1}{10} \frac{s_2 + s_1}{s_2 + s_1} = \frac{1}{10}.
	\]
	If instead $s_2 - s_1 \leq \frac{1}{5} (s_2 + s_1)$, then we observe by rearrangement that $\frac{4}{5} s_2 \leq \frac{6}{5} s_1$ and therefore $s_1 + s_2 \leq \frac{5}{2} s_1 \leq \frac{5}{2} \rho(\ma)$. With this, we get
	\[
	\delta = \frac{1}{2}  \frac{s_2 - s_1}{s_2 + s_1}  > \frac{\epsilon}{4} \frac{\rho(\ma)}{s_2 + s_1} > \frac{\epsilon}{10} \frac{\rho(\ma)}{\rho(\ma)}. = \frac{\epsilon}{10}.
	\] 
	Thus $\delta \geq \frac{\epsilon}{10}$. Since each call to {\sc \mdecide} requires 
	\[
	O\Big(\mathcal{T}_{solve}\Big(m, n, \frac{72 K^2}{\delta^2}, \frac{\delta^3}{1728 K^3}\Big) \log(n \frac{K}{\epsilon}) \log\Big(\frac{\|\ma\|_1 + \|\ma\|_\infty}{\epsilon} \Big) \Big)	
	\]
	time by Lemma \ref{lem:decision}, substituting in our bound for $\delta$ as well as our iteration cost yields the claimed running time.
\end{proof}

\subsubsection{Approximate top eigenvectors}\label{sec:approx-pvector}

So far we have shown how to compute the largest eigenvalue of a non-negative irreducible matrix $\ma$ within a multiplicative error of $\epsilon$. Next, we show how to use our developed machinery to get an approximate eigenvector of $\ma$ corresponding to this largest eigenvalue (See Definition~\ref{def:approx-ev} for approximate eigenvector). 

The following lemma shows that if we are able to approximately solve linear systems in M-matrices, we can generate approximate top eigenvectors. 

\begin{lemma}\label{lem:perron-approx}
	Given a non-negative irreducible matrix $\ma\in \R_{\geq 0}^{n\times n}$ and $\epsilon \in \R_{>0}$, let $\vw \in \R^n$ be a vector where $\vw_i \in [\frac{1}{4}, \frac{7}{4}]$ for all $i$ and define $\vr = [(1+\epsilon)\eye - \frac{\ma}{\rho(\ma)}]^{-1} \vw$. We have $\|(\eye-\frac{\ma}{\rho(\ma)})\vr \|_\infty \leq 8 \epsilon \|\vr\|_\infty$.
\end{lemma}
\begin{proof}
By the triangle inequality, we have
\begin{align*}
\left\lVert(\eye-\frac{\ma}{\rho(\ma)})\vr \right\rVert_\infty  &=  \left\lVert \left[\left((1+\epsilon)\eye-\frac{\ma}{\rho(\ma)}\right) - \epsilon \eye\right] \vr \right\rVert_\infty \\
&\leq \left\lVert((1+\epsilon)\eye-\frac{\ma}{\rho(\ma)})\vr \right\rVert_\infty \! + \! \left\lVert\epsilon \vr \right\rVert_\infty \leq \left\lVert \vw \right\rVert_{\infty} + \epsilon\left\lVert\vr \right\rVert_\infty \leq \frac{7}{4} + \epsilon \left\lVert \vr \right\rVert_\infty.
\end{align*}
Now, we note that 
\[
	\|\vr\|_\infty = \|[(1+\epsilon)\eye - \frac{\ma}{\rho(\ma)}]^{-1} \vw \|_\infty = \|\frac{1}{1+\epsilon} \sum_{i=0}^{\infty} \left(\frac{\ma}{(1+\epsilon) \rho(\ma)}\right)^i \vw\|_\infty.
\]
Let $\vv$ be the top right eigenvector for $\ma$: note that this is a positive vector by Perron-Frobenius. As $\ma^i$ is nonnegative for any integer $i \geq 0$ and since $\vw_i \geq \frac{1}{4} \geq \frac{1}{4} \frac{\vv_i}{\|\vv\|_\infty}$ for any $i$, we obtain
\begin{align*}
	\|\vr\|_\infty &= \|\frac{1}{1+\epsilon} \sum_{i=0}^{\infty} \left(\frac{\ma}{(1+\epsilon) \rho(\ma)}\right)^i \vw\|_\infty \\
	&\geq \frac{1}{4}  \|\frac{1}{1+\epsilon} \sum_{i=0}^{\infty} \left(\frac{\ma}{(1+\epsilon) \rho(\ma)}\right)^i \frac{\vv}{\|\vv\|_\infty}\|_\infty \\
	&= \frac{1}{4 \| \vv \|_\infty} \|\frac{1}{1+\epsilon} \sum_{i=0}^{\infty} \left(\frac{\rho(\ma)}{(1+\epsilon) \rho(\ma)}\right)^i \vv \|_\infty \\
	&= \frac{1}{4 \| \vv \|_\infty} \frac{1}{1+\epsilon} \sum_{i=0}^{\infty} \left(\frac{1}{1+\epsilon}\right)^i \| \vv \|_\infty = \frac{1}{4 \epsilon}.
\end{align*}
Combining these two facts yields 
$$
\left\lVert(\eye-\frac{\ma}{\rho(\ma)})\vr \right\rVert_\infty \leq \frac{7}{4}+ \epsilon\left\lVert\vr \right\rVert_\infty \leq 8 \epsilon \norm{\vr}_\infty.
$$	
\end{proof}

This fact inspires an algorithm for computing approximate Perron vectors. We first find $s \in (\rho(\ma), (1+\epsilon)\rho(\ma)]$ by using {\sc \findperronvalue} and next use {\sc \mscale} to compute  $\vl \approx {[(1+\epsilon)\eye-\frac{\ma}{s}]^{-1}}^\top \vones$ and $\vr \approx [(1+\epsilon)\eye-\frac{\ma}{s}]^{-1} \vones$. Note that by Lemma~\ref{lem:perron-approx}, $\vl$ and $\vr$ are approximate eigenvectors corresponding to $\rho(\ma)$. We formally state this in the following algorithm and thoerem.

\newcommand{\simpleperron}{$\code{Simple-Perron}$}
\begin{algorithm}[h]
	\caption{Simple Algorithm for computing the approximate largest eigenvalue and eigenvectors of a non-negative irreducible matrix $\ma$ when a bound on the condition number of the left and right eigenvectors are given.} \label{alg:perron-simple}
	\begin{algorithmic}[1]	
		\Function{\simpleperron}{$\ma$, $\epsilon$, $K$}
			\State \textbf{Input:} $\ma\in \R_{\geq 0}^{n\times n}$ is an irreducible matrix, and $\epsilon,K >0$ is a real number.
			\State \textbf{Output:} $\langle s, \vl, \vr \rangle$ 
			such that $1< \frac{s}{\rho(\ma)} \leq 1+\epsilon$, $\norm{\vl^\top(\eye-\frac{\ma}{s})}_\infty <4\epsilon$, and $\norm{(\eye-\frac{\ma}{s})\vr}_\infty <4\epsilon$
			\State $s \eq $ {\sc \findperronvalue}($\ma$, 0, $||A||_\infty$, $\epsilon$,  $K$)
			\State $\vl, \vr \eq$ {\sc \mscale($\frac{\ma}{(1+\epsilon/2)s}$, $\epsilon/3$, $\frac{2K}{\epsilon}$)} 
			
			\State \Return $\langle s, \vl, \vr \rangle$
		\EndFunction
	\end{algorithmic}
\end{algorithm}

\begin{lemma}
	Let $\ma$ be an irreducible nonnegative matrix. Let $1/4 > \epsilon > 0$ be a parameter, and let $K \geq \kappa(\eigl) + \kappa(\eigr)$. Then with high probability the procedure {\sc \simpleperron}($\ma, \epsilon, K$) vectors $\vl, \vr$ satisfying $\norm{\vl^\top(\eye-\frac{\ma}{s})}_\infty < O(\epsilon) \norm{\vl}$ and $\norm{(\eye-\frac{\ma}{s})\vr}_\infty < O(\epsilon) \norm{\vr}$. This procedure runs in time 
	\[
		O\Big(\mathcal{T}_{solve}\Big(m, n, O(\frac{K^2}{\epsilon^2}), O(\frac{\epsilon^3}{K^3}) \Big) \log(n \frac{K}{\\epsilon}) \log\Big(\frac{\|\ma\|_1 + \|\ma\|_\infty}{\epsilon} \Big)  \log(\frac{|s_1-s_2|}{\epsilon\rho(\ma)}) \Big)		
	\]
	\end{lemma}
\begin{proof}
	By Lemma \ref{lem:perrron-find}, we have that $s \in (\rho(\ma), (1+\epsilon) \rho(\ma))$. Thus $\eye-\frac{\ma}{s}$ is an M-matrix. Thus, {\sc \mscale} will succeed in finding scaling vectors $\vl_\epsilon$ and $\vr_\epsilon$ which satisfy $||((1+\epsilon/3)\eye-\frac{\ma}{(1+\epsilon/2)s}) \vr_\epsilon - \vones||_\infty \leq 1/2$ and $||((1+\epsilon/3)\eye-\frac{\ma}{(1+\epsilon/2)s})^\top \vl_\epsilon - \vones||_\infty \leq 1/2$. Define $\alpha = \frac{s}{\rho(\ma)}$. By rearranging these two equations we get
	\[
		||((1 + \epsilon) \alpha \eye - \frac{\ma}{\rho(\ma)}) \vr_\epsilon - \vones||_\infty \leq \frac{1+\epsilon/3}{2} \alpha 	
\]		
and
\[		
		||((1+\epsilon) \alpha \eye-\frac{\ma}{\rho(\ma)})^\top \vl_\epsilon - \vones||_\infty \leq \frac{1+\epsilon/3}{2} \alpha.
	\]
	Since $(1+\epsilon)\alpha) \leq 1+ 2\epsilon$ and $\frac{1+\epsilon/3}{2} \leq \frac{3}{4}$, applying Lemma \ref{lem:perron-approx} implies the result. The running time follows from Lemma \ref{lem:decision}.
\end{proof}

\subsection{Fix for not knowing $K$}\label{sec:kfix}

All the algorithms we have presented for the M-matrix decision problem and the Perron problem depend on knowing an upper bound on the condition number of the left and right eigenvectors of $\ma$, namely $K$. Unfortunately we do not know this value in advance. To fix this we use a doubling technique and start with an initial guess $K=1$ and double it until we get the appropriate results. However, we run into a problem when naively implementing this strategy. For example, consider the M-matrix decision problem and function {\sc \mdecide} (Algorithm~\ref{alg:decision}). If {\sc \mdecide($\ma$, $\epsilon$, $K$)} fails (i.e. reports $\eye-\ma$ is not an M-matrix) it might be caused either by the fact that $\rho(\ma)>1$ or our guess for $K$ is wrong. Therefore, doubling does not give a fix in that case. But we are able to show that doubling gives us a fix for the Perron problem which in turn can be used to solve the M-matrix decision problem (without dependence on knowing $K$ in advance).

The main trick we use in our fix is represented in Lemma~\ref{lem:perron-approx-inverse} below. If we use doubling combined with {\sc \findperronvalue}, we are guaranteed that after $\log(\kappa(\eigr)+\kappa(\eigl))$ guesses for $K$, for any positive $\epsilon$ we get $\epsilon$ approximation to $\rho(\ma)$, but we cannot verify that as we do not know $\rho(A)$. Lemma~\ref{lem:perron-approx-inverse} gives us a tool to do that.
Let $\delta >0$ be any positive number less than one. Lemma~\ref{lem:perron-approx-inverse} suggests that if we had $s$ where 
\begin{equation}\label{eq:fixk}
\rho(\ma) \leq s\leq (1+\frac{\delta}{8k(\eigr)^2})\rho(\ma),
\end{equation}
 then we can construct a proof that  $\rho(\ma) \geq (1-\delta) s$. 
 This gives us a convenient way to verify whether the approximation to the eigenvalue generated by {\sc \findperronvalue} is correct, and hence a way to determine if our choice of $K$ was large enough. Note that no matter the choice of $K$ our algorithm returns some $s > \rho(\ma)$.
.

\begin{lemma}\label{lem:perron-approx-inverse}
	Given a non-negative irreducible matrix $\ma\in \R_{\geq 0}^{n\times n}$ and a real number $0<\epsilon <\frac{1}{8\kappa(\eigr)^2}$ where $\eigr\in \R_{>0}^n$ is the right eigenvector of $\ma$; Let $s$ be any real number such that $\rho(\ma) \leq s < (1+\epsilon)\rho(\ma)$, $\vr = [(1+\epsilon)\eye - \frac{\ma}{s}]^{-1}\vones$, and $\vv = \frac{\vr}{\|\vr\|_\infty}$, then $\mv^{-1} \frac{\ma}{s} \vv \geq (1-8\epsilon \kappa(\eigr)^2) \vones$.
\end{lemma}

\begin{proof}
	We note that since $\eye - \frac{\ma}{s}$ is an M-matrix. By Lemma \ref{lemma:eigvec_scaling}, we obtain $\|\mr\|_\infty \leq \| [(1+\epsilon)\eye - \frac{\ma}{s}]^{-1}\|_\infty \leq \frac{\kappa(\eigr)}{\epsilon}$. Next we show an element-wise lower-bound on $\vr$. Let $\alpha = \rho(\ma)/s$. Since $\vones \geq  \frac{\eigr}{\|\eigr\|_\infty}$, we observe
	\begin{align}
		\vr 
		&\geq  \frac{1}{1+\epsilon}\sum_{i=0}^\infty \left(\frac{\alpha}{1+\epsilon}\right)^i \left(\frac{\ma}{\rho(\ma)}\right)^i \frac{\eigr}{\|\eigr\|_\infty}\nonumber\\
		&= \frac{1}{1+\epsilon}\sum_{i=0}^\infty \left(\frac{\alpha}{1+\epsilon}\right)^i  \frac{\eigr}{\|\eigr\|_\infty} \nonumber \\
		&\geq \frac{1}{1+\epsilon-\alpha} \cdot \frac{\eigr}{\|\eigr\|_\infty} \label{eq:perron-approx-inverse-eq3}
	\end{align}
	Note that $ \frac{\eigr}{\|\eigr\|_\infty} \geq \frac{1}{\kappa(\eigr)}\vones$, and therefore by \eqref{eq:perron-approx-inverse-eq3}
	\begin{equation}\label{eq:perron-approx-inverse-eq4}
	 \vr \geq \frac{1}{1+\epsilon-\alpha} \cdot \frac{\eigr}{\|\eigr\|_\infty} \geq  \frac{1}{1+\epsilon-\alpha} \cdot \frac{1}{\kappa(\eigr)}\vones
	\end{equation}
	Now combining the above yields
	\begin{equation}\label{eq:perron-approx-inverse-eq5}
		\vv = \frac{\vr}{\|\vr\|_\infty} \geq \frac{\epsilon}{(1+\epsilon-\alpha)\kappa(\eigr)^2}\vones \quad   
		\Rightarrow \quad  \diag(\mv)^{-1} \vones \leq  \frac{(1+\epsilon-\alpha)\kappa(\eigr)^2}{\epsilon}\vones
	\end{equation}
	By Lemma~\ref{lem:perron-approx}, we have $\|(\eye - \frac{\ma}{s}) \vv \|_\infty < 4\epsilon$ and therefore
	\begin{equation}\label{eq:perron-approx-inverse-eq6}
		(\eye- \frac{\ma}{s}) \vv \leq  \epsilon \vones \quad \Rightarrow \quad \frac{\ma}{s} \vv \geq \vv - 4\epsilon \vones.
	\end{equation}
	Now combining~\eqref{eq:perron-approx-inverse-eq5} and \eqref{eq:perron-approx-inverse-eq6}, we have,
	\[
		\diag(\mv)^{-1} \frac{\ma}{s} \vv \geq \vones - 4 \epsilon \cdot \diag(\mv)^{-1} \vones \geq \vones -  4(1+\epsilon-\alpha)\kappa(\eigr)^2\vones
	\]
	But $\alpha \geq \frac{1}{1+\epsilon}$ and hence $1+\epsilon - \alpha \leq 1+\epsilon - \frac{1}{1+\epsilon} \leq 2 \epsilon$. Therefore,
	\[
		\diag(\mv)^{-1} \frac{\ma}{s} \vv  \geq \vones -  4(1+\epsilon-\alpha)\kappa(\eigr)^2\vones \geq (1-8\epsilon\kappa(\eigr)^2) \vones.
	\]
	
\end{proof}

\newcommand{\computeperron}{$\code{Compute-Perron}$}
\begin{algorithm}[t]
	\caption {The nearly linear time algorithm for computing the largest eigenvalue and its corresponding approximate eigenvectors} \label{alg:perron}
	\begin{algorithmic}[1]	
		\Function{\computeperron}{$\ma$, $\delta$}
		\State \textbf{Input:} $\ma\in \R_{\geq 0}^{n\times n}$ is an irreducible matrix, and $\delta \in \R_{>0}$ less than one.
		\State \textbf{Output:} Real number $s>0$, and vectors $\vl, \vr\in \R_{>0}^{n}$ such that $(1-\delta) \rho(\ma) < s \leq \rho(\ma)$, $\norm{\vl^\top(\eye-\ma/s)}_\infty <\frac{\delta}{2K^2} \norm{\vl}_\infty$, and $\norm{(\eye-\ma/s)\vr}_\infty <\frac{\delta}{2K^2} \norm{\vr}_\infty$
		\State $K \eq 1$
		\While{true}
		\State $\langle s, \vl, \vr \rangle \eq $ {\sc \simpleperron}($\ma$, $\frac{\delta}{8K^2}$,  $K$)
		\If{$\vl$, $\vr$ are $\frac{\delta}{2K^2}$-approximate eigenvectors of $s$, \\and either $\mr^{-1}\ma\vr \geq (1-\delta)s\vones$ or $\ml^{-1}\ma^\top\vl \geq (1-\delta)s\vones$}
		\State \Return $\langle s, \vl, \vr \rangle$
		\EndIf
		\State $K \eq 2K$
		\EndWhile
		\EndFunction
	\end{algorithmic}
\end{algorithm}

Note that if $K \geq \kappa(\eigr)+\kappa(\eigl)$, then {\sc \simpleperron($\ma$, $\frac{\delta}{8K^2}$,$K$)} returns $s$ and $\vr$ that satisfy the terms of Lemma~\ref{lem:perron-approx-inverse}. Therefore, they can be used to verify whether we are close enough to $\rho(\ma)$ or not. Knowing this we are ready to give our final algorithm which does not depend on knowing the right $K$ in advance. We start with $K=1$, and double $K$ in each iteration. In each iteration {\sc \simpleperron($\ma$, $\frac{\delta}{8K^2}$,$K$)} is called. We do this until we get either $\diag(\vr)^{-1} \ma \vr \geq (1-\delta) s$, or $\diag(\vl)^{-1} \ma^\top \vl \geq (1-\delta) s$. A pseudo-code for this algorithm is provided in Algorithm~\ref{alg:perron}. Below is a short proof of Theorem~\ref{thm:perron-overview}.

\begin{proof}[Proof of Theorem~\ref{thm:perron-overview}]
 Note that after $O\left(\log(\kappa(\eigl)+\kappa(\eigr))\right)$, $K$ becomes greater than or equal to $\kappa(\eigl)+\kappa(\eigr)$. At this point by Lemma~\ref{lem:perrron-find} and \ref{lem:perron-approx-inverse} , {\sc \simpleperron} returns $\langle s, \vl, \vr \rangle$ such that $(1-\delta)\rho(\ma) < s \leq \rho(\ma)$, $\|(s \eye - \mm) \vr\|_\infty \leq \frac{\delta}{K^2}\norm{\vr}_\infty$, and $\|\vl^\top (s \eye - \mm)\|_\infty \leq \frac{\delta}{K^2}\norm{\vl}_\infty$. Thus giving the conditions in the problem statement. For the running time, note that each call to {\sc \simpleperron} consists of a call to {\sc \findperronvalue} and a call to {\sc \mscale}. By Lemma~\ref{lem:perrron-find}, each call to {\sc \simpleperron} takes at most
 \[
	O\Big(\mathcal{T}_{solve}\Big(m, n, O(\frac{K^6}{\delta^2}), O(\frac{\delta^3}{K^9}) \Big) \log(n \frac{K}{\delta}) \log\Big(\frac{\|\ma\|_1 + \|\ma\|_\infty}{\delta} \Big)  \log \left(\frac{|s_1-s_2|}{\delta\rho(\ma)}\right) \Big)		
\]
 time. As we run {\sc \simpleperron} $O(\log K)$ times, the total running follows.
\end{proof}

As an immediate corollary to this result, our main theorem follows. We restate it for clarity.

\perron*
\begin{proof}
	By using Theorem \ref{thm:rcdd_solver} inside of Theorem \ref{thm:perron-overview}, the result follows immediately.
\end{proof}

\section{Applications}
\label{sec:applications}

In this section, we present a variety of problems in which M-matrices appear and show how our methods can be applied to them.
Matrices with non-positive off diagonals and non-negative diagonals (known as $Z$-matrices) appear in many applications like finite difference methods for partial differential equations, input-output productions and growth models in economics, linear complementarity problems in operations research,  Markov processes in probability and statistics, and social networks analysis \cite{cottle2008linear,berman1994nonnegative,murty1988linear,ten2010input}.
In many of these problems, if the underlying $Z$-matrix is actually an M-matrix, we will have additional interesting and meaningful structural properties. For example, in the Katz centrality problem (covered in Section~\ref{sec:katz}) the centrality measure is only valid if the corresponding matrix is an M-matrix. 

In the following subsections, we provide faster algorithms for some well-known applications which involve M-matrices. We start with examples like
computing Katz Centrality, and solving Leontief Equations in which our results can be applied directly. Then, in Section~\ref{sec:singular} we show how to compute the top singular value and its corresponding singular vectors of a non-negative matrix in nearly linear time. In Section~\ref{sec:graph-kernels}, we show another application of our method in computing random-walk based graph kernels. These kernels are used in a variety of fields like biology, social networks, and natural language processing to compare the structure of two graphs \cite{kash2003, kash2004, vish2010}.

\subsection{Leontief Equations}


In economics, input-output analysis is an important quantitative technique that models the interdependencies between different branches of a national economy \cite{miller2009input, ten2010input}. In recognition of W.~Leontief contributions for which he received the Nobel prize in economics, the name Leontief model is widely used to refer to these models.

The Leontief model explains the interdependencies between different sectors within a national economy. It models how the output of one sector can be used as input to another sector  and how  changes in the production of one will affect the others. Let $n$ be the number of sectors in an economy. Suppose sector $i$ produces $x_i$ units of a single homogeneous good. Assume that the $j$th sector requires $a_{ij}$ units from sector $i$ in order to produce one unit of good $j$. Further assume there is an external demand $d_i \geq 0$ for the good produced by sector $i$. Then for each sector $i$ we have
$$
x_i = \sum_{j=1}^{n} a_{ij}x_j + d_i.
$$
Using vector notation, we can write these conditions as $\vx = \ma \vx + \vd$, or equivalently $(\mI-\ma)\vx = \vd$. Existence of a non-negative vector $\vx$ (production rates of each sector) that satisfies the input-output relation is crucial for the economy to be in equilibrium (demand equals supply).

In economics, the condition which guarantees existence of such vector is known as Hawkin-Simons condition~\cite{hawkins1949note}. The Hawkin-Simons theorem states that the necessary and sufficient condition for the existence of a non-negative vector $\vx$ which solves $(\mI - \ma)\vx = \vd$ is that all principal minors of $\mI-\ma$ be positive, which is equivalent to $\mI - \ma$ be an M-matrix \cite{berman1994nonnegative}. Therefore checking whether $\mI-\ma$ satisfies Hawkin-Simons condition is equivalent to solving the M-matrix decision problem. Note that we showed how to solve this problem in nearly linear time in Section~\ref{sec:perron-alg}. Moreover, for a given demand vector $\vd$, one can use our matrix scaling algorithm to obtain a  vector $\hat{x}$ which satisfies $(\mI-\ma)\hat{x} \approx \vd $ in nearly linear time.

\subsection{Katz Centrality}\label{sec:katz}

As a relatively simple application of our techniques to solve M-matrices, we present a nearly-linear time algorithm for computing the Katz centralities in a graph. This is a measure of centrality similar to PageRank following the same intuition that in a social network, a person is influential if they can influence other influential people. Algebraically, we define the Katz vector $\vv$ to be the solution to the equation $\vv = \alpha \ma \vv + \vb$, where $\ma$ is the weighted adjacency matrix of the underlying graph, $\alpha < \rho(\ma) $ is some appropriately chosen decay parameter, and $\vb$ is some initial ground truth vector\cite{katz1953new}. In contrast, the PageRank vector $\vp$ is the vector which satisfies $\vp = \beta \md^{-1} \ma \vp + \vb$  where $\beta < 1$ and $\md$ is the degree matrix of the underlying graph. 

Unlike PageRank, Katz centrality has remained relatively unexplored. This is partially because of the algorithmic challenges in computing Katz centrality. For PageRank  we have nearly-linear time exact algorithms obtained by directly applying Laplacian system solvers, and in the special case of unweighted graphs the celebrated PPR-Push algorithm of Andersen et al computes vectors $vp'$ satisfying $\|\md^{-1} r(vp') \|_\infty \leq \epsilon$ in nearly-constant time.\cite{andersen2006local}
While we do not present an analogue of this  nearly-constant running time, we do leverage our M-matrix solver to obtain an analogous nearly-linear time exact algorithm for computing Katz centrality. This follows almost immediately from the definition: basic algebra shows that the Katz centrality satisfies 
\[
\vv = (\eye - \alpha \ma)^{-1} \vb.
\]

Thus by using Theorem \ref{thm:msolve} on the linear system which arises here, we get our desired nearly-linear time algorithm.

\subsection{Left and Right Top Singular Vectors}\label{sec:singular}
	Here we show how to compute the top singular value and associated eigenvectors of a non-negative matrix in nearly linear time. Let $\ma \in \R_{\geq 0}^{n\times n}$ be a non-negative matrix. The top singular value $\sigma_{\max}$, top right singular vector $\vv_{R}$, and top left singular vector $\vv_{L}$ are defined as follows:

$$\sigma_{\max} = \max_{\|\vv\|_2 = 1} \|\ma\vv\|_2$$
$$\vv_R = \argmax_{\|\vv\|_2 = 1} \|\ma \vv\|_2$$
$$\vv_L = \argmax_{\|\vv\|_2 = 1} \|\vv^\top \ma\|_2$$

From these definitions, it is clear that $\sigma_{\max}$ is equal to the top eigenvalue of $\ma^\top \ma$. And $\vv_R$ and $\vv_L$ are the top eigenvectors of $\ma^\top \ma$ and $\ma \ma^\top$ respectively. So the problem of finding the top singular value, and the top right singular vector of $\ma$ reduces to computing the top eigenvalue and eigenvector of symmetric matrix $\ma^\top \ma$ (and $\ma \ma^\top$ for the top left singular vector). 

One might apply our methods directly to matrix $\ma^\top \ma$ to find its largest eigenvalue and corresponding eigenvector. However, $\ma^\top \ma$ might be dense and in the worst case the number of non-zero elements in $\ma^\top \ma$ might be quadratic in the number of non-zero elements of $\ma$.  Below, we describe a fix for this problem.

Let $\mm=\mI-\ma^\top \ma$.  Our framework still gives a nearly linear time algorithm for computing $\rho(\ma^\top \ma)$, if we are able to do two operations in nearly linear time: (i) apply $\mm$ to an arbitrary vector $\vy \in \R^n$ and (ii) solve linear systems in $\mm$ when $\mm$ is RCDD (see Oracle~\ref{oracle:RCDDSolve}). We can obviously do (i) in linear time as one can compute $\ma \ma^\top \vy$ in $O(\nnz(\ma)+n)$. Next we show how to do (ii) in nearly linear time. 

Because of the structure in $\ma \ma^\top$, we are able to compute a constant preconditioner for $\ma$, which then can be used by preconditioned richardson to produce a nearly linear time solver for $\ma$. The following lemma from Cohen et. al. \cite{coh16} which in turn is inspired from \cite{PengS14} is crucial for obtaining the preconditioner for  $\ma$ in nearly linear time. 

\begin{lemma}\label{lem:vecsparse}
Let $\vx \in \R_{\geq 0}^n$ be a non-zero vector, and  $\sigma, p\in (0,1)$. Also let $t$ be the number of non-zero elements in $\vx$ and  $\mlap = \diag(\vx) - \frac{1}{\|\vx\|_1}\vx \vx^\top$. There is a randomized algorithm that in time $O(t\sigma^{-2} \log(t/p))$ and with probability at least $1-p$ computes a Laplacian $\tilde{\mlap}$ such that $\tilde{\mlap}\approx_\sigma \mlap$ and number of non-zero elements in $\tilde{\mlap}$ is $O(t\sigma^{-2} \log(t/p))$.
\end{lemma}

Using Lemma~\ref{lem:vecsparse} we are able to give a nearly linear time solver for $\mm$. Note that the following lemma is standard, when we have access to a sparsifier for $\mm$.

\begin{lemma}
	For a given non-negative matrix $\ma \in \R_{\geq 0}^{n\times n}$, let $\mm \eqdef \mI - \ma\ma^\top$ be an RCDD matrix. Given a vector $\vb\in \R^n$ and $\epsilon \in \R_{>0}$, let $\vx = \mm^{-1}\vb$, there is an algorithm which computes an operator $Z: \R^n \rightarrow \R^n$, such that
	$$
	\norm{Z(\vb) - \vx}_\mm \leq \epsilon \norm{\vx}_\mm
	$$
	in time $\otilde(\nnz(\ma))$. 
\end{lemma}
\begin{proof}
Let $\va_i$ be the $i$th column of $\ma$, \ and $\mlap_i \eqdef \norm{\va_i}_1\diag(\va_i) - \va_i \va_i^\top$. Since $\mm$ is RCDD, we can write
 $$
  \mm=\mI - \sum_{i=1}^n \va_i \va_i^\top  = F + \sum_{i=1}^n \mlap_i,
 $$
where $F$ is a non-negative diagonal matrix. Let $d_i$ be the number of non-zero elements in $\va_i$. 
 Using Lemma~\ref{lem:vecsparse} we can sparsify each $\mlap_i$  in nearly linear time in $d_i$. 
 Let $\tilde{\mlap}_i$ be the $\sigma$-sparsifier of $\mlap_i$ obtained from Lemma~\ref{lem:vecsparse}. 
 Then with high probability $\tilde{\mm} \eqdef F + \sum_{i=1}^n \tilde{\mlap}_i$ is a spectral sparsifier of $\mm$ (i.e. $\tilde{\mm} \approx_\sigma \mm$) with size $\tilde{O}(m\sigma^{-2})$, where $m=\sum_{i} d_i = \nnz(\ma)$. Note that
  $$
  \tilde{\mm} \approx_\sigma \mm \Rightarrow \tilde{\mm}^{-1} \approx_\sigma \mm^{-1} \cite{peng14}
  $$
  But $\tilde{M}$ has linear size.
Since $\tilde{\mm}$ is symmetric strictly diagonally dominant (SDD) matrix, we can apply results from standard SDD solvers  ~\cite{koutis10,kel13} to get an operator $\tilde{Z}$, such that
$\norm{\tilde{\mm}\tilde{Z}-\mI}_{\tilde{\mm}} < c = O(\sigma)$. By tunning constant $c$, we can get the same result for $\mm$ as well. Therefore, we can get an operator $\tilde{Z}$ such that 
$$
\norm{\mm\tilde{Z}-\mI}_{\mm} < O(\sigma)
$$
in time $\tilde{O}(m\sigma^{-2})$. Having this, $Z=${\sc \prichardson}($\mm$, $\tilde{Z}$,$b$,$\epsilon$, $\vx_0$) has the properties mentioned in the statement (\cite{peng14}, Lemma 4.4). 
 \end{proof}

 By the above lemma, we can get the same results as in Theorem~\ref{thm:perron-overview} for computing $\sigma_{\max}(\ma)$ and its corresponding singular vectors. Note that since $\ma \ma^\top$ is symmetric, a tighter analysis of our general M-matrix scaling algorithm is possible and gives better bounds with fewer polylogarithmic factors for the top singular value problem.

\subsection{Graph Kernels}\label{sec:graph-kernels}
	In this section we present another important application of our method on computing graph kernels. Kernel functions are widely used in machine learning for comparing the similarity of two complex objects. Graph kernels are specifically used for comparing the similarity between two graphs with applications in social networks, studying chemical compounds, comparison and function prediction of protein structures, and analysis of semantic structures in natural language processing \cite{kash2003, kash2004, vish2010}. 

One of the main obstacles in applying the existing algorithms for computing the kernel function between two graphs is their large running time~\cite{vish2010}. Here, we show how to obtain improved running times for computing canonical kernel functions known as random walk kernels.

Commonly, graphs considered in this contexts are labeled. An edge labeled graph is denoted by $G=\langle V, E, \ell \rangle$, where $V$ is the set of vertices, $E \subseteq V\times V$ is the set of edges, $\ell:E\rightarrow \Sigma$ is a label assignment function to edges, and $\Sigma$ is the set of labels. 
Here we assume the set of labels is finite, and to ease the notations let $\Sigma\eqdef \{1,2,\cdots, d\}$ (i.~e. labels are positive integers).

Let $G$ and $H$ be edge labeled graphs. Random walk kernels, compare $G$ and $H$  based on the similarity of labels in simultaneous random walks on $G$ and $H$. 
It will be easier to think of simultaneous walks on $G$ and $H$ as a random walk on their Cartesian product: 
\begin{definition}
Given two graphs $G = \langle V_G, E_G, \ell_G \rangle$ and $H = \langle V_H, E_H, \ell_H\rangle$, Cartesian product of $G$ and $H$ is a labeled graph denoted by $G\otimes H = \langle V_{G\otimes H}, E_{G\otimes H}, \ell_{G\otimes H}\rangle$, where
$$
V_{G\otimes H} = \{(u,v)|u\in V_G, v\in V_H\}
$$
$$
E_{G\otimes H} = \{((u,v), (w,z))|(u,w)\in E_G, (v,z) \in E_H\}
$$
and $\ell_{G\otimes H}: E_{G\otimes H} \rightarrow \Sigma\times \Sigma$ is a label assignment function such that for an edge $e \eqdef ((u,v), (w,z)) \in E_{G\otimes H}$, $\ell_{G\otimes H}(e) = (\ell_G((u,w)), \ell_H((v,z))$.
\end{definition}

Note that a walk on $G\otimes H$ corresponds to two simultaneous walks on $G$ and $H$. Let $\mw \in \R^{nm\times nm}$ be an arbitrary matrix. We use $\mw[(i,j),(g,h)]$ notation to refer to $\mw[i \cdot m+j, g \cdot m+h]$.  For a given product graph $G\otimes H$, we say $\mw_{G\otimes H} \in \R_{\geq 0}^{|V_G||V_H|\times |V_G||V_H|}$ is a weighted adjacency matrix of $G\otimes H$ if 
$$
\begin{cases}
\mw_{G\otimes H}[(u,v), (w,z)]  \geq 0 & \text{if }(u,w)\in E_G \wedge (v,z) \in H, \\
\mw_{G\otimes H}[(u,v), (w,z)]  = 0 & \text{otherwise.}
\end{cases}
$$
When it is clear from the context, we may use $\mw_\times$ instead of $\mw_{G\otimes H}$ to refer to a weighted adjacency matrix of $G\otimes H$. Note that  $\mw_\times[(u,v),(w,z)]$ can be used as a measure for how similar are edges $(u,w)\in E_G$ and $(v,z) \in E_H$ (i.e. if this value is large they are considered to be more similar). For example the following weight matrix

$$
\begin{cases}
\mw_{G\otimes H}[(u,v), (w,z)]  = 1 & \text{if }(u,w)\in E_G \wedge (v,z) \in H \wedge \ell((u,w))=\ell((v,z)),\\
\mw_{G\otimes H}[(u,v), (w,z)]  = 0 & \text{otherwise.}
\end{cases}
$$
considers two edges to be similar if they have the same label. 

In the same sense, one can use $\mw_\times$ to come up with a measure for comparing length $k$ walks in graph $G$ and $H$. Particularly one can use $\mw_\times^k[(u,v),(w,z)]$, to compare the similarity of length $k$ walks from $u$ to $w$ in $G$ and from $v$ to $z$ in $H$. 
Using $\mw_\times$ as a similarity measure, Vishnawathan et al. \cite{vish2010} define the random walk graph kernel between $G$ and $H$ as follows.

\begin{definition}
Given labeled graphs $G$ and $H$, and initial and stopping probability distributions  $p_\times,q_\times \in \R_{\geq 0}^{|V_G||V_H|}$ such that $\|p_\times\|_1=1$ and $\|q_\times\|_1=1$, the graph kernel of $G$ and $H$ is defined as
\begin{equation}\label{eq:kernel}
\kappa(G,H) = \sum_{k=0}^{\infty} \lambda_k q^\top_\times \mw_\times^k p_\times
\end{equation}
where $\lambda_k$ is a decay factor depending on the length of walks, such that the above sum converges.
\end{definition}

Note that $q^\top_\times \mw_\times^k p_\times$ can be considered as the expected similarity of length $k$ walks in graph $G$ and $H$.
Vishnawathan et al. \cite{vish2010} consider a  ``geometric decay'' for $\lambda_k$, which means $\lambda_k =\lambda^k$ for some positive $\lambda<1$. In this case Equation~\ref{eq:kernel} can be written more concisely as
\begin{equation}\label{eq:kernel-def}
\kappa(G,H) = \sum_{k=0}^{\infty}  q^\top_\times (\lambda \mw_\times)^k p_\times = q_\times^\top(\mI-\lambda \mw_\times)^{-1} p_\times
\end{equation}
Observe that the problem of computing graph kernels reduces to inverting $\mI-\lambda \mw_\times$. Moreover we see that a sufficient condition for $\kappa(G,H)$ to converge is $\|\lambda \mw_\times\| < 1$. Given this one may immediately apply our techniques to compute graph kernels in nearly linear time in the size of $\mw_{\times}$. Moreover, for a given $\mw_{\times}$ we are able to check whether it defines a valid kernel, by applying our M-matrix decision algorithm on $\mI-\lambda \mw_{\times}$.

\section{Acknowledgments}
We thank Michael B. Cohen for helpful conversations which identified key technical challenges in solving M-matrices that motivated this work.

\bibliographystyle{abbrv}
\bibliography{bibs/abbrv,bibs/ref}

\section{M-matrix Facts}
\label{sec:mmatrix_facts}

Here we provide several facts about M-matrices that we use extensively throughout the paper. For a good survey on M-matrices refer to \cite{berman1994nonnegative}.

\begin{lemma} \label{lemma:MMatrix_monotone} For non-negative $\ma \in \R^\nn_{\geq 0}$ the matrix $\mm = \mI - \ma$ is an M-matrix if and only if 
	\begin{enumerate}[label=\textbf{\roman*})]
		\item for any vector $x$ satisfying $\mm x > 0$ entrywise we have $x > 0$. 
		\item for any vector $y$ satisfying $\mm^\top y > 0$ entrywise we have $y > 0$.
	\end{enumerate}
\end{lemma}

\MMatrixRCDD*

\begin{proof}
	We begin with the first statement. Recall by definitions that $\ms$ is RCDD if for every $i$, 
	\[
	\ms_{ii} \geq \sum_{j \neq i} |\ms_{ij}| \hspace{3mm} \text{and} \hspace{3mm} \ms_{ii} \geq \sum_{j \neq i} |\ms_{ji}|.
	\]
	Note that by Lemma \ref{lemma:MMatrix_monotone}, we have $y > \vzero$ and $x > \vzero$. Thus, since $\ms_{ij} = x_i \mm_{ij} y_j$ $\ms_{ij}$ has the same sign as $\mm_{ij}$ for all $i, j$. As $\mm_{ii} > 0$ for all $i$ and $\mm_{ij} \leq 0$ for $i \neq j$, we see that $\ms$ is RCDD if 
	\[
	\ms_{ii} + \sum_{j \neq i} \ms_{ij} \geq 0 \hspace{3mm} \text{and} \hspace{3mm} \ms_{ii} + \sum_{j \neq i} \ms_{ji} \geq 0,
	\]
	or in other words if $\ms \vones \geq \vzero$ and $\ms^\top \vones \geq \vzero$. Now, observe for every $i$, $(\ms \vones)_i = (\mx\mm\my \vones)_i = (\mx\mm y)_i = x_i (\mm y)_i  > \vzero$ and $(\ms^\top \vones)_i = (\my \mm^\top \mx \vones)_i = (\my \mm^\top x)_i = y_i (\mm^\top x)_i > \vzero$ by definition of $x$ and $y$. Thus $\ms$ is RCDD. 

	We now analyze the second statement. As $x$ and $y$ are positive vectors, we again observe that $\ms_{ij}$ is has the same sign that $\mm_{ij}$ does by the same argument as in the previous claim. Now, observe that since $\ms$ is RCDD and has nonpositive off diagonal,$\ms$ being RCDD is equivalent to the claims that $\ms \vones = \mx \mm y> \vzero$ and $\ms^\top \vones = \my \mm^\top x > \vzero$. As $x$ and $y$ are positive, these in turn are equivalent to $\mm y > \vzero$ and $\mm^\top x > \vzero$.. By the Perron-Frobenius theorem, we know that there exists a positive vector $\vv$ where $\ma \vv = \rho(\ma) \vv$. As $\vv > 0$, we have for any vector $\va > 0$ that $\va^\top \vv > 0$ and so 
	\[
	0 = \vzero^\top \vv < \vv^\top \mm^\top x = s \vv^\top \vx - (\ma \vv)^\top \vx = (s -\rho(\ma)) \vv^\top \vx.	
	\]
	Thus $s - \rho(\ma) > 0$: this implies $\rho(\ma) < s$ and $\mm$ is an M-matrix.
\end{proof}

\pdscaling*
\begin{proof}
	We note that by construction
	\begin{align*}
	\md^{-1/2} \mm^\top \md^{1/2} + \md^{1/2} \mm \md^{-1/2} &= \md^{-1/2} (\mm^\top \md + \md \mm )\md^{-1/2} \\
	&= \md^{-1/2} (\mm^\top \ml \mr^{-1} + \mr^{-1} \ml \mm )\md^{-1/2} \\
	&= \md^{-1/2} \mr^{-1} (\mr \mm^\top \ml +  \ml \mm \mr) \mr^{-1} \md^{-1/2}.
	\end{align*}
	Now note that $\ms = \ml \mm \mr$ is RCDD by assumption. By the previous calculation, we see that our goal is equivalent to proving that $\ms + \ms^\top$ is positive definite. But now $\ms + \ms^\top$ is a symmetric diagonally dominant matrix: it clearly has nonpositive off diagonal, and $\ms + \ms^\top \geq \vzero$ since $\ms$ is RCDD. The fact that $\ms + \ms^\top$ is positive definite follows by the Gershgorin disk theorem. 
\end{proof}	


\begin{lemma} \label{lemma:pos_infnorm}
	Let $\mb$ be a nonnegative-entried matrix. Then $\|\mb\|_\infty = \max_i (\mb \vones)_i$. 
\end{lemma}
\begin{proof}
	Observe that $\|\mb\|_\infty = \max_i \sum_{j} |\mb_{ij}| = \max_i \sum_{j} \mb_{ij}$ since $\mb$ is nonnegative. Now $\sum_{j} \mb_{ij} = (\mb \vones)_i$, and the result follows.
\end{proof}

\begin{lemma} \label{lemma:eigvec_scaling}
	Let $\ma$ be an irreducible nonnegative matrix. Let $\eigr$ and $\eigl$ be the top right and left eigenvectors of $\ma$ respectively. Then $\|\mv_r^{-1} \ma \mv_r\|_\infty = \rho(\ma)$ and $\|\mv_l^{-1} \ma^\top \mv_l\|_\infty = \rho(\ma)$.
\end{lemma}
\begin{proof}
	We prove the result for $\eigr$: the proof for $\eigl$ is essentially identical. Note that by the Perron-Frobenius theorem $\eigr$ is a positive vector. Thus, we can see that $\mb = \mv_r^{-1} \ma \mv_r$ is itself a nonnegative matrix. By Lemma \ref{lemma:pos_infnorm}, we have 
	\[	
	\|\mb\|_\infty = \max_i \, (\mb \vones)_i = \max_i \, (\mv_r^{-1} \ma \eigr)_i = \rho(\ma) \max_i \, (\mv_r^{-1} \eigr)_i = \rho(\ma) \max_i \, \vones_i = \rho(\ma).
	\]
\end{proof}

\section{Proofs for Section~\ref{sec:scaling}}
\label{sec:scaling_proofs}

\begin{lemma} \label{lemma:init_scaling}
	Let $\mm$ be an M-matrix, and consider $\mm_\alpha = \mm + \alpha \eye$ with $\alpha = 2\max\{\|\mm\|_\infty, \|\mm\|_1 \}$. If $\vl = \vr = \frac{1}{\alpha} \vones$, we have for every $i$
	\[
	(\mm_\alpha \vr)_i \in \left[\frac{1}{2}, \frac{3}{2}\right]	\quad \text{and} \quad (\mm_\alpha^\top \vl)_i \in \left[\frac{1}{2}, \frac{3}{2}\right].
	\]
\end{lemma}
\begin{proof}
	We see that $\mm \vr = \vones + \frac{1}{\alpha} \mm \vones$. Now, since $\|\mm \vones\|_\infty \leq \frac{\alpha}{2}$, we have for any $i$ that $(\mm \vones)_i \in [-\frac{\alpha}{2}, \frac{\alpha}{2}]$ and thus $(\mm \vr)_i = \vones_i + \frac{1}{\alpha} (\mm \vones)_i \in [1 - \frac{1}{2}, 1 + \frac{1}{2}] = [\frac{1}{2}, \frac{3}{2}]$. The proof for $\vl$ is similar.
\end{proof}

\begin{lemma} \label{lemma:smallest_scaling}
	Let $\mm = \eye - \ma$ be an M-matrix. Let $\alpha \geq 0$ be a parameter, and let $\vw$ be a vector where $\vw_i \in [\frac{1}{2}, \frac{3}{2}]$ for every $i$. Consider the vectors $\vr = (\mm + \alpha \eye)^{-1} \vw$ and $\vl = (\mm + \alpha \eye)^{-\top} \vw$. Then for every $i$ we have $\vr_i \geq \frac{1}{2(1+\alpha)}$ and $\vl_i \geq \frac{1}{2(1+\alpha)}$
\end{lemma}
\begin{proof}
	By the definition of $\vr$, we have  $\frac{1}{2} \leq \vw_i = ((\mm + \alpha \eye) \vr)_i = \sum_j ((1+\alpha)\eye + \ma)_{ij} \vr_j = (1+\alpha) \vr_i - \sum_j \ma_{ij} \vr_j$. Note that by Lemma \ref{lemma:MMatrix_monotone} $\vr$ is a positive vector and that $\ma$ is a nonnegative matrix. Thus $\sum_j \ma_{ij} \vr_j \geq 0$, and we conclude that $\frac{1}{2} \leq (1+\alpha) \vr_i$: this implies the result. An identical proof gives the result for $\vl$ as well. 
\end{proof}

\begin{lemma} \label{lemma:scaling_infnorms}
	Let $\mm = \eye - \ma$ be an M-matrix. Let $\alpha \geq 0$, and let $\mm_\alpha = \mm + \alpha \eye$. Let $\vl_\alpha$ and $\vr_\alpha$ be vectors where $(\mm_\alpha^\top \vl_\alpha)_j \in [\frac{1}{2}, \frac{3}{2}]$ and $(\mm_\alpha \vr_\alpha)_j \in [\frac{1}{2}, \frac{3}{2}]$ for all $j$. Then
	\begin{enumerate}[label=\textbf{\roman*})]
	\item $\|\mr_\alpha\|_\infty \leq \frac{3}{2 (1+ \alpha)} \|\mm^{-1}\|_\infty$
	\item $\|\ml_\alpha\|_\infty \leq \frac{3}{2 (1+ \alpha)} \|\mm^{-1}\|_1$
	\end{enumerate}
\end{lemma}
\begin{proof}
	Define $w = \mm_\alpha \vr_\alpha$. Note that by definition $w(i) \in [\frac{1}{2}, \frac{3}{2}]$ for every $i$. Since $\frac{3}{2} \vones - w \geq \vzero$, we conclude that $\mm_\alpha^{-1} (\frac{3}{2} \vones - w) = \frac{3}{2}\mm_\alpha^{-1} \vones - \vr_\alpha \geq \vzero$ and hence $\|\vr_\alpha\|_\infty \leq \frac{3}{2} \|\mm_\alpha^{-1} \vones\|_\infty$. As $\mm_\alpha = (1 + \alpha) \eye - \ma$ and $\rho(\ma) < 1$, we see
	\[
	\|\mm_\alpha^{-1} \vones\|_\infty =  \frac{1}{1+\alpha} \|\sum_{i=0}^\infty \Big(\frac{\ma}{1+\alpha}\Big)^i \vones\|_\infty \leq \frac{1}{1+\alpha} \|\sum_{i=0}^\infty \Big(\frac{\ma}{1+\alpha}\Big)^i\|_\infty.
	\] 
	Now, observe that snce $\ma$ is a nonnegative matrix, $\ma^i$ is nonnegative for every nonnegative integer $i$. Thus, $\sum_{i=0}^\infty \Big(\frac{\ma}{1+\alpha}\Big)^i$ is a nonnegative matrix. By Lemma \ref{lemma:pos_infnorm}, we see that 
	\[
	\|\mm_\alpha^{-1} \vones\|_\infty \leq 	\frac{1}{1+\alpha} \|\sum_{i=0}^\infty \Big(\frac{\ma}{1+\alpha}\Big)^i\|_\infty = \frac{1}{1+\alpha} \max_j (\sum_{i=0}^\infty \Big(\frac{\ma}{1+\alpha}\Big)^i \vones)_j.
	\]
	Now as $\frac{1}{1+\alpha} \leq 1$ and $\ma^i$ is nonnegative, replacing the $\alpha$ in the above with $\epsilon$ can only increase every entry of the sum and therefore can only increase our bound. Thus
	\[
	\|\mm_\alpha^{-1} \vones\|_\infty \leq \frac{1}{1+\alpha} \max_j (\sum_{i=0}^\infty \Big(\frac{\ma}{1+\epsilon}\Big)^i \vones)_j = \frac{1}{1+\alpha} \|\sum_{i=0}^\infty \ma^i \|_\infty = \frac{1}{1+\alpha} \|\mm^{-1}\|_\infty
	\]
	where we again used Lemma \ref{lemma:pos_infnorm}. Combining the pieces we obtain
	\[
	\|\mr_\alpha\|_\infty = \|\vr_\alpha\|_\infty \leq \frac{3}{2} \|\mm_\alpha^{-1} \vones\|_\infty \leq  \frac{3}{2 (1+\alpha)} \|\mm^{-1}\|_\infty.  
	\]
	Essentially the same argument gives $\textbf{ii)}$ as well. 
\end{proof}

\begin{lemma} \label{lemma:scaling_condition}
	Let $\mm = \eye - \ma$ be an M-matrix. Let $\alpha \geq 0$ and $\mm_\alpha = \mm + \alpha \eye$. Let $\vl_\alpha$ and $\vr_\alpha$ be vectors where $(\mm_\alpha^\top \vl_\alpha)_j \in [\frac{1}{2}, \frac{3}{2}]$ and $(\mm_\alpha \vr_\alpha)_j \in [\frac{1}{2}, \frac{3}{2}]$ for all $j$. Define $\ms_\alpha = \ml_\alpha \mm_\alpha \mr_\alpha$. Then
	\begin{enumerate}[label=\textbf{\roman*})]
	\item $\kappa(\mr_\alpha) \leq 3 \|\mm^{-1}\|_\infty$
	\item $\kappa(\ml_\alpha) \leq 3 \|\mm^{-1}\|_1$
	\item $\kappa(\ms_\alpha) \leq 18  \|\mm^{-1}\|_\infty \|\mm^{-1}\|_1$
	\item $\kappa(\ml_0 \mr_0^{-1}) \leq 9 \|\mm^{-1}\|_\infty \|\mm^{-1}\|_1$
	\end{enumerate}
\end{lemma}
\begin{proof}
	We again prove the claims in order. We start with \textbf{i)}. We observe that $\|\mr_\alpha\|_\infty \leq \frac{3}{2(1+\alpha)} \|\mm^{-1}\|_\infty$ by Lemma \ref{lemma:scaling_infnorms} and that $\|\mr_\alpha^{-1}\|_\infty \leq 2(1+\alpha)$ by Lemma \ref{lemma:smallest_scaling}. Combining these two implies the result. A similar proof gives \textbf{ii)} as well.
	
	We now move to \textbf{iii)}. Note that $\kappa(\ms_\alpha) = \|\ms_\alpha\|_2 \|\ms_\alpha^{-1}\|_2.$. We will bound these two norms separately. Starting with $\|\ms_\alpha\|_2$, we observe by Lemma \ref{lemma:MMatrix_RCDD} that $\ms_\alpha$ is an RCDD matrix. Thus, for any $i$ we observe 
	\[
	\sum_j |(\ms_\alpha)_{ij}| = (\ms_\alpha)_{ii} + \sum_{j\neq i} |(\ms_\alpha)_{ij}| \leq 2 (\ms_\alpha)_{ii} =  2 (1 + \alpha) (\vr_\alpha)_i (\vl_\alpha)_i \leq 2 (1+\alpha) \|\mr_\alpha\|_\infty \|\ml_\alpha\|_\infty. 
	\] 
	and
	\[
	\sum_j |(\ms_\alpha)_{ji}| = (\ms_\alpha)_{ii} + \sum_{j\neq i} |(\ms_\alpha)_{ji}| \leq 2 (\ms_\alpha)_{ii} =  2 (1 + \alpha) (\vr_\alpha)_i (\vl_\alpha)_i \leq 2 (1+\alpha) \|\mr_\alpha\|_\infty \|\ml_\alpha\|_\infty. 
	\]
	Thus both $\|\ms_\alpha\|_\infty$ and $\|\ms_\alpha\|_1$ are at most $2 (1+\alpha) \|\mr_\alpha\|_\infty \|\ml_\alpha\|_\infty$. As $\|\mb\|_2^2 \leq \|\mb\|_1 \|\mb\|_\infty$ for any matrix $\mb$, we see that $\|\ms_\alpha\|_2 \leq 2 (1+\alpha) \|\mr_\alpha\|_\infty \|\ml_\alpha\|_\infty$ as well. By applying Lemma \ref{lemma:scaling_infnorms} to this, we obtain
	\[
	\|\ms_\alpha\|_2 \leq \frac{9(1+\alpha) \|\mm^{-1}\|_\infty  \|\mm^{-1}\|_1}{2(1+\alpha)^2}.
	\]
	For $\|\ms_\alpha^{-1}\|_2$, we use a result of \cite{Varah}: if $\beta = \min_k (\ms_\alpha)_{kk} - \sum_{j \neq k} |(\ms_\alpha)_{jk}|$ and $\gamma = \min_k (\ms_\alpha)_{kk} - \sum_{j \neq k} |(\ms_\alpha)_{kj}|$, then $\|\ms_\alpha^{-1}\|_2 \leq \frac{1}{\sqrt{\beta \gamma}}$. As $\ms_\alpha$ has a nonnegative off-diagonal, we see that 
	\[
	\beta = \min_k (\ms_\alpha)_{kk} - \sum_{j \neq k} |(\ms_\alpha)_{jk}| = \min_k (\ms_\alpha \vones)_k \geq \min_i (\vl_\alpha)_i \min_k (\mm_\alpha \vr_\alpha)_k 
	\]
	and likewise
	\[
	\gamma = \min_k (\ms_\alpha^\top \vones)_k \geq \min_i (\vr_\alpha)_i \min_k (\mm_\alpha^\top \vl_\alpha)_k.
	\]
	As  $(\mm_\alpha^\top \vl_\alpha)_j \in [\frac{1}{2}, \frac{3}{2}]$ and $(\mm_\alpha \vr_\alpha)_j \in [\frac{1}{2}, \frac{3}{2}]$ for all $j$, we obtain $\beta \geq \frac{1}{2} \min_i (\vl_\alpha)_i$ and $\gamma \geq \frac{1}{2} \min_i (\vr_\alpha)_i$. We have by Lemma \ref{lemma:smallest_scaling} that both $\min_i (\vr_\alpha)_i$ and $\min_i (\vl_\alpha)_i$ are at least $\frac{1}{2(1 + \alpha)}$. Thus combining these we see that $\|\ms_\alpha^{-1}\|_2 \leq 4 (1+\alpha)$. Putting this together with the previous bound on $\|\ms_\alpha\|_2$ we obtain
	\[
	\kappa(\ms_\alpha) \leq \frac{36(1+\alpha)^2  \|\mm^{-1}\|_\infty  \|\mm^{-1}\|_1}{2 (1+\alpha)^2} = 18 \|\mm^{-1}\|_\infty  \|\mm^{-1}\|_1.
	\]
	Finally, $\textbf{iv)}$ follows from $\kappa(\ml_0 \mr_0^{-1}) \leq \kappa(\ml_0) \kappa(\mr_0)$ and $\textbf{i)}$ and $\textbf{ii)}$.
\end{proof}

\begin{lemma} \label{lemma:mmatrix_to_eigenvector}
	Let $\ma$ be a nonnegative matrix with $\rho(\ma) < 1$, and let $\mm = \eye - \ma$. Let $\epsilon > 0$ be a parameter and define $\mm_\epsilon = \mm + \epsilon \eye$. Let $\eigl$ and $\eigr$ represent the top left and right eigenvectors of $\ma$ respectively. Then $\|\mm_\epsilon^{-1}\|_\infty \leq \frac{1}{\epsilon} \kappa(\eigr)$ and $\|\mm_\epsilon^{-1}\|_1 \leq \frac{1}{\epsilon} \kappa(\eigl)$.
\end{lemma}
\begin{proof}
	By Lemma \ref{lemma:eigvec_scaling}, we know that $\|\mv_r^{-1} \ma \mv_r\|_\infty = \rho(\ma) < 1.$ With this in mind, we define $\mb = \mv_r^{-1} \ma \mv_r$ and write
	\begin{align*}
	\|\mm_\epsilon^{-1}\|_\infty &= \|((1+\epsilon) \eye - \ma)^{-1}\|_\infty = \frac{1}{1+\epsilon} \| \sum_{i=0}^\infty \Big(\frac{\ma}{1+\epsilon}\Big)^i \|_\infty \\
	&=  \frac{1}{1+\epsilon} \| \sum_{i=0}^\infty \Big(\frac{\mv_r \mb \mv_r^{-1}}{1+\epsilon}\Big)^i \|_\infty =  \frac{1}{1+\epsilon} \| \sum_{i=0}^\infty \mv_r \Big(\frac{\mb}{1+\epsilon}\Big)^i \mv_r^{-1} \|_\infty\\
	&\leq \frac{1}{1+\epsilon} \|\mv_r\|_\infty \|\mv_r^{-1}\|_\infty \| \sum_{i=0}^\infty \Big(\frac{\mb}{1+\epsilon}\Big)^i \|_\infty.
	\end{align*}
	Now, $\|\mv_r\|_\infty \|\mv_r^{-1}\|_\infty = \kappa(\eigr)$. Further, we have $\|\mb^i\|_\infty \leq \|\mb\|_\infty^i < 1.$ Thus, by triangle inequality we obtain
	\[
	\|\mm_\epsilon^{-1}\|_\infty \leq \frac{\kappa(\eigr)}{1+\epsilon} \| \sum_{i=0}^\infty \Big(\frac{\mb}{1+\epsilon}\Big)^i \|_\infty \leq \frac{\kappa(\eigr)}{1+\epsilon} \sum_{i=0}^\infty \frac{\|\mb^i\|}{(1+\epsilon)^i} < \frac{1}{1+\epsilon} \kappa(\eigr) \sum_{i=0}^\infty \frac{1}{(1+\epsilon)^i} = \frac{\kappa(\eigr)}{\epsilon}
	\]
	as desired. An equivalent proof gives the result with $\eigl$. 
\end{proof}

\section{Scaling Symmetric M-matrices}
\label{sec:symmscaling}

Here, we show how the M-matrix scaling algorithm and analysis from Section \ref{sec:scaling} can be simplified in the case where our input M-matrix is symmetric. Additionally, we give our algorithm's exact running time assuming access to the best SDD linear system solvers currently available. We first define our oracle for solving symmetric SDD linear systems:

\SDDsolve*

We prove the following result:

\symmscaling*

\begin{algorithm} [t]
    \caption {\sc \fullsymmscale$(A, \epsilon, K)$} \label{alg:symMMatrix}    
\begin{algorithmic}[1]        
            \State \textbf{Input:} $\ma \in \R_{\geq 0}^{n\times n}$ a symmetric matrix with $\rho(\ma)<1$, $\epsilon$ real numbers.
            \State \textbf{Output:} Positive diagonal matrix $\mv$ such that $\mv [(1+\epsilon) \eye - \ma] \mv$ is SDD
        \State $\alpha \eq 1, k \eq 0, \vv_\alpha^{\, \, (0)} \gets \vzero$
        \While{$\|\mm_\alpha \vv_\alpha^{\, \, (k)} - \vones \|_\infty > \frac{1}{2}$}
            \State $\vv_{\alpha}^{\, \,(k+1)} \eq \vv_{\alpha}^{\, \, (k)} -  \frac{1}{4} [\mm_{\alpha} \vv_{\alpha}^{\, \, (k)} - \vones]$ 
            \State $k \gets k + 1$  
        \EndWhile
        \State $\vv_\alpha \gets \vv_\alpha^{\, \, (k)}$
        \While{$\alpha > \epsilon$}
            \State $\alpha \eq \alpha/2, \quad  k \gets 0$
            \State $\vv_{\alpha}^{\, \, (0)} \gets \vzero$
            \State $\mm_\alpha = (1+\alpha) \eye - \ma$.
            \While{$\|\mm_{\alpha} \vv_{\alpha}^{\, \, (k)} - \vones\|_\infty > \frac{1}{2}$}
                \State $\opZ_{2 \alpha}  \gets  \code{SDDSolve}(\mv_{2\alpha}\mm_{2\alpha}\mv_{2\alpha}, \frac{1}{4})$
                \State $\vv_{\alpha}^{\, \,(k+1)} \eq \vv_{\alpha}^{\, \, (k)} - \mv_{2\alpha} \opZ_{2 \alpha} (\mv_{2\alpha} [\mm_{\alpha} \vv_{\alpha}^{\, \, (k)} - \vones])$ 
                \State $k \gets k + 1$
            \EndWhile
            \State $\vv_\alpha \eq \vv_\alpha^{\, \, (k)}$
        \EndWhile
        
    \State \Return $\mv_\alpha$ with $\mv_\alpha \mm_{\epsilon} \mv_\alpha$ being SDD 
\end{algorithmic}	
\end{algorithm}

Our proof of correctness follows much the same blueprint as the proof of the asymmetric scaling algorithm in Section \ref{sec:scaling}. As before, we begin by making the simplifying assumption that $s=1$ for our M-matrix $\mm$: this proceeds without loss of generality as before. We first show that our iteration to compute the initial scaling $\vv_1$ converges quickly, and then show that under our assumption on $\code{SDDSolve}$'s properties the iteration to compute $\vv_\alpha$ from $\vv_{2\alpha}$ also converges quickly. We then combine these facts into our claim.

\begin{lemma}
    \label{symm:initscale}
Let $\ma$ be a symmetric nonnegative matrix with $\rho(\ma) < 1$, and let $\mm_1 = 2 \eye - \ma$. Let $\vx$ be any vector, and define $\vy = \vx - \frac{1}{4} [\mm_1 \vx - \vones]$. Then 
\[
\| \mm_1 \vy - \vones \|_2 \leq \frac{3}{4} \| \mm_1 \vx - \vones \|_2.    
\]
\end{lemma}
\begin{proof}
We observe that 
\begin{align*}
    \|\mm_1 \vy - \vones\|_2 &= \|\mm_1 \vx - \frac{1}{4} \mm_1[\mm_1 \vx - \vones] - \vones \|_2 \\
                            &=\|[\eye - \frac{1}{4} \mm_1][\mm_1 \vx - \vones] \|_2 \\
                            &\leq \|\eye - \frac{1}{4} \mm_1\|_2 \|\mm_1 \vx - \vones \|_2
\end{align*}
Note that since $\ma$ is symmetric and satisfies $\rho(\ma) < 1$, we have $-\eye \preceq \ma \preceq \eye$. Thus $\eye \preceq \mm_1 \preceq 3 \eye$ and so $\frac{1}{4} \eye \preceq \eye - \frac{1}{4} \mm_1 \preceq \frac{3}{4} \eye$. Therefore as $\ma$ is symmetric we can conclude $\| \eye - \frac{1}{4} \mm_1 \|_2 \leq \frac{3}{4}$. Substituting this into the above implies the claim.
\end{proof}

The above fact shows that the iteration to compute the initial $\vv_1$ reduces a residual term by a constant factor in each step. We will now show that the iteration to compute $\vv_\alpha$ converges quickly. Before we do this, we will need to establish that $\mv_{2\alpha}\mm_{2\alpha}\mv_{2\alpha}$ is SDD:

\begin{lemma}
\label{symm:SDD}
Let $\mm$ be a symmetric M-matrix, and let $\vv$ be any vector where $||\mm \vv - \vones ||_\infty \leq \frac{1}{2}$. Then $\mv \mm \mv$ is SDD.
\end{lemma}
\begin{proof}
We observe first that for any index $i$ $(\mm \vv)_i \in [\frac{1}{2}, \frac{3}{2}]$. Thus by Lemmas \ref{lemma:MMatrix_monotone} and \ref{lemma:MMatrix_RCDD} and the fact that $\mm$ is symmetric we conclude that $\mv \mm \mv$ is SDD. 
\end{proof}

We now show our iteration for computing $\vv_\alpha$ from $\vv_{2\alpha}$ converges quickly. 

\begin{lemma}
    \label{symm:inductivescale}
    Let $\ma$ be a symmetric nonnegative matrix with $\rho(\ma) < 1$. Let $\alpha > 0$ be some constant, and let $\mm_{\alpha} = (1+\alpha) \eye - \ma$. Let $\vv_{2\alpha}$ be a positive vector such that $\ms = \mv_{2\alpha} \mm_{2\alpha}\mv_{2\alpha}$ is SDD. Define $\vw = \mm_{\alpha}^{-1} \vones$. Let $\vx$ be a vector, and let $\vy = \mv_{2\alpha} \mm_{\alpha} [\vx - \vw]$. Let $\opZ$ be some linear operator satisfying 
    \[
    \| \ms^{-1} \vy - \opZ_{2\alpha} (\vy)\|_\ms  \leq \frac{1}{4} \| \ms^{-1} \vy \|_\ms.     
    \]
     Then if $\vz = \vx - \mv_{2\alpha} \opZ_{2 \alpha}(\vy)$ then
    \[
    \|\vz - \vw \|_{\mm_{2\alpha}} \leq \frac{3}{4} \|\vx - \vw\|_{\mm_{2\alpha}}.
    \]    
\end{lemma}
\begin{proof}
We observe 
\begin{align*}
    \vz - \vw &= \vx - \mv_{2\alpha} \opZ_{2 \alpha}(\vy) - \vw \\
    &= \vx -  \mm_{2\alpha}^{-1} \mm_{\alpha} [\vx - \vw] - \vw  +  \mm_{2\alpha}^{-1}  \mm_{\alpha} [\vx - \vw] -\mv_{2\alpha} \opZ_{2 \alpha}(\vy) \\
    &= \alpha \mm_{2\alpha}^{-1} [\vx - \vw] + \mv_{2\alpha} \ms^{-1} \vy - \mv_{2\alpha} \opZ_{2 \alpha}(\vy)
\end{align*}
where we used the fact that $\mm_\alpha = \mm_{2\alpha} - \alpha \eye$. Taking norms of both sides and applying the triangle inequality yields
\begin{align*}
    \|\vz - \vw \|_{\mm_{2\alpha}} &\leq \alpha \|\mm_{2\alpha}^{-1} [\vx - \vw] \|_{\mm_{2\alpha}} + \| \mv_{2\alpha} \ms^{-1} \vy - \mv_{2\alpha} \opZ_{2 \alpha}(\vy) \|_{\mm_{2\alpha}}.
\end{align*}
We split the analysis of this expression in two parts. First, we observe that 
\[
\alpha \|\mm_{2\alpha}^{-1} [\vx - \vw] \|_{\mm_{2\alpha}} \leq \alpha \| \mm_{2\alpha}^{-1} \|_{\mm_{2\alpha}} \| \vx - \vw\|_{\mm_{2\alpha}} =   \alpha \| \mm_{2\alpha}^{-1} \|_2 \| \vx - \vw\|_{\mm_{2\alpha}}.
\]
Now, $\mm_{2\alpha}^{-1} \preceq \frac{1}{2\alpha} \eye$: since $\mm_{2\alpha}$ is symmetric this implies 
\[
\alpha \|\mm_{2\alpha}^{-1} [\vx - \vw]\|_{\mm_{2\alpha}} \leq \frac{1}{2} \| \vx - \vw\|_{\mm_{2\alpha}}.
\]
For the second part, we have by our assumption on $\opZ_{2\alpha}$
\[
\| \mv_{2\alpha} \ms^{-1} \vy - \mv_{2\alpha} \opZ_{2 \alpha}(\vy) \|_{\mm_{2\alpha}} = \|\ms^{-1} \vy - \opZ_{2 \alpha}(\vy) \|_{\ms} \leq \frac{1}{4} \|\ms^{-1} \vy\|_\ms = \frac{1}{4} \|\mv_{2\alpha} \ms^{-1} \mv_{2\alpha} \mm_{\alpha} [\vx - \vw] \|_{\mm_{2\alpha}}.
\]
Simplifying this expression yields 
\[
\| \mv_{2\alpha} \ms^{-1} \vy - \mv_{2\alpha} \opZ_{2 \alpha}(\vy) \|_{\mm_{2\alpha}} \leq \frac{1}{4} \|\mm_{2\alpha}^{-1} \mm_{\alpha} [\vx - \vw] \|_{\mm_{2\alpha}} \leq  \frac{1}{4} \|\mm_{2\alpha}^{-1} \mm_{\alpha} \|_{\mm_{2\alpha}} \| \vx - \vw \|_{\mm_{2\alpha}}.
\]
Observe that $\|\mm_{2\alpha}^{-1} \mm_{\alpha} \|_{\mm_{2\alpha}} = \|\mm_{2\alpha}^{-1/2} \mm_{\alpha} \mm_{2\alpha}^{-1/2}\|_2.$ As $\mm_{\alpha}$ and $\mm_{2\alpha}$ are symmetric and $\mm_\alpha \preceq \mm_{2\alpha}$, we note that $\|\mm_{2\alpha}^{-1/2} \mm_{\alpha} \mm_{2\alpha}^{-1/2}\|_2 \leq \|\mm_{2\alpha}^{-1/2} \mm_{2\alpha} \mm_{2\alpha}^{-1/2}\|_2 = 1$. Substituting into the above expression and combining yields
\[
\| \mv_{2\alpha} \ms^{-1} \vy - \mv_{2\alpha} \opZ_{2 \alpha}(\vy) \|_{\mm_{2\alpha}}  \leq  \frac{1}{4} \| \vx - \vw \|_{\mm_{2\alpha}}.
\]
Combining this with the previous bound gives the claim.
\end{proof}

We now prove our theorem:
\begin{proof}
To begin, we show the correctness of our procedure. If we assume that at every value of $\alpha$ encountered we obtain a $\vv_\alpha$ satisfying $\|\mm_\alpha \vv_\alpha - \vones \|_\infty \leq \frac{1}{2}$, then by Lemma \ref{symm:SDD} we must conclude that $\vv_\alpha$ can be used to construct an SDD scaling for $\mm_\alpha$. Thus our output is clearly a valid scaling for $\mm_\epsilon$. 

We now prove that our algorithm runs in the claimed time. We begin by analyzing the time required to compute the initial scaling $\vv_1$. By Lemma \ref{symm:initscale} and line 5 of the algorithm we observe that for any $k$ we have
\[
\| \mm_1 \vv_1^{\, \, (k+1)} - \vones \|_2 \leq \frac{3}{4} \| \mm_1 \vv_1^{\, \, (k)} - \vones \|_2.    
\]
This implies
\[
\| \mm_1 \vv_1^{\, \, (k)} - \vones \|_\infty \leq \| \mm_1 \vv_1^{\, \, (k)} - \vones \|_2 \leq \Big(\frac{3}{4}\Big)^k \|\vones \|_2 = (\frac{3}{4})^k \sqrt{n}.
\]
Thus after $O(\log n)$ iterations of line 5 we will find some $\vv_1$ which is a valid scaling. Each iteration can be implemented in $O(m)$ time for $O(m \log n)$ total work.

We now analyze the time required to compute $\vv_\alpha$ from $\vv_{2\alpha}$. Assume that all of our calls to \SDDsolve on lines 14 and 15 of the algorithm return a valid approximate solution. We will show that under this assumption we will perform only $\otilde(1)$ calls to \SDDsolve: since each individual call returns a valid estimator with high probability our algorithm itself succeeds with high probability. We note by Lemma \ref{symm:inductivescale} and lines 14 and 15 of the algorithm that for any $k$
\[
\|\vv_\alpha^{\, \, (k+1)} - \mm_\alpha^{-1} \vones \|_{\mm_{2\alpha}} \leq \frac{3}{4} \|\vv_\alpha^{\, \, (k)} - \mm_\alpha^{-1} \vones \|_{\mm_{2\alpha}}.
\] 
Like before this implies
\[
\|\mm_\alpha \vv_\alpha^{\, \, (k)} -  \vones \|_{\mm_\alpha^{-1} \mm_{2\alpha} \mm_\alpha^{-1}} = \|\vv_\alpha^{\, \, (k)} - \mm_\alpha^{-1} \vones \|_{\mm_{2\alpha}} \leq \Big(\frac{3}{4}\Big)^k \|\mm_\alpha^{-1} \vones \|_{\mm_{2\alpha}} = \Big(\frac{3}{4}\Big)^k \| \vones \|_{\mm_\alpha^{-1} \mm_{2\alpha} \mm_\alpha^{-1}}.
\]
Now note that $\mm_\alpha^{-1} \preceq \mm_\alpha^{-1} \mm_{2\alpha} \mm_\alpha^{-1} \preceq 2 \mm_\alpha^{-1}$. With this, it is clear that $||\mm_\alpha^{-1} \mm_{2\alpha} \mm_\alpha^{-1}||_2 \leq 2 ||\mm_\alpha^{-1}||_2$ and $||\mm_\alpha \mm_{2\alpha}^{-1} \mm_\alpha||_2 \leq ||\mm_\alpha||_2.$ Therefore $\kappa(\mm_\alpha^{-1} \mm_{2\alpha} \mm_\alpha^{-1}) \leq 2\kappa(\mm_\alpha) \leq 2 \kappa(\mm)$. Applying this fact in the above yields
\[
\|\mm_\alpha \vv_\alpha^{\, \, (k)} -  \vones \|_\infty \leq \|\mm_\alpha \vv_\alpha^{\, \, (k)} -  \vones \|_2 \leq 2 \Big(\frac{3}{4}\Big)^k \kappa(\mm) \| \vones \|_2 = 2\sqrt{n} \Big(\frac{3}{4}\Big)^k \kappa(\mm). 
\]
Thus after $k = O(\log(n \kappa(\mm)))$ iterations we will have some $\vv_\alpha$ which gives a valid scaling for $\mm_\alpha$. Each iteration in this process to compute $\vv_\alpha$ requires one matrix-vector product with $\mm_\alpha$ (requiring $O(m)$ time), one query to $\opZ$ which takes $\mathcal{T}_{solve}\Big(m, n, \frac{1}{4} \Big)$ time, and $O(n)$ additional work. Thus the iteration to compute $\vv_\alpha$ from $\vv_{2\alpha}$ requires only $\mathcal{T}_{solve}\Big(m, n, \frac{1}{4} \Big) \log(n \kappa(\mm))$ time. To conclude, we observe that in only $\log(1/\epsilon)$ iterations $\alpha$ falls from $1$ to $\epsilon$: thus we need only construct $\vv_\alpha$ from $\vv_{2\alpha}$ that many times before we obtain a scaling for $\mm_\epsilon$. The claimed running time follows.  
\end{proof}
If we equip this algorithm with the current fastest SDD system solvers (from \cite{coh14}), we obtain the following corollary:
\begin{corollary}
\label{symm:fastsolve}
Let $\ma \in \R^{\nn}$ be a symmetric nonnegative matrix with $m$ nonzero entries and $\rho(\ma) < 1$, and let $\mm = \eye - \ma$. Pick $\epsilon > 0$. Then we can compute a diagonal matrix $\mv$ where $\mv((1+\epsilon)\eye - \ma)\mv$ is SDD with high probability in time 
\[
O\Big(m \log(n \kappa(\mm)) \log^{1/2}(n) \mathrm{poly}(\log \log n) \log\Big(\frac{1}{\epsilon} \Big) \Big).
\]
\end{corollary}
Finally, we obtain the following corollary for solving linear systems in symmetric M-matrices:
\begin{corollary}
Let $\ma \in \R^{\nn}$ be a symmetric nonnegative matrix with $m$ nonzero entries and $\rho(\ma) < 1$, and let $\mm = \eye - \ma$. Then for any vector $\vb$ we can find $\vx$ such that $||\mm \vx - \vb||_2 \leq \epsilon ||\vb||_2$ in time
\[
O\Big(m \log(n \kappa(\mm)) \log(\kappa(\mm)) \log^{1/2}(n) \mathrm{poly}(\log \log n) \log\Big(\frac{1}{\delta} \Big) \Big).
\]
\end{corollary}
\begin{proof}
Note that if $\epsilon = \lambda_{min}(\mm)$ $\mm \preceq \mm_\epsilon \preceq 2\mm$. Thus, by the standard analysis of preconditioned Richardson iteration $O(\log{1}{\delta})$ applications of a linear system solver for $\mm_\epsilon$ yields a linear system solver for $\mm$. By simply running this procedure for each $\mm_\alpha$ encountered in \fullsymmscale and checking if the resulting vector $\vx_\alpha$ satsifes $||\mm \vx_\alpha - \vb||_2 \leq \delta ||\vb||_2$, we will eventually find a vector satisfying this and we will do so within $\log(1/\epsilon) = \log(\kappa(\mm))$ iterations. The running time follows.
\end{proof}
\section{Factor Width 2 Matrices}
\label{sec:factor2}

Here we discuss how to use an algorithm for solving linear systems in symmetric M-matrices to solve linear systems in factor width 2 matrices. The idea behind this reduction is simple. We will show that if $\mm$ is a factor width 2 matrix and $\mm'$ is a matrix where the sign of every off-diagonal entry is made negative, i.e. $\mm'_{ij} = - |\mm_{ij}|$ for all $i \neq j$ and $\mm'_{ii} = \mm_{ii}$, then $\mm'$ is an M-matrix. Consequently, as we saw in Section~\ref{sec:symmscaling}, this implies that there is a diagonal matrix, $\mv$, with positive diagonal entries such that $\mv \mm' \mv$ is symmetric diagonally dominant. This immediately implies that $\mv \mm \mv$ is symmetric diagonally dominant as well. Consequently, we can solve linear systems in $\mm$ simply by computing $\mv$ as in Section~\ref{sec:symmscaling} and then solving $\mv \mm' \mv$ using a SDD solver (note that we can apply the same trick for asymmetric matrices and RCDD solvers). 

In the remainder of this section we formally show that $\mm'$ is an M-matrix. We break this into two lemmas about M-matrices and factor width two matrices. This result and the lemma we use to prove it are fairly well known and we encourage the reader to view \cite{DaitchS09} for further context. Here we provide a short self-contained proof for completeness.

\begin{lemma}[Characterization of M-Matrices]
\label{lem:m-char} All symmetric PSD Z-matrices are M-matrices.
\end{lemma}

\begin{proof}
Let $\mm$ be an arbitrary symmetric PSD Z-matrix. Since $\mm$ is a symmetric Z-matrix we can write it as $\mm = \md - \mm'$ where $\md$ is a  diagonal matrix with $\md_{ii} = \mm_{ii}$ and $\mm'$ is a non-negative diagonal matrix. Since $\mm$ is PSD we know that $\md$ is non-negative.

Let $s = \max_{i} \md_{ii}$ and $\ma = s \mI - \md + \mm'$. Clearly, $\mm = s \mI - \ma$ and $\ma$ is symmetric and entry-wise non-negative by our choice of $s$. Thus to show that $\mm$ is an M-matrix it simply remains to show that $\rho(\ma) \leq s$. However, for any vector $\vx$ if $\vy$ is a vector with $y_i = |x_i|$ non-negativity of $\mm$ implies
\[
\frac{| \vx^\top \ma \vx |}{\vx^\top \vx}
\leq \frac{\vy ^\top \ma \vy}{\vx^\top \vx}
= \frac{\vy^\top \ma \vy}{\vy^\top \vy} 
\]
and consequently,
\begin{align*}
\rho(\ma) &= \max_{\vx \neq 0} \frac{| \vx^\top  \ma \vx|}{\vx^\top \vx}
\leq \max_{\vx \neq 0} \frac{\vx^\top  \ma \vx}{\vx^\top \vx}
= \max_{\vx \neq 0} \frac{\vx^\top (s \mI - \md + \mm') \vx}{\vx^\top \vx}
= \max_{\vx \neq 0} \left[ s - 
\frac{\vx^\top \mm \vx}{\vx^\top \vx}
\right]
\leq s
\end{align*}
where in the last step we used that $\mm$ is PSD by assumption.
\end{proof}

\begin{lemma}[Factor Width 2 Matrices Induce M-Matrices]
If $\mm$ is a factor width 2 matrix and $\mm'$ satisfies $\mm'_{ij} = - |\mm_{ij}|$ for all $i \neq j$ and $\mm'_{ii} = \mm_{ii}$ for all $i$ then $\mm'$ is an M-matrix.
\end{lemma}

\begin{proof}
Since $\mm$ is factor width 2 we have that $\mm = \mb^\top \mb$ where each row of $\mb$ has at most 2 nonzero entries. Now, suppose we let $\mb'$ be the same matrix as $\mb$ except with the signs of the entries changed so that no row contains 2 positive or 2 negative entries. Note that $\mm'' = (\mb')^\top (\mb')$ has the same diagonal entries as $\mm$ however we have that for all $i \neq j$ it is the case that $[\mm'']_{ij} \leq - |\mm_{ij}|$. Consequently, $\mm''$ is a symmetric Z-matrix and clearly PSD (since we explicitly have its factorization). By  Lemma~\ref{lem:m-char} this implies that $\mm''$ is an M-matrix. Furthermore, since its off-diagonal entries are larger in magnitude then those of $\mm'$ and its diagonal entries are the same from the definition of an M-matrix this implies that $\mm'$ is an M-matrix as desired. This follows from the fact that if $\mn \geq \ma \geq 0$ entrywise then $\rho(\mn) \geq \rho(\ma)$.
\end{proof}

\end{document}